\def\notes{0}
\def\neurips{0}

\ifnum\neurips=0

\documentclass[11pt]{article}
\usepackage{graphicx,macros}
\usepackage{mathtools}
\usepackage[margin=1in]{geometry}
\usepackage{multirow}
\usepackage{array}
\usepackage{booktabs}

\title{Privately Estimating Monotone Statistics in Polynomial Time}
\author{
Gavin Brown\thanks{University of Wisconsin--Madison, \texttt{gavin.brown@wisc.edu}} \
\and Ephraim Linder\thanks{Boston University, \texttt{ejlinder@bu.edu}. Supported in part by NSF grant BCS-2218803 and a grant from the Sloan Foundation.}  
\and Mahbod Majid\thanks{MIT Mathematics, \texttt{mahbod@mit.edu}} 
\and Vikrant Singhal\thanks{University of Copenhagen, \texttt{vikrant.singhal@di.ku.dk}. Supported in part by NSF grant BCS-2218803, the Digital Data Design Trustworthy AI Lab at Harvard, a grant from the Sloan Foundation, and the Novo Nordisk Foundation grant NNF24OC0087820.} 
}
\date{}

\fi

\ifnum\neurips=1
\documentclass{article}
\usepackage{neurips/neurips_2026, neurips/neurips_macros}
\usepackage{mathtools}
\usepackage{times}
\usepackage[utf8]{inputenc} %
\usepackage[T1]{fontenc}    %
\usepackage{hyperref}       %
\usepackage{url}            %
\usepackage{booktabs}       %
\usepackage{amsfonts}       %
\usepackage{nicefrac}       %
\usepackage{microtype}      %
\usepackage{xcolor}         %

\usepackage{multirow}
\usepackage{array}
\title{Privately Estimating Monotone Statistics in Polynomial Time}

\author{%
 \Name{Gavin Brown} \Email{gavin.brown@wisc.edu}\\
 \addr University of Wisconsin--Madison.
 \AND
 \Name{Ephraim Linder} \Email{xyz@sample.com}\\
 \addr Boston University. 
 \AND
 \Name{Mahbod Majid} \Email{mahbod@mit.edu}\\
 \addr MIT Mathematics.
 \AND
 \Name{Vikrant Singhal} \Email{vikrant.singhal@di.ku.dk}\\
 \addr Datalogisk Institut, Universitetsparken 1, 2100 K{\o}benhavn, Denmark
}
\fi

\begin{document}

\maketitle

\begin{abstract}
    We study efficient differentially private algorithms for  estimating monotone statistics, i.e., statistics that are monotone under the addition of new observations. The starting point for our investigation is \emph{subsample-and-aggregate}: a classical paradigm that partitions the dataset into blocks, estimates the statistic on each block, and then privately aggregates the estimates.
    While practical and generically applicable, this approach is quite data-hungry. We improve upon this framework for the class of monotone statistics---compared to subsample-and-aggregate, our algorithms save a factor of $t$ in sample complexity and pay a factor of $e^t$ in running time, where $t>0$ is a tunable parameter.
    We complement our results with a query-complexity lower bound, showing that our algorithms are essentially optimal for this task.
    As an application, we obtain improved results for private eigenvalue estimation, private loss estimation, and privately estimating a single parameter of a high-dimensional model, e.g., in linear regression.
\end{abstract}

\ifnum\neurips=0
\newpage
\tableofcontents
\newpage
\fi

\ifnum\notes=1
\begin{center}
    {\color{red} \Large{Notes are on!}}
\end{center}

\fi

\section{Introduction}

In this work, we consider the task of generically producing differentially private (DP)~\cite{DworkMNS06} algorithms from non-private estimators.
The literature on differentially private statistics supplies a toolkit for a variety of fundamental tasks such as mean estimation, linear regression, and principal component analysis.
For all of these tasks and more, the literature provides efficient algorithms that achieve near-optimal accuracy.
These algorithms, however, require a bespoke design and considerable care in their analysis.
Even the general-purpose statistical tools built on the inverse sensitivity mechanism seem to require considerable technical effort in their application \cite{AsiD20,HopkinsKMN23,AsiUZ2023}.
Building on prior work with the same goal, we develop differentially private algorithms that treat the base estimator as a black box.

Specifically, we start with any function $f$ that takes in a dataset and returns a (univariate) quantity. 
As a motivating example, consider estimating the first parameter of a larger regression problem.
We seek to produce differentially private algorithms that, given a dataset $Z$ and query access\footnote{The algorithm can specify a query in the form of a dataset $X$ and receives response $f(X)$. The \emph{query complexity} of an algorithm is the maximum number of queries it makes when given query access to any function $f$.} to $f$, return an output with accuracy comparable to the non-private estimator.
~
Since $f$ is an arbitrary function, the measure of accuracy is context dependent. 
To begin our informal discussion, we suppose there is a benchmark number of samples $N$ under which we expect an ``adequate'' non-private estimator. 
For example, for least squares in $d$ dimensions we require $N \approx d$ for non-trivial estimation.
We will evaluate the private estimator in terms of the number of samples $n$ it needs in order to compete with the non-private estimator on $N$ samples. 

One classic approach to this problem is \emph{subsample-and-aggregate} (S\&A)~\cite{NissimRS07}.
Algorithms in this framework proceed as follows: split the data into $\tau$ disjoint buckets $B_1,\dots, B_\tau$, and use a private aggregation method to release an approximation to some center (such as the average or median) of the $f(B_i)$'s. 
This approach is computationally efficient, as we make $\tau$ queries to $f$, but suffers a blowup in sample complexity---in order for any one $f(B_i)$ to be meaningful, we must have $|B_i|\gtrsim N$. 
Many standard settings require $\tau\gtrsim \frac 1\eps\log\frac1\delta$ buckets, and thus the private sample complexity is at least $N \cdot \frac{\log1/\delta}{\eps}$, a multiplicative factor larger than the non-private counterpart. 

Recent work of \cite{LinderRSS25} shows that one can do considerably better: They construct an $(\eps,\delta)$-differentially private mechanism requiring only $n\approx  N+\frac{1}{\eps}\log\frac1\delta$ samples---that is, an additive overhead in sample complexity. Unfortunately, their mechanism requires exponential time---it makes $n^{\Theta\paren{\frac1\eps\log\frac 1\delta}}$ queries to $f$. %
Shortly thereafter, \cite{SteinkeS25} demonstrated a tradeoff between sample complexity and query complexity: compared to subsample-and-aggregate, one can save a factor of $\trade$ in sample complexity at the cost of a blowup of %
$e^{O(\trade)}$ in query complexity\footnote{The main result of \cite{SteinkeS25} provides %
a general tradeoff in terms of the size of a certain combinatorial object, which, upon close inspection, yields the bound $e^{O(\trade)}$ (ignoring log factors). For further discussion see \Cref{sec:ss25}}.  
However, despite making fewer queries to $f$, the mechanism of \cite{SteinkeS25} nevertheless has runtime $n^{O\paren{\frac1\eps\log\frac 1\delta}}$. 

In this work, we provide an efficient algorithm that improves on the sample complexity of subsample-and-aggregate for the class of \emph{monotone} functions, i.e., functions whose output does not decrease when a new observation is added to the input. In the context of privacy, monotone functions were first studied by \cite{FangDY22} who design an algorithm to privately evaluate a monotone function $f$ with finite range $\kappa$. In the language above, their algorithm outputs an accurate estimate with probability at least $1-\beta$ and achieves sample complexity $n\approx N + \frac{1}{\eps}\log\frac\kappa\beta$ and runtime $n^{\Theta\paren{\frac1\eps\log\frac\kappa\beta}}$. Moreover, the tools they develop for the monotone function setting play a key role in the development of algorithms for general (i.e., non-monotonic) functions in \cite{LinderRSS25,SteinkeS25}. We revisit the setting of monotone functions with an eye towards computational efficiency. In particular, we show that for all $\trade\leq \frac{\log1/\delta}{2\eps}$ one can achieve sample complexity $\frac N \trade \cdot \frac{\log1/\delta}{\eps}$ while only incurring a blowup of $e^{O(\trade)}$ in \emph{both} query complexity and runtime. In \Cref{tab:black-box-dp-comparison_v2}, we highlight several settings of $\trade$ that demonstrate our improvement over the aforementioned baselines.

\newcolumntype{C}[1]{>{\centering\arraybackslash}m{#1}}

\newcommand{\startup}{\ensuremath{k_{\eps,\delta}}}

\begin{table}[H]
\caption{Comparison of our approach for monotone functions with black-box $(\eps,\delta)$-DP estimators. 
Here $\startup := \frac {1}{\eps} \log \frac 1 \delta$ and
$N$ denotes the non-private sample benchmark. We suppress absolute constants and assume that evaluating the underlying function takes unit time, and we compare with \cite{FangDY22,SteinkeS25} in the text.}
\label{tab:black-box-dp-comparison_v2}
\centering
\begin{tabular}{C{3cm} C{2.1cm} C{2.3cm} C{3.1cm}}
\toprule
\textbf{Approach} 
& \textbf{Setting} 
& \textbf{Samples $n$} 
& \textbf{Time} \\
\midrule
S\&A \cite{NissimRS07}
& -
& $N \cdot\startup$
& $\startup$ \\
\midrule
\cite{LinderRSS25}
& -
& $N + \startup$
& $n^{\startup}$\\
\midrule
\multirow{3}{3cm}{\centering{This work \\ (monotone functions)}}
& any $\trade\le \startup$
& $N\cdot \frac{\startup}{\trade} $
& $e^{\trade}\cdot \mathrm{poly}(\startup)$ \\
\cmidrule{2-4}
& $\trade = \startup$ 
& $N + \startup$
& $(1/\delta)^{1/\eps}\cdot \mathrm{poly}(\startup)$ \\
\cmidrule{2-4}
& $\trade = \log(\startup)$ 
& $N \cdot\frac{\startup}{\log(\startup)}$
& $\mathrm{poly}(\startup)$ \\
\bottomrule
\end{tabular}
\end{table}

\ifnum\neurips=1
\paragraph{Organization}
We next informally present our main results: new algorithms for monotone functions, a lower bound on query complexity, and applications to learning problems. 
We then discuss additional related work. In \Cref{sec:median_of_quantiles}, we present and analyzes our median-of-quantiles mechanism, which illustrates our central algorithmic advancements. The rest of our contributions, including the average-of-quantiles mechanism---our sister mechanism of median-of-quantiles, which obtains a different flavor of guarantee---are deferred to the appendix.
\fi

\subsection{Our results}
\label{sec:results_intro}

We consider functions and distributions that satisfy the following assumption, which allows us to meaningfully compare our results to S\&A, and briefly defer our more general results.

\begin{assumption}
    \label{ass:intro_concentration}
    For function $f$ and distribution $\cD$ there exist $N:\R\to\R$ and $\nu\in\R$ such that for all $n\geq N(\alpha) + \frac{\log1/\beta}{\alpha^2}$, we have
    \[
    \Pr_{Z\sim\cD^n}\brackets{\abs{\frac{f(Z)}{n} - \nu}\geq \alpha}\leq \beta,
    \]
    where $\cD^n$ denotes the distribution over $n$ i.i.d.\ samples from $\cD$.
\end{assumption}

As an example, if $f(Z) = \sum z_{i}$ and $\cD$ is subgaussian, then we can take $\nu = \Ex_{z\sim\cD}[z]$. While the exact dependence on lower order terms varies depending on the private aggregation method used, for functions $f$ and distributions $\cD$ that satisfy \Cref{ass:intro_concentration}, S\&A requires at least
\[
n=\Omega\paren{\frac{N(\alpha)\log1/\delta}{\eps}}
\]
samples to estimate $\nu$ up error $\alpha$. We note that if the function has finite range $[\kappa]$ for some $\kappa\in\N$, then $1/\delta$ can be replaced with $\kappa/\beta$. Our results focus on improving the dependence on the privacy parameters in this leading term of the sample complexity bound.

\subsubsection{Privately evaluating monotone statistics}

Our general results apply to functions that do not decrease when new data is added to input.
While our main results provide accuracy guarantees in terms of the value of $f$ on random subsets of the input dataset, we now state a simplified version that applies to functions that satisfy \Cref{ass:intro_concentration}. 
We assume a computational model where queries to $f$ take unit time. 
For a function $f$, a dataset $Z$, and an algorithm $\cA$, we let $\cA^f$ denote algorithm $\cA$ with query access to $f$, and we let random variable $\cA^f(Z)$ denote the output of algorithm $\cA^f$ on input $Z$.

In this setting, our algorithm \emph{average-of-quantiles} gives the following guarantees.
\begin{theorem}[Informal Corollary of \Cref{thm:avg-of-q}]
    \label{thm:avg_of_q-intro}
    Fix privacy parameters $\eps,\delta>0$ and $\trade\leq \frac1\eps\log\frac1\delta$. There exists a mechanism $\cM$ such that for all monotone, real-valued functions $f$ mechanism $\cM^f$ is $(\eps,\delta)$-DP and has query complexity and runtime $e^{O\paren{\trade}}\poly\paren{\frac{\log1/\delta}{\eps}}$. Additionally, if $f$ and $\cD$ satisfy \Cref{ass:intro_concentration}, then for all $\alpha,\beta>0$ and  
    \[
    n = \Omega\paren{\frac{N(\alpha)\log1/\delta}{\trade\eps} + \frac{\log1/\delta \log1/\beta}{\trade\alpha^2\eps} + \frac{\log1/\delta}{\alpha^2\eps}}
    \]
    we have
    \[
    \Pr_{Z\sim\cD^n}\brackets{\abs{\cM^f(Z,\alpha) - \nu}\geq \alpha}\leq \beta.
    \]
\end{theorem}
Average-of-quantiles requires some prior information about the concentration of the non-private estimator (in the above setting, it requires knowledge of $\alpha$).
We analyze a second algorithm, \emph{median-of-quantiles}, which does not use such information but requires functions that have a bounded range. 
We state its slightly different guarantees in \cref{thm:med-of-q}.

We also prove a query complexity lower bound, which implies that our results are essentially tight for privately estimating black-box monotone statistics. The version we state below applies to functions with unbounded range, however the full version (\Cref{thm:query_lb}) also applies to the setting where $f$ has finite range, and it shows that our results are essentially tight for that setting as well. In fact, our lower bound holds for mechanisms that satisfy a much weaker accuracy guarantee: informally, if for all $n\geq N$, we have $f(Z)=\nu$ with probability $1$ when $Z\sim\cD^n$, then the mechanism should output $y\approx \nu$ with probability at least $2/3$ when given $n=\Omega\paren{\frac{N\log1/\delta}{\trade\eps}}$ samples from $\cD$. We define this weaker guarantee more formally in \Cref{sec:query_lb}, and present an informal version of the lower bound below.

\begin{theorem}[Informal Corollary of \Cref{thm:query_lb}]
    Fix $\eps,\delta\in(0,1)$ and $\trade\leq \frac1\eps\log\frac1\delta$. Suppose $\cM^f$ is $(\eps,\delta)$-DP for all monotone functions $f$. Additionally, suppose that for all monotone $f$ and distributions $\cD$ that are ``eventually constant'', %
    mechanism $\cM$ has the following accuracy guarantee:
    \[
    ~~\text{If}~~ n= \Omega\paren{\frac{N(1)\log1/\delta}{\trade\eps}} ~~\text{then}~~
   \Pr_{Z\sim\cD^{n}}\brackets{\cM^f(Z) \approx \nu}\geq \frac 23.
    \]
    Then $\cM$ has query complexity $e^{\Omega\paren{\trade}}$.      
\end{theorem}

\subsubsection{Application to eigenvalue estimation}

We apply our tools towards the problem of privately estimating the $\ord{i}$ eigenvalue $\lambda_{i}(\Sigma)$ of the covariance $\Sigma$ of a subgaussian distribution $\cD$. Our mechanism witnesses a tradeoff between sample and time complexity in line with \Cref{thm:avg_of_q-intro}, providing a multiplicative approximation to $\lambda_{i}$.

\begin{theorem}[Informal version of \Cref{thm:eigenvalue_est}]
    \label{thm:eigenvalue_est_intro}
    Fix $\eps,\delta,\beta,\alpha\in(0,1)$ and $\trade\leq \frac{1}{\eps}\log\frac1\delta$. There exists an $(\eps,\delta)$-DP mechanism $\cM$ that gets as input $Z\sim\cD^n$ and for each $i\in[d]$ satisfies the following guarantee: 
    if $\cD$ is subgaussian and 
    \[
    n=\Omega\paren{\frac{d\log1/\delta}{\alpha^2\trade\eps} + \frac{\log1/\delta\log1/\beta}{\alpha^2\trade\eps} + \frac{\log1/\delta}{\alpha^2\eps}}
    \]
    then
    \[
    \Pr_{Z\sim\cD^n}\brackets{1-\alpha\leq \frac{\cM(Z,i)}{ \lambda_{i}(\Sigma)} \leq 1+\alpha}\geq 1-\beta.
    \]
    Moreover, $\cM$ has runs in time $e^{O\paren{\trade}}\poly(n,d)$.
\end{theorem}

\subsubsection{Applications around linear regression}

Taking as inspiration the work of \cite{AsiDT25}, we consider the problem of privately estimating a single parameter in a larger parametric task. As a highlight, in the introduction, we will focus on the setting of linear regression in $d$ dimensions with least squares loss.
In \Cref{sec:m-estimation}, we discuss applications to more general M-estimation tasks.

Our results apply to estimating a single coordinate of the population minimizer. This problem arises naturally in scientific analyses when attempting to infer the effect of some variable $x_1$ on an outcome $y$, while controlling for $d-1$ additional variables $x_2,\dots,x_d$. 

In the following theorem, we assume the data $Z=((x_1,y_1),\dots, (x_n, y_n))$ satisfy $y_i = x_i^\top\theta + e_i$ for some $\theta\in\R^d$, noise $e_i\sim\cN(0,\sigma^2)$ and $x_i\sim\cN(0,\Sigma)$. 
Let $\cD$ denote this distribution over $(x_i,y_i)$. Our first result gives a private mechanism for testing if $\theta_1$ is positive or negative.

\begin{theorem}[Informal Corollary of \Cref{thm:param_testing}]
    Fix $\eps,\delta\in(0,1)$, and $\trade\leq \frac1\eps\log\frac1\delta$. There exists an $(\eps,\delta)$-DP mechanism $\cM$ that gets as input $Z\sim\cD^n$, runs in time $e^{O\paren{\trade}}\poly(n,d)$, and for all $\alpha>0$ and 
    \[
    n = \Omega\paren{\frac{d\log1/\delta}{\trade\eps} + \poly\paren{\frac{1}{\alpha\eps}\log\frac n\delta}},
    \]
    satisfies the following guarantees:
    \begin{itemize}
        \item If $|\theta_1| > 2\alpha\sqrt{\Sigma^{-1}_{11}}$ then $\cM(Z)$ outputs $\sign(\theta_1)$ with probability at least $1-\delta$.
        \item If $|\theta_1| < \alpha\sqrt{\Sigma^{-1}_{11}}$ then $\cM(Z)$ outputs $0$ with probability at least $1-\delta$.
    \end{itemize}
\end{theorem}

Our second result states that we can obtain a similar tradeoff for the problem of estimating $\theta_1$.
For simplicity, we state our result in the regime where $\alpha$ and $\sigma^2$ are constant.

\begin{theorem}[Informal Corollary of \Cref{thm:theta1_estimation}]
    Assume $\theta_1\in[\pm 1]$. Fix $\eps,\delta,\beta\in(0,1)$, and $\trade\leq \frac1\eps\log\frac1\beta$. There exists an $(\eps,\delta)$-DP mechanism $\cM$ that gets as input $Z\sim\cD^n$, runs in time $e^{O\paren{\trade}}\poly(n,d)$, and if 
    \[
    n = \Omega\paren{\frac{d\log1/\beta}{\trade\eps} + \poly\paren{\frac{1}{\eps}\log\frac n\beta\log\log\frac1\delta}},
    \]
    outputs an estimate $\wh\theta_1$ such that $\abs{\wh\theta_1-\theta_1}\leq 0.01\paren{1+ \sqrt{\Sigma^{-1}_{11}}}$ with probability at least $1-\beta$. 
\end{theorem}

In addition to our application to estimating a single parameter, we obtain similar results for estimating $\sigma^2$, and more generally, for estimating the population loss of a class of hypotheses.

\subsection{Our techniques}

\paragraph{Subsampling quantiles for privately evaluating monotone functions.} In order to privately evaluate a monotone function $f$ at a dataset $Z$, we develop a novel analysis of the distribution of $f(S)$ where $S\sim\cS_p(Z)$ is the distribution given by subsampling each $z_i\in Z$ with probability $p$. In a slight abuse of notation, we let  $f(\cS_p(Z))$ denote the distribution of $f(S)$ where $S\sim\cS_p(Z)$, and let $Q_Z(\alpha)$ denote the $\alpha$ quantile of the distribution $f(\cS_p(Z))$. Our privacy analysis shows that, on adjacent datasets $Z$ and $Z'$, the subsampling quantiles $Q_Z$ and $Q_{Z'}$ are \emph{interleaved} in the following sense:
\[
Q_{Z}\paren{\alpha(1-p)}\leq Q_{Z'}\paren{\alpha} \leq Q_{Z}\paren{\frac{\alpha}{1-p}}.
\]
As in \cite{FangDY22,LinderRSS25}, the interleaving relationship allows us to transform the problem of evaluating $f$ into the problem of releasing an interior point from a sequence of increasing quantiles of $f(\cS_p(Z))$. While \cite{FangDY22,LinderRSS25} construct two interleaved sequences by performing a brute force search over all ``large'' subsets of the dataset, our approach can be implemented efficiently via subsampling. Crucially, we show that one can efficiently compute an appropriate sequence of empirical quantiles via subsampling, and that with high probability over the subsampling, the empirical quantiles that are computed on neighbors $Z$ and $Z'$ will satisfy the interleaving relationship. Since each subsample has size approximately $pn$, this approach gives an approximation to the value of $f$ on subsets $S$ with $|S|\approx pn$. The bottleneck in the running time of our algorithms is the number of subsamples required to accurately estimate very small quantiles, which we show is approximately $\exp\paren{{O\paren{\frac{p}{\eps}\log\frac1\delta}}}$. Notice that smaller $p$ translates into a faster runtime but a worse accuracy guarantee---the estimation is on smaller subsets. Thus, our mechanism witnesses a tradeoff between sample complexity and runtime for privately estimating monotone statistics. 

\paragraph{Estimating a single parameter via private loss comparison.} Our algorithm for estimating the $\ord{i}$ coordinate of a high dimensional model works by privately computing estimates of the empirical loss of the minimizer when the $\ord{i}$ coordinate is fixed to different candidates $w\in\cC$ for some set of candidates $\cC\subset\R$, and outputting the candidate $w$ that achieves the smallest empirical loss. We leverage the tools we develop for privately evaluating monotone functions to estimate the loss for each candidate. In order to argue that our mechanism is accurate, we analyze the distribution of the minimum empirical loss on subsampled data. Specifically, we show that for a typical dataset the empirical loss of subsampled data concentrates around its mean at a dimension independent rate. The fast rate of concentration under subsampling allows us to argue that the private loss comparisons are accurate without paying a dimension dependent accuracy term (e.g., $d/\alpha^2$) in the sample complexity.

\paragraph{Query complexity lower bound.} We build on the approaches of \cite{LinderRSS25, SteinkeS25} and prove a query complexity lower bound for privately estimating monotone statistics. %
We combine the central ideas behind the constructions of \cite{LinderRSS25,SteinkeS25} with a technique from \cite{LangeLRV25} used to prove a lower bound in the context of property testing. At a high level, the proof of our lower bound proceeds as follows: we construct two families distributions $\set{\cD_0}$ and $\set{\cD_1}$ and a family of monotone functions $\set{f}$, such that for each $f$, there is a corresponding $\cD_0$ and $\cD_1$ such that for all $n\geq N$, the following holds with probability 1: If $Z'\sim\cD_0^N$ then $f(Z')=y_0$, and if $Z\sim\cD_1^N$ then $f(Z)=y_1$, where $y_0$ and $y_1$ are chosen uniformly from $[\kappa]$. Thus, any algorithm that satisfies a weak accuracy guarantee should output $y_0$ when the data comes from $\cD_0$, and $y_1$ when the data comes from $\cD_1$. We then construct a distribution over pairs of datasets $(X,X')$ at distance $\tau\approx \frac{1}{\eps}\log\min\paren{\frac\kappa\beta, \frac1\delta}$, and show that any algorithm, which has sample complexity $n=N/p$ and query complexity less than $e^{p\tau}$, cannot distinguish $Z$ from $X$ or $Z'$ from $X'$. Thus, if an algorithm outputs $y_1$ on input $Z$ and $y_0$ on input $Z'$, it will also output $y_1$ on input $X$ and $y_0$ on input $X'$. However, since $X$ and $X'$ are at distance $\tau$, this contradicts group privacy, and we obtain that any private and accurate mechanism must make at least $e^{p\tau}$ queries.

\subsection{Related work}

\paragraph{Local sensitivity and subsample-and-aggregate.} The \emph{local sensitivity} of a function $f$ at a dataset $Z$ is the maximum over all $X$ in some neighborhood around $Z$ of $\abs{f(X) - f(Z)}$. 
In order to privately evaluate a function with error proportional to the local sensitivity, \cite{NissimRS07} develop a smoothed version of local sensitivity, called \emph{smooth sensitivity}, and introduce subsample-and-aggregate as a framework for constructing estimators with bounded smooth-sensitivity. Given the generality of subsample-and-aggregate it is no surprise that is has found application in many statistical estimation tasks: see \cite{smith2011privacy, kazan2023test, chadha2024resampling, SinghalS21, tsfadia2024differentially} for examples and further discussion.

\paragraph{Privately evaluating black-box functions.} 
A long line of work explicitly considers the problem of privately evaluating a black-box function $f$. The \emph{inverse sensitivity mechanism} is the most prominent example: it adds noise proportional to the local sensitivity of the function.
\cite{AsiD20} introduced the study of its instance optimality; see \cite{McSherry09,cormode2012differentially} for early applications.
In black-box settings, the mechanism apparently requires a brute-force search over the space of all possible datasets (and thus it is not clear how to even compute it when data domain is infinite).
Work of \cite{JhaR13,LangeLRV25} study the connections between privately release of black-box functions and sublinear-time algorithms.
A related line of work \cite{CummingsD20,FangDY22,KohliL23,LinderRSS25,SteinkeS25} 
 focuses on \emph{down-local} algorithms---that is, algorithms that only query $f$ on subsets of the dataset. 
These algorithms are computable even when the data domain is infinite.

We refer the reader to \cite{LinderRSS25,SteinkeS25} for a more extensive discussion of the benefits of down-local algorithms, and we note that all of our algorithms satisfy this constraint. While the algorithms in this line of work suffer minimal blowup in sample complexity for the setting of \Cref{ass:intro_concentration}, they are unfortunately all inefficient---they run in time at least $n^{\Omega\paren{\log n}}$ in general. Additionally, an important building block in both \cite{LinderRSS25} and \cite{SteinkeS25}, who focus on privately evaluating arbitrary functions, is a technique developed in \cite{FangDY22} for privately evaluating monotone functions. 
Cast in this light, we hope our results will lead toward the development of efficient algorithms for arbitrary functions.

\paragraph{Single-parameter release for GLMs.}
A central motivation of our work aligns with the recent work of \cite{AsiDT25}, who provide algorithms for privately estimating a single parameter from a larger model. 
Operating on a large subset of generalized linear models (GLMs), they show how to privately certify the stability of the empirical estimate and, when that certification passes, adds Gaussian noise at a scale they show is instance-optimal.
However, it appears that their techniques only show that this certification succeeds with high probability once the sample size is quite large; larger than the number of samples needed for adequate non-private estimation.

\paragraph{Privacy and Robustness.} 
A recurring theme in differentially private algorithm design is to use robust estimation as a starting point.
Early work such as propose--test--release \cite{DworkL09} formalized how to privatize estimators that are stable on typical datasets by first privately certifying a sensitivity bound and then releasing a suitably noised estimate.
Most recently, multiple results give general-purpose blueprints for converting robust (or stable) algorithms into private ones \cite{liu2022differential,kothari2022private,hopkins2022efficient,HopkinsKMN23,AsiUZ2023}. 

Going in this direction, efficient algorithms for estimating mean and covariance in the presence of arbitrary outliers have been used to obtain sample-optimal private algorithms for estimating a Gaussian distribution in total variation distance \cite{HopkinsKMN23}.
More broadly, these results and other reductions between robustness and privacy \cite{AsiUZ2023} suggest that privacy and robustness can often be viewed as complementary stability guarantees, and that robust estimators can serve as useful primitives for obtaining private ones.

\ifnum\neurips=1

\subsection{Notation}

In the interest of space, we defer the rest of the preliminaries to \Cref{sec:prelim}, but for clarity define our notation here. 

Let $\cZ$ denote the set of all data points (or data universe). A dataset $Z\in\cZ^n$ is a tuple of $n$ elements where $z_i\in\cZ$ for each $i\in[n]$. Let $\cZ^*=\bigcup_{n\in\N} \cZ^n$. We say two datasets $Z$ and $Z'$ are \emph{neighbors} if there exists exactly one $i\in[n]$ such that $z_i\neq z'_i$. Our results leverage subsampling from a dataset $Z$, which we define as follows: Let $\cZ_\bot=\cZ\cup\set{\bot}$ and let $\cZ^*_\bot = \bigcup_{n\in\N} \cZ^n_\bot$. Suppose $Z\in\cZ^n_\bot$, then a subsample $S\subset Z$ is a tuple of $n$ elements $(s_1,\dots,s_n)$ such that $z_i=\bot\implies s_i=\bot$, and $s_i\in\set{z_i,\bot}$ otherwise. 
We let $|S| = |\set{i\in[n] \mid s_i\neq \bot}|$, and for all $z\in\cZ_\bot$ we let $S^{j\gets z}$ denote $S$ with $s_j\gets z$ and we let $S_{-j}$ denote $S^{j\gets \bot}$.

With this notation, we can now formally define our notion of monotonicity.\footnote{We remark that the literature supplies several distinct notions with the name ``monotonicity.'' The definition we consider here is, in particular, distinct from the notion of \emph{sample-monotone functions} used in \cite{AsiD20}.}

\begin{definition}[Monotone functions] 
  A function $f:\cZ^*_\bot\to\R$ is {\em monotone} if $f(S)\leq f(Z)$ for all $Z\in\cZ^*_\bot$ and $S\subset Z$.
\end{definition}  

\fi

\ifnum\neurips=0

\section{Preliminaries}
\label{sec:prelim}

Let $\cZ$ denote the set of all data points (or data universe). A dataset $Z\in\cZ^n$ is a tuple of $n$ elements where $z_i\in\cZ$ for each $i\in[n]$. Let $\cZ^*=\bigcup_{n\in\N} \cZ^n$. We say two datasets $Z$ and $Z'$ are \emph{neighbors} if there exists exactly one $i\in[n]$ such that $z_i\neq z'_i$. We define subsampling from a dataset $Z$ as follows: let $\cZ_\bot=\cZ\cup\set{\bot}$ and let $\cZ^*_\bot = \bigcup_{n\in\N} \cZ^n_\bot$. Suppose $Z\in\cZ^n_\bot$, then a subsample $S\subset Z$ is a tuple of $n$ elements $(s_1,\dots,s_n)$ such that $z_i=\bot\implies s_i=\bot$, and $s_i\in\set{z_i,\bot}$ otherwise. 
We let $|S| = |\set{i\in[n] \mid s_i\neq \bot}|$, and for all $z\in\cZ_\bot$ we let $S^{j\gets z}$ denote $S$ with $s_j\gets z$ and we let $S_{-j}$ denote $S^{j\gets \bot}$.\gbnote{Definitions and lemmas we need: subgaussian \& pointer to Vershynin}

With this notation, we can now formally define our notion of monotonicity.\footnote{We remark that the literature supplies several distinct notions with the name ``monotonicity.'' The definition we consider here is, in particular, distinct from the notion of \emph{sample-monotone functions} used in \cite{AsiD20}.}

\begin{definition}[Monotone functions] 
  A function $f:\cZ^*_\bot\to\R$ is {\em monotone} if $f(S)\leq f(Z)$ for all $Z\in\cZ^*_\bot$ and $S\subset Z$.
\end{definition}

Next, we present the definition of differential privacy as well as a few standard results that will be useful throughout the paper.

\begin{definition}[Differential privacy \cite{DworkMNS06}]\label{def:dp}
    For random variables $X$ and $Y$ we say $X\approx_{\eps,\delta} Y$ if for all measurable sets $E\subset\supp(X)\cup\supp(Y)$
    \[
    \Pr[X\in E]\leq e^{\eps}\Pr[Y\in E] + \delta,
    \]
    and the same inequality holds with $X$ and $Y$ swapped. If $\delta=0$ we write $X\approx_{\eps} Y$. A randomized algorithm $\cM$ is \emph{$(\eps,\delta)$-differentially private} (DP) if $\cM(Z)\approx_{\eps,\delta} \cM(Z')$ for all neighboring datasets $Z,Z'\in\cZ^n$. 
\end{definition}

\begin{fact}[Composition]
\label{fact:composition}
Fix $\eps_1,\eps_2>0$ and $\delta_1,\delta_2\in(0,1)$. Suppose $\cM_1$ and $\cM_2$ are (respectively) $(\eps_1,\delta_1)$-DP and $(\eps_2,\delta_2)$-DP. Then, the mechanism that, on input $x$, outputs $(\cM_1(x),\cM_2(x))$ is $(\eps_1+\eps_2,\delta_1+\delta_2)$-DP. 
\end{fact}

\begin{fact}[Group privacy]\label{fact:group-privacy}
    Fix $\eps,\delta>0$, and $n \in \N$. Suppose $\cM$ is $(\eps,\delta)$-DP and $E \subsetneq \cY$ is measurable. For any pair of datasets $Z,Z' \in \cZ^n$ that differ in exactly $k$ elements,
    \[
        \Pr_{\cM}\brackets{\cM(Z) \in E} \leq e^{\eps\cdot k}\paren{\Pr_{\cM}\brackets{\cM(Z') \in E} + \frac{\delta}{\eps}}.
    \]
\end{fact}

\begin{definition}[$\Lap$ and $\TLap$]
    \label{def:laplace_distribution}
    The \emph{Laplace} distribution, denoted $\Lap(b)$, is defined over $\R$ by the probability density function $f(x)=\frac1{2b}e^{-|x|/b}$. The \emph{truncated Laplace} distribution, denoted $\TLap(b,\tau)$, is given by the probability density function $f(x)=a_{b,\tau}\cdot\frac1{2b}e^{-|x|/b}$ when $|x|\leq\tau$ and $0$ otherwise, where $a_{b,\tau}$ is a normalizing constant.
\end{definition}

\begin{fact}[Laplace mechanism \cite{DworkMNS06}]
    \label{fact:laplace_mechanism}
    Fix $\eps>0$ and a function $f:\cZ^n_\bot\to \R$. For all neighboring datasets $Z,Z'\in\cZ^n$ with $\abs{f(Z)-f(Z')}\leq \Delta$ we have $f(Z) + \Lap\paren{\frac{\Delta}{\eps}}\approx_{\eps} f(Z') + \Lap\paren{\frac{\Delta}{\eps}}$. Additionally, for all $\delta>0$ we have $f(Z) + \TLap\paren{\frac{\Delta}{\eps}, \frac{\Delta}{\eps}\ln\frac1\delta}\approx_{\eps,\delta} f(Z') + \TLap\paren{\frac{\Delta}{\eps}, \frac{\Delta}{\eps}\ln\frac1\delta}$.
\end{fact}

\begin{fact}[Exponential Mechanism \cite{McSherryT07}]
    \label{lem:exp-mechanism}
    The \emph{exponential mechanism} %
    takes a dataset $Z \in \cZ^n$, computes a score ($\score : \cZ^n \times \cY \to \R$)
    for each $y \in \cY$ with respect to $Z$, and
    outputs $y \in \cY$ with probability proportional
    to $\exp\paren{\frac{\eps \cdot \score(Z,y)}{2 \cdot \Delta}}$,
    where
    $$\Delta = \max\limits_{y \in \cY}\max\limits_{Z\sim Z' \in \cZ^n}
        {\abs{\score(Z,y)-\score(Z',y)}}.$$
    The mechanism is $\eps$-DP and outputs $\wt y$ that satisfies the following guarantee:
    \[
    \Pr\paren{\score(Z, \wt y) \leq
                \max\limits_{y \in \cY}\{\score(Z,y)\} - \frac{2\Delta}{\eps}(\ln\abs{\cY} + t)}
                \leq e^{-t}.
    \]
\end{fact}

\begin{fact}[DP with high probability]
\label{fact:conditional_DP}
    Fix $\eps,\delta,\delta'\in(0,1)$, algorithm $\cA$ and neighbors $Z,Z'$. Suppose there exists an event $G$ over the coins of both $\cA(Z)$ and $\cA(Z')$ such that $\cA(Z)|_G\approx_{\eps,\delta} \cA(Z')|_G$. 
    If $\Pr[G]\geq 1-\delta'$ then $\cA(Z)\approx_{\eps,\delta + \delta'} \cA(Z')$. 
\end{fact}
\begin{proof}
    Fix a set $E$. Then
    \begin{align*}
        \Pr[\cA(Z)\in E]
        &\leq \Pr[\cA(Z)\in E\mid G]\Pr[G] + \delta'\\
        &\leq e^\eps\Pr[\cA(Z')\in E\mid G]\Pr[G] + \delta + \delta'\\
        &\leq e^{\eps}\Pr[\cA(Z')\in E] + \delta + \delta'.\qedhere
    \end{align*}
\end{proof}

\begin{definition}[TV distance]
    For distributions $\cD$ and $\cD'$ over a set $S$ define
    \[
    \dtv(\cD,\cD')=\max_{T\subset S}\paren{\abs{\Pr_{x\sim\cD}[x\in T]-\Pr_{x\sim\cD'}[x\in T]}}.
    \]
\end{definition}

\begin{lemma}[Coupling lemma]
    Let $\cD$ and $\cD'$ be distributions over $\Omega$. For any coupling of $\cV$ of $\cD$ and $\cD'$ we have $    \dtv(\cD,\cD')\leq \Pr_{(X,Y)\sim\cV}\brackets{X\neq Y}$.
\end{lemma}

\begin{lemma}[Weyl's Inequality]\label{lem:weyl}
    Let $A,B\in \R^{d \times d}$ be symmetric matrices with ordered eigenvalues $\lambda_1(\cdot)\le \lambda_2(\cdot)\le \cdots\le \lambda_d(\cdot)$.
    For any $i\in [d]$ we have
        $\abs{\lambda_i(A) - \lambda_i(B)} \le \norm{A - B}{\mathrm{op}}$.
\end{lemma}

\fi

\section{Framework for privately evaluating monotone functions}

In this section, we develop our main tools for privately approximating monotone functions.
We first present and analyze \Cref{alg:quantile_finder}, our ``quantile finder.''
This algorithm is not itself differentially private but provides a strong interleaving guarantee when run on monotone functions, as discussed in the introduction.
We then show how to use the output of this algorithm to produce an approximate median (in \Cref{sec:median_of_quantiles}) or approximate mean (in \Cref{sec:average_of_quantiles}).

\subsection{Subsampling quantiles}

Recall that we write $\cS_p(Z)$ to denote the distribution over $S\subset Z$ given by $s_i\gets z_i$ with probability $p$ and $s_i\gets \bot$ with probability $1-p$. For any function $f$ over $\cZ^n_\bot$ let $f(\cS_p(Z))$ denote the distribution given by sampling $S\sim\cS_p(Z)$ and evaluating $f(S)$. 
For a distribution $\cD$ over a totally ordered domain, let  $Q_\cD(v)=\inf_\ell\set{\ell\mid \Pr_{y\sim\cD}[y\leq \ell]\geq v}$ be the $v$-quantile of $\cD$.

\begin{lemma}[Quantile-finder]
    \label{lem:quantile_finder}
    Fix $p\in (0,1/4)$, $\delta\in(0,1)$, and $\tau\in\N$ and let $T=O\paren{e^{4p\tau}p^{-2}\log\frac 1\delta}$.
    There exists an algorithm $\cA$ that gets query access to a monotone function $f$ on $\cZ^n_\bot$, and on input $Z$, samples $S_1,\dots,S_T\sim\cS_p(Z)$ and outputs a list $q(1), \ldots, q(\tau)$ such that
    \begin{itemize}
        \item \textbf{Accuracy:} $\min_{t\in [T]} f(S_t)\leq q(1)\leq q(2)\leq \ldots \leq q(\tau)\leq \max_{t\in [T]} f(S_t)$.
        \item \textbf{Interleaving:} Fix a neighbor $Z'$ of $Z$, and let $q$ and $q'$ denote the output of $\cA$ on input $Z$ and $Z'$ respectively. Then with probability at least $1-\delta$ for all $t\in [\tau-2]$ we have
        \[
        q'(t)\leq q(t+1)\leq q'(t+2).
        \]
    \end{itemize}
    Moreover $\cA$ has query complexity and runtime $O(T)$.
\end{lemma}
\begin{proof}
Let $c>0$ be sufficiently large and let $\gamma$ be a parameter to be determined later. 
    
    \begin{algorithm}[H]
    \ifnum\neurips=0
    \Statex \textbf{Input:} Dataset $Z\in\cZ^n$, query access to monotone $f$, and parameters $p,\tau$ and $\delta$. 
    \Statex \textbf{Output:} Empirical quantiles $(q(1),\dots, q(\tau))$ of $f(\cS_p(Z))$. 
    \fi
	\begin{algorithmic}[1]
		\caption{\label{alg:quantile_finder} quantile-finder $\cA^f(Z)$}
        \ifnum\neurips=1
        \Statex \textbf{Input:} Dataset $Z\in\cZ^n$, query access to monotone $f$, and parameters $p,\tau$ and $\delta$.
        \Statex \textbf{Output:} Empirical quantiles $(q(1),\dots, q(\tau))$ of $f(\cS_p(Z))$. 
        \fi
        \State Set $\eta\gets \paren{\frac{1-p}{1+\gamma}}^\tau$, and $m\gets \frac{c\log(1/\delta)}{(\eta\gamma(1-\gamma^2)(1-p))^2}$.
		\State Draw $m$ samples from $f(\cS_p(Z))$ and let $\wh\cD$ denote the resulting empirical distribution. 
        \State For each $t\in[\tau]$ compute 
        \[
        q(t) = Q_{\wh\cD}\paren{\paren{\frac{1+\gamma}{1-p}}^t\eta}.
        \]
        \State Return $q(1),\ldots, q(\tau)$.
	\end{algorithmic}
    \end{algorithm}

    Let $Z$ and $Z'$ be neighboring datasets, and let $q$ and $q'$ denote the output of $\cA$ on input $Z$ and $Z'$ respectively.

    \begin{claim}[Quantile interleaving]
        \label{claim:quantile_interleaving}
        For all $t\in[\tau-2]$ we have
        \[
        q'(t)\leq q(t+1)\leq q'(t+2)
        \]
        with probability at least $1-\delta$. 
    \end{claim}

    \ifnum\neurips=0

    \begin{proof}
        Let $F$ and $F'$ denote the CDF of $f(\cS_p(Z))$ and $f(\cS_p(Z'))$ respectively. Additionally, let $\wh F$ and $\wh F'$ denote the empirical CDF of $m$ samples from $f(\cS_p(Z))$ and $f(\cS_p(Z'))$. To prove \Cref{claim:quantile_interleaving}, we will prove that $\wh F$ and $\wh F'$ satisfy the following with probability at least $1-\delta$: For all $\ell\in\R$ the following holds: %
        \begin{align}
            \label{eq:empirical_cdf_interleaving-1}
            \wh F'(\ell)\geq \eta(1-\gamma^2) \implies \wh F(\ell)\in\brackets{(1-\gamma)(1-p), \frac{1+\gamma}{1-p}}\cdot \wh F'(\ell) \\
            \label{eq:empirical_cdf_interleaving-2}
            \wh F(\ell)\geq \eta(1-\gamma^2) \implies \wh F'(\ell)\in\brackets{(1-\gamma)(1-p), \frac{1+\gamma}{1-p}}\cdot \wh F(\ell)  
        \end{align}
        To see that \eqref{eq:empirical_cdf_interleaving-1} and \eqref{eq:empirical_cdf_interleaving-2} suffice to prove \Cref{claim:quantile_interleaving}, let $\ell=q(t+1)$ for some $t\in [\tau-1]$. Then by definition of $q(t+1)$ and $\wh F$ we have
        \(
        \eta(1-\gamma^2)\leq \paren{\frac{1+\gamma}{1-p}}^{t+1}\eta\leq \wh F(\ell),
        \)
        which by \eqref{eq:empirical_cdf_interleaving-2} implies that $\wh F'(\ell)\geq \eta(1-\gamma^2)$. But by \eqref{eq:empirical_cdf_interleaving-1}, this implies that 
        \[
        \paren{\frac{1+\gamma}{1-p}}^{t+1}\eta\leq \wh F(\ell)\leq \frac{1+\gamma}{1-p}\cdot \wh F'(\ell).
        \]
        Rearranging terms yields $\wh F'(\ell)\geq \paren{\frac{1+\gamma}{1-p}}^{t}\eta$, and thus $q'(t)\leq \ell = q(t+1)$. An analogous argument suffices to show that $q(t+1)\leq q'(t+2)$. Since \eqref{eq:empirical_cdf_interleaving-1} and \eqref{eq:empirical_cdf_interleaving-2} hold for all $\ell\in\R$ with probability at least $1-\delta$ we have $q'(t)\leq q(t+1)\leq q'(t+2)$ for all $t\in[\tau-2]$ with probability at least $1-\delta$ as well. 

        To complete the proof of \Cref{claim:quantile_interleaving}, it remains to prove \eqref{eq:empirical_cdf_interleaving-1} and \eqref{eq:empirical_cdf_interleaving-2}. First, we show that the following stronger relationship between $F$ and $F'$ holds for all $\ell\in\R$:
        \begin{align}
            \label{eq:cdf_interleaving}
            (1-p)\cdot F'(\ell) \leq F(\ell) \leq \frac{1}{1-p}\cdot F'(\ell).
        \end{align}
        To see why \eqref{eq:cdf_interleaving} holds, suppose $Z$ and $Z'$ differ on $z_j$, and let $G=\set{S\subset Z \mid s_j=\bot \wedge f(S)\leq \ell}$ and let $B=\set{S\subset Z \mid s_j=z_j \wedge f(S)\leq \ell}$. Define $G'$ and $B'$ analogously for dataset $Z'$. Observe that $F(\ell) = \Pr_{S\sim \cS_p(Z)}\brackets{S\in G} + \Pr_{S\sim \cS_p(Z)}\brackets{S\in B}$. We bound $F(\ell)$ by $F'$ using the following two observations: First, since $Z_{-j} = Z'_{-j}$ we have $\Pr_{S\sim \cS_p(Z)}[S\in G] = \Pr_{S'\sim\cS_p(Z')}[S'\in G']$. And second, since $f$ is monotone, each $S\in B$ can be mapped to a unique $S\in G$ by setting $z_j\gets\bot$, and thus $\Pr_{S\sim\cS_p(Z)}[S\in B]\leq \frac{p}{1-p}\Pr_{S\sim \cS_p(Z)}[S\in G]$. Combining these observations yields
        \begin{align*}
        F(\ell)
        \leq \paren{1+\frac{p}{1-p}}\Pr_{S'\sim\cS_p(Z')}[S'\in G']
        &\leq \frac{1}{1-p}\paren{\Pr_{S'\sim\cS_p(Z')}[S'\in G'] + \Pr_{S'\sim\cS_p(Z')}[S'\in B']} \\
        &= \frac{1}{1-p}\cdot F'(\ell). 
        \end{align*}
        An analogous argument with the roles of $Z$ and $Z'$ swapped suffices to complete the proof of \eqref{eq:cdf_interleaving}. 

        Next, we complete the proof of \eqref{eq:empirical_cdf_interleaving-1} by combining \eqref{eq:cdf_interleaving} with the well-known DKW inequality.
        \begin{fact}[DKW inequality \cite{DvoretzkyKW56,Massart90}]
            \label{fact:DKW}
            Fix $\eps>0$ and $m\in\N$. Let $Y_1,\dots,Y_m$ be independent samples from a distribution $\cD$ over $\R$ with CDF $F$, and for all $t\in\R$ let $\wh F(t) = \frac1m\sum_{i=1}^{m} \Ind[Y_i\leq t]$ be the empirical CDF. Then,
            \[
            \Pr_{Y_1,\dots, Y_m}\brackets{\sup_{t\in\R}\abs{\wh F(t)-F(t)}\geq \eps}\leq 2e^{-2m\eps^2}.
            \]
        \end{fact}

        Let $\zeta = \sup_{\ell\in\R}\abs{\wh F(\ell) - F(\ell)}$ and define $\zeta'$ analogously for $\wh F'$ and $F'$. By \Cref{fact:DKW} and our setting of $m$ we see that $\zeta^*\coloneqq\max(\zeta,\zeta')\leq \gamma(1-\gamma^2)(1-p)\eta/2$ with probability at least $1-\delta$. Conditioned on this event, we see that \eqref{eq:cdf_interleaving} implies 
        \[
        (1-p)\cdot\wh F'(\ell)-2\zeta^* \leq \wh F(\ell)\leq \frac{1}{1-p}\paren{\wh F'(\ell) + 2\zeta^*}.
        \] 
        Now, fix $\ell\in\R$ and suppose that $\wh F'(\ell)\geq\eta(1-\gamma^2)$. Since we assumed $\zeta^*\leq \gamma(1-p)(1-\gamma^2)\eta/2$, we have $2\zeta^*\leq \gamma(1-p) \cdot\wh F'(\ell)$ and hence 
        \[
        (1-p)(1-\gamma)\cdot \wh F'(\ell)\leq \wh F(\ell)\leq \frac{1+\gamma}{1-p}\wh F'(\ell).
        \]
        An analogous argument with the roles of $\wh F$ and $\wh F'$ swapped suffices to prove \eqref{eq:empirical_cdf_interleaving-2}, and completes the proof of \Cref{claim:quantile_interleaving}.  
    
        \end{proof}
    
    \else

    We defer the proof of \Cref{claim:quantile_interleaving} to \Cref{sec:quantile_interleaving_proof} and complete the proof of \Cref{lem:quantile_finder} below.
    
    \fi

    The first guarantee of the lemma follows by inspection of the algorithm, and the second guarantee follows from \Cref{claim:quantile_interleaving}. The query complexity and runtime of $\cA$ are easily seen to be $O(m)$, where $m$ is the number of subsamples drawn by $\cA$. Note that if $\gamma= p/2$ and $p\leq 1/4$, then $\frac{1+\gamma}{1-p}\leq 1+2p$ and thus $m=O\paren{e^{4p\tau}p^{-2}\log\frac1\delta}$.
\end{proof}

\subsection{Median-of-quantiles}\label{sec:median_of_quantiles}

In this section, we use \Cref{lem:quantile_finder} to construct a mechanism for privately evaluating a monotone function $f:\cZ^n_\bot\to[\kappa]$. 
We use the exponential mechanism to release the median of the quantiles returned by the algorithm of \Cref{lem:quantile_finder}.
Accuracy follows immediately, but privacy requires a new argument that relies on the interleaving property.

\begin{theorem}[Median-of-quantiles]
    \label{thm:med-of-q}
    Fix $\eps,\delta,\beta\in(0,1)$ and $p\in(0,1/4)$. There exists an $(\eps,\delta)$-DP mechanism $\cM$ that for all monotone functions $f:\cZ^n_\bot\to[\kappa]$ and datasets $Z$ runs in time at most $T=\exp\paren{O\paren{\frac{p\log\kappa/\beta}{\eps}}}p^{-2}\log\frac1\delta$, samples $S_1,\dots, S_T\sim \cS_p(Z)$, and with probability at least $1-\beta$ satisfies
    \[
    \min f(S_i) \leq \cM^f(Z) \leq \max f(S_i).
    \]
\end{theorem}
\begin{proof}
We construct the median-of-quantiles mechanism below (\Cref{alg:med-of-q}). The mechanism first uses the quantile-finder algorithm given by \Cref{lem:quantile_finder}, and subsequently outputs an approximate median of the resulting quantiles using the exponential mechanism.

    \begin{algorithm}[H]
    \ifnum\neurips=0
    \Statex \textbf{Input:} Dataset $Z\in\cZ^n$, query access to monotone $f$, and parameters $\eps,\delta,\beta,p$.
    \Statex \textbf{Output:} $y\in[\kappa]$.
    \fi
	\begin{algorithmic}[1]
		\caption{\label{alg:med-of-q} median-of-quantiles mechanism $\cM^f(Z)$}
        \ifnum\neurips=1
        \Statex \textbf{Input:} Dataset $Z\in\cZ^n$, query access to monotone $f$, and parameters $\eps,\delta,\beta,p$.
        \Statex \textbf{Output:} $y\in[\kappa]$.
        \fi
        \State Set $\tau\gets\frac{4}{\eps}\log\frac{\kappa}{\beta}$ and let $\Bar{q} = (q(1),\ldots, q(\tau))\gets \cA^f(Z)$.
        \Statex\Comment{$\cA$ is the mechanism given by \Cref{lem:quantile_finder} with parameters $\tau,p$ and $\delta$.}
        \State For each $y\in[\kappa]$ let $\score(y ;  Z) = \abs{\rank(y : \Bar{q}) - \tau/2}$.
        \State Release $y\in[\kappa]$ via the exponential-mechanism with $\score$ as the score.
	\end{algorithmic}
    \end{algorithm}

First, we argue that \Cref{alg:med-of-q} is $(\eps,\delta)$-DP.
Let $EM$ denote the exponential mechanism as in \Cref{alg:med-of-q}, which receives a vector of quantiles and returns some $y\in [\kappa]$.
When we have two vectors of quantiles $\bar q=(q(1),\ldots,q(\tau))$ and $\bar q'=(q'(1),\ldots,q'(\tau))$ that are interleaved---that is, $q'(i)\leq q(i+1)\leq q'(i+2)$ for all $i\in[\tau-2]$, the random variables $EM(\bar q)$ and $EM(\bar q')$ are $\eps$-indistinguishable.
To see this, observe that for any $y\in [\kappa]$ we must have $\abs{\score(y; Z) - \score(y; Z')}\leq 1$, and so the standard privacy argument for the exponential mechanism applies. By \Cref{lem:quantile_finder}, the quantiles $\bar{q}$ and $\bar{q}'$ are interleaved with probability at least $1-\delta$, and thus by \Cref{fact:conditional_DP} mechanism $\cM$ is $(\eps,\delta)$-DP.

    To see why accuracy holds, observe that by the accuracy guarantee of the exponential mechanism we have $\cM^f(Z)\in [q(\tau/4), q(3\tau/4)]$ with probability at least $1-\beta$, which, by the guarantee of \Cref{lem:quantile_finder} is contained in $[\min_i f(S_i), \max_i f(S_i)]$. Last, the final runtime follows from the runtime of $\cA$ and our setting of $\tau$.
\end{proof}

\ifnum\neurips=1

We defer the statement and proof of the average-of-quantiles mechanism, the version of \Cref{thm:med-of-q} that applies real-valued functions, to \Cref{sec:average_of_quantiles}. 

\bibliography{bibliography,biblio}

\appendix

\tableofcontents

\fi

\ifnum\neurips=0

\subsection{Average-of-quantiles}\label{sec:average_of_quantiles}

\else

\section{Average-of-quantiles}\label{sec:average_of_quantiles}

\fi

In this section, we use \Cref{lem:quantile_finder} to construct a mechanism for privately evaluating a monotone function $f:\cZ^n_\bot\to\R$. Unlike \Cref{thm:med-of-q}, the mechanism in this section requires an input parameter $\alpha$ which controls the maximum allowed diameter of the list of quantiles returned by the algorithm of \Cref{lem:quantile_finder}; however, the mechanism in this section works for functions with unbounded range. 
The mechanism first tests that the diameter of the list of quantiles is not too large, and then releases a noisy average of the list of quantiles. 

\begin{theorem}[Average-of-quantiles]
    \label{thm:avg-of-q}
    Fix $\alpha>0$, $\eps,\delta\in(0,1)$ and $p\in(0,1/4)$. There exists an $(\eps,\delta)$-DP mechanism $\cM$ that  for all monotone functions $f:\cZ^n_\bot\to\R$ and, on input $Z$, runs in time  $T=\exp\paren{O\paren{\frac{p\log1/\delta}{\eps}}}p^{-2}\log\frac1\delta$, samples $S_1,\dots S_T\sim \cS_p(Z)$, and satisfies the following: 
    If $\max f(S_i) - \min f(S_i)\leq \alpha$ then 
    \[
    \min f(S_i)-\alpha \leq \cM^f(Z) \leq \max f(S_i)+\alpha.
    \]
\end{theorem}

\begin{proof}
    The proof is similar to that of \Cref{thm:med-of-q}, except instead of running the exponential mechanism, the mechanism first tests if there is a large ``core'' of quantiles within an interval of width $\alpha$, and if there is, the mechanim releases the average of these quantiles plus truncated Laplace noise. 

    \begin{algorithm}[H]
    \ifnum\neurips=0
    \Statex \textbf{Input:} Dataset $Z\in\cZ^n$, query access to monotone $f$, and parameters $\eps,\delta,\alpha,p$. 
    \Statex \textbf{Output:} $y\in\R\cup\set{\bot}$
    \fi
	\begin{algorithmic}[1]
		\caption{\label{alg:avg-of-q} average-of-quantiles mechanism $\cM^f(Z)$}
        \ifnum\neurips=1
        \Statex \textbf{Input:} Dataset $Z\in\cZ^n$, query access to monotone $f$, and parameters $\eps,\delta,\alpha,p$. 
        \Statex \textbf{Output:} $y\in\R\cup\set{\bot}$
        \fi
        \State Set $c=16$, $\eps\gets\eps/2$, $\delta\gets\delta/3$ and $\tau\gets\frac{c}{\eps}\log\frac{1}{\delta}$. 
        \State Let $(q(1),\ldots, q(\tau))\gets \cA^f(Z)$.
        \Statex\Comment{$\cA$ is the mechanism given by \Cref{lem:quantile_finder} with parameters $\tau, p$, and $\delta$.}
        \State Let $t^* = \min\set{t\in[\tau/2] \mid q(\tau-t) - q(t)\leq \alpha}$
        \State If $t^* + \TLap\paren{\frac 1\eps, \frac {\tau}{8}}\leq \frac \tau4 - 1$ then set
        \[
        y\gets \frac{4}{\tau}\sum_{i\in[\tau/4]} q(t^* + i) \hspace{1cm} 
        \]
        and release $\wh y\sim y+\TLap\paren{\frac{16\alpha}{\tau\eps},\frac{16\alpha\log(1/\delta)}{\tau\eps}}$.
        \State Else release $\bot$.
	\end{algorithmic}
    \end{algorithm}

    First, we argue that $\cM$ is $(\eps,\delta)$-DP. Recall that \Cref{lem:quantile_finder} states that for any two neighboring datasets $Z$ and $Z'$ and all $t\in[\tau-2]$ we have
    \[
    q'(t)\leq q(t+1)\leq q'(t+2)
    \]
    with probability at least $1-\delta$. Let $E$ denote the event that the above interleaving holds. Let $\bar q$ and $\bar q'$ denote a fixed list quantiles output by $\cA^f(Z)$ and $\cA^f(Z')$ respectively. We first prove that for fixed $\bar{q}$ and $\bar{q}'$, \Cref{alg:avg-of-q} is DP conditioned on the event $E$, and then we apply \Cref{fact:conditional_DP} and the fact that $\Pr[\bar{q},\bar{q}'\in E]\geq 1-\delta$ to complete the proof. In the remaining parts of the proof, fix $\bar{q}$ and $\bar{q}'$, and let $t^*(Z),y(Z)$ and $t^*(Z'),y(Z')$ denote the values of $t^*$ and $y$ when the inputs are $Z$ and $Z'$ and the quantiles are fixed to $\bar{q}$ and $\bar{q}'$, respectively. 

    The crux of the proof of privacy is the following sensitivity bound.

    \begin{claim}[Stability of the average of quantiles]
    \label{claim:avg_stability}
    Assume that for all $t\in[\tau-2]$ we have
    \(
    q'(t)\leq q(t+1)\leq q'(t+2),
    \)
    and $t^*(Z),t^*(Z') \leq \frac{3\tau}{8}$. Then $|y(Z)-y(Z')|\leq \frac {16\alpha}{\tau}$.
    \end{claim}
    
    \begin{proof}
        First, for all $i\in[\tau/4]$ let $g(i)=q(t^*(Z) + i))$, and define $g'$ analogously for $Z'$. Since the interleaving holds we have $|t^*(Z)-t^*(Z')|\in\zo$, and thus for all $i\in[\tau/4]$
        \[
        g'(i)=q'(t^*(Z')+i)\leq q'(t^*(Z) + i+1)\leq q(t^*(Z) + i+2) = g(i+2),
        \]
        where the second inequality follows by the hypothesis on the quantiles. By the same argument, we also have that $g(i)\leq g'(i+2)$. By definition of $g$, and because $t^*(Z)\leq 3\tau/8$, we have $g(i)
        \in [q(t^*(Z)+1), q(t^*(Z) + \tau/4)]$ for all $i\in[\tau/4]$ (the same holds for $g'$), and by our assumption on $t^*$, the width of this interval is at most $\alpha$. Taken together, this implies that $g'(1)\leq g(3)\leq g'(5)\leq g(\tau/4)$, and thus, that $g(\tau/4) - g'(1)\leq 2\alpha$ (since $g'(5)-g'(1)\leq \alpha$ and $g(\tau/4)-g(3)\leq \alpha$). By the same argument, $g'(\tau/4)-g(1)\leq 2\alpha$.

        Next, we express $y(Z)$ and $y(Z')$ in terms of $g$ and $g'$, and using the above inequalities argue that the difference is small.
        \begin{align*}
            y(Z) - y(Z') &= \frac{4}{\tau} \paren{\sum_{i\in[\tau/4]}g(i)-g'(i)} \\
                &= \frac{4}{\tau} \paren{\sum_{i\in[\tau/4-2]}g(i)-g'(i+2)} + \frac 4 \tau \paren{g(\tau/4-1) + g(\tau/4) - g'(1) - g'(2)}.
        \end{align*}
        Since $g(i) \le g'(i+2)$, the first term on the right-hand side is at most zero.
        Continuing, we have
        \begin{align*}
            y(Z) - y(Z') &\le \frac 4 \tau \paren{g(\tau/4-1) + g(\tau/4) - g'(1) - g'(2)} \\
                &\le \frac 8 \tau\paren{g(\tau/4) -  g'(1)} \\
                &\le \frac{16\alpha }{\tau}.
        \end{align*}
        A symmetrical argument establishes an upper bound on $y(Z')-y(Z)$, so we are done.
    \end{proof}

    To complete the proof of privacy, let $G=\set{(Z,Z') \mid \abs{y(Z) - y(Z')}\leq \frac{16\alpha}{\tau}}$ and condition on the event $E$. We consider two cases: First, suppose $(Z,Z')\not\in G$, then by \Cref{claim:avg_stability}, at least one of $t^*(Z)$ or $t^*(Z')$ is at least $\frac{3\tau}{8}$. Since we are conditioned on $E$, we have that both are at least $\frac{3\tau}{8} -1$, and thus $\cM^f(Z)$ and $\cM^f(Z')$ both output $\bot$. Next, suppose $(Z,Z')\in G$. Observe that by \Cref{claim:avg_stability} we have $\abs{y(Z)-y(Z')}\leq \frac{16\alpha}{\tau}$ and thus 
    \[
    y(Z)+\TLap\paren{\frac{16\alpha}{\eps\tau},\frac{16\alpha\log(1/\delta)}{\eps\tau}}\approx_{\eps,\delta} y(Z')+\TLap\paren{\frac{16\alpha}{\eps\tau},\frac{16\alpha\log(1/\delta)}{\eps\tau}}.
    \]
    Similarly, since we are conditioned on $E$, we have $\abs{t^*(Z)-t^*(Z')}\leq 1$. Therefore,
    \[
    t^*(Z) + \TLap\paren{\frac1\eps,\frac \tau8}\approx_{\eps,\delta} t^*(Z') + \TLap\paren{\frac 1\eps, \frac\tau8}.
    \]
    By composition (\Cref{fact:composition}), we have that $\cM^f(Z)|_E\approx_{2\eps,2\delta} \cM^f(Z')|_E$. Since $E$ occurs with probability at least $1-\delta$ we can apply \Cref{fact:conditional_DP} to see that $\cM^f(Z)\approx_{2\eps,3\delta} \cM^f(Z')$. This completes the proof of privacy since we rescaled $\eps$ and $\delta$ at the start of \Cref{alg:avg-of-q}.

    Runtime follows by \Cref{lem:quantile_finder} and our setting of $\tau$. To see why the accuracy guarantee holds, observe that if $\max f(S_i) - \min f(S_i)\leq \alpha$ then $\cM^f(Z)$ releases
    \[
    \wh y\sim y(Z)+\TLap\paren{\frac{16\alpha}{\eps\tau},\frac{16\alpha\log(1/\delta)}{\eps\tau}}.
    \] 
    Since $\tau=\frac{c\log(1/\delta)}{\eps}$ and $c=16$, the noise is at most $\alpha$, and hence $\abs{\wh y - y(Z)}\leq \alpha$. By the definition of $y(Z)$, this implies that $\min f(S_i)-\alpha\leq \wh y\leq \max f(S_i)+\alpha$.
\end{proof}

\ifnum\neurips=1

\section{Proof of \texorpdfstring{\Cref{claim:quantile_interleaving}}{quantile interleaving}}
\label{sec:quantile_interleaving_proof}

In this section, we complete the proof of the quantile interleaving claim.

\fi

\section{Application to eigenvalue estimation}\label{sec:eigenvalue_estimation}

In this section, we explain how to apply our techniques to estimating the $\ord{i}$ largest eigenvalue $\lambda_{i}(\Sigma)$ of the covariance $\Sigma$ of a subgaussian distribution for any $i\in[d]$. Our main result provides a tradeoff between sample and time complexity for estimating a single eigenvalue. 

\begin{theorem}[Private eigenvalue estimation]
    \label{thm:eigenvalue_est}
    There exists a mechanism $\cM$ that, for all privacy parameters $\eps,\delta\in(0,1)$, failure probability $\beta>0$, accuracy parameter $\alpha\in(0,1/4)$, and subsampling probability $p\in(0,\frac 14)$, satisfies $(\eps,\delta)$-DP. Moreover, for all subgaussian distributions $\cD$ over $\R^d$ with covariance $\Sigma \succ 0$, there exists a constant $K_\cD$, such that for all $i\in[d]$, the following holds: 
    
    If
    \[
    n=\Omega\paren{\frac{K_\cD}{\alpha^2}\paren{\frac{d + \log(1/\beta)}{p} + \frac{\log(1/\delta)}{\eps} + \frac{\log(p^{-1}\log 1/\delta)}{p}}},
    \]
    then
    \[
    \Pr_{Z\sim\cD^n}\brackets{1-\alpha\leq \frac{\cM(Z,i)}{ \lambda_{i}(\Sigma)} \leq 1+\alpha}\geq 1-\beta.
    \]
    Moreover, $\cM$ has runs in time $\exp\paren{O\paren{\frac{p\log1/\delta}{\eps}}}\poly\paren{\frac{\log1/\delta}{p}}$.
\end{theorem}

The proof of \Cref{thm:eigenvalue_est} leverages the following standard concentration result regarding the convergence of the empirical covariance for subgaussian distributions. For a dataset $Z$ of size $n$ let $\wh\Sigma(Z)=\frac 1n\sum_{i\in[n]} Z_iZ^T_i$.

\begin{theorem}[\cite{vershynin2018high}]
    \label{thm:cov_concentration}
    Let $\cD$ be a subgaussian distribution over $\R^d$ with covariance matrix $\Sigma\succ 0$. Then there exists a constant $K_\cD$ such that for all $n\in\N$, $\beta>0$, and $Z\sim\cD^n$ 
    \[
    \norm{\Sigma^{-1/2}\wh\Sigma(Z)\Sigma^{-1/2} - \mathbb I}{op} \leq K_\cD\paren{\sqrt{\frac{d+\log(2/\beta)}{n}} + \frac{d+\log(2/\beta)}{n}}
    \]
    with probability at least $1-\beta$.
\end{theorem}

\begin{remark}
    \label{rem:weyl}
    By Weyl's inequality (\cref{lem:weyl}), \Cref{thm:cov_concentration} implies that the eigenvalues of the empirical covariance and the true covariance are close on a multiplicative scale---that is, 
    \[
    \abs{\frac{\lambda_{i}(\wh\Sigma(Z))}{\lambda_{i}(\Sigma)} - 1}\leq \norm{\Sigma^{-1/2}\wh\Sigma(Z)\Sigma^{-1/2} - \mathbb I}{op}
    \]
    for all $i\in[d]$. 
\end{remark}

\begin{proof}[Proof of \Cref{thm:eigenvalue_est}]
    The proof proceeds by applying the average-of-quantiles mechanism (\Cref{thm:avg-of-q}) to estimate $\lambda_{i}$ on a $\log$ scale. The additive error guarantee of \Cref{thm:avg-of-q}, combined with the concentration guarantee of \Cref{thm:cov_concentration}, provides a multiplicative estimate of $\lambda_{i}$. Note that for all $m\in\N$, the function $h_m(Z,i)=\frac{1}{m}\cdot \lambda_{i}\paren{\sum_{j\in[n]} Z_jZ_j^T}$, the $\ord{i}$ largest eigenvalue, is monotone (this follows from the min-max characterization of the eigenvalues). First, we argue that with high probability over $Z$ and $S\sim\cS_p(Z)$ we have $h_{np}(S,i)/ \lambda_{i}(\Sigma)\approx 1$. Second, we apply the average-of-quantiles mechanism from \Cref{thm:avg-of-q} to privately estimate $\log (h_{np}(S))$ to within some appropriate additive error, which translates to a multiplicative estimate of $\lambda_{i}(\Sigma)$. 

    \begin{claim}
        \label{claim:approx_eigenvalue}
        For all $\alpha\in(0,1)$ and
        \[
        n=\Omega\paren{K_\cD\paren{\frac{d + \log(1/\beta)}{\alpha^2p}}},
        \]
        if $Z\sim\cD^n$ and $S\sim\cS_p(Z)$, then with probability at least $1-\beta$
        \[
        1-\alpha\leq \frac{h_{np}(S,i)}{\lambda_{i}\paren{\Sigma}}\leq 1 \pm \alpha.
        \]
    \end{claim}
    \begin{proof}
        We prove this in two steps. 
        First, we show that the empirical estimate $h_m(S,i)$ (with $m$ random) is close to the true quantity 
        $\lambda_i(\Sigma)$.
        Then, we show that $m$ concentrates around its expectation, so $h_m(S,i)$ is close to $h_{np}(S,i)$.
        
        Let $m=|S|$. Then $m\sim\Bin(n,p)$ and thus, by a Chernoff bound we have $\abs{m-pn}\leq \sqrt{cpn\log1/\beta}$ with probability at least $1-\beta/2$. It follows that $\abs{h_{np}(S,i) - h_{m}(S,i)}=h_{m}(S,i)\abs{\frac{m-np}{np}}$ and that
        \[
        \abs{h_{np}(S,i) - h_{m}(S,i)}\leq h_{m}(S,i)\sqrt{\frac{c\log1/\beta}{np}}
        \]
        with probability at least $1-\beta/2$. Additionally, since $\sqrt{cpn\log1/\beta}\leq pn/2$, we have that $m\geq pn/2$ with probability at least $1-\beta/2$, and conditioned on this event, \Cref{thm:cov_concentration}, \Cref{rem:weyl}, and our setting of $n$ imply that $\frac{h_{m}(S,i)}{\lambda_{i}(\Sigma)}\in [1\pm \frac{\alpha}{4}]$ with probability at least $1-\beta/2$. By the union bound, we obtain that $\frac{ h_{np}(S,i)}{\lambda_{i}(\Sigma)}\in [1\pm \alpha]$ with probability at least $1-\beta$.  
    \end{proof}

    Let $\cA$ denote the average-of-quantiles mechanism (\Cref{thm:avg-of-q}) with accuracy parameter $\alpha$, privacy parameters $\eps$ and $\delta$, subsampling probability $p$, and query access to $\log (h_{np})$. We argue that $\cA(Z)$ outputs a $y$ such that $|y-\log h_{np}(S)|\leq \alpha$ for some subsample $S$ drawn by $\cA$, and thus, by \Cref{claim:approx_eigenvalue} that $|y - \log\lambda_{i}(\Sigma)|\leq 2\alpha$, and by extension, that $e^{-\alpha}\leq \frac{e^y}{\lambda_{i}(\Sigma)}\leq e^{\alpha}$. By the accuracy guarantee of $\cA$, it suffices to argue that $\abs{\log (h_{np}(S,i)) - \log(\lambda_{i}(\Sigma))}\leq \alpha/2$ for every subsample $S$ drawn by the algorithm. Let $S_1,\dots, S_T$ denote the $T=\exp\paren{O\paren{\frac{p\log1/\delta}{\eps}}}\poly\paren{\frac{\log1/\delta}{p}}$ subsamples drawn by $\cA$, and let $B_j$ denote the event that $\abs{\log \frac{h_{np}(S_j,i)}{\lambda_{i}(\Sigma)}}> \alpha/2$. By \Cref{claim:approx_eigenvalue} and our setting of $n$, we have
    $\Pr\brackets{\bigcup_{j\in[T]} B_j}\leq \beta$. Now, conditioned on $\bigcap \overline{B_j}$, by the accuracy guarantee of $\cA$, we see that $\cA(Z)$ outputs $y$ such that $\abs{y-\log \lambda_{i}(\Sigma)}\leq 2\alpha$. By our bound on $\Pr[\bigcap \overline B_j]$, we have $e^{-\alpha}\leq \frac {e^y}{\lambda_{i}(\Sigma)}\leq e^{\alpha}$ with probability at least $1-\beta$. Letting $\cM$ be the mechanism that sets $y\gets \cA(Z)$ and returns $e^{y}$, completes the proof.     
\end{proof}

\section{Applications to M-estimation}\label{sec:m-estimation}

In this section, we will turn the tools we have developed to \emph{M-estimation}, that is, solving problems of the form
\begin{equation}
    \min_\theta \frac 1 n \sum_{i=1}^n \ell_\theta(z_i)
    \label{eq:m_est_def}
\end{equation}
for some loss function $\ell$.
We will present applications to univariate estimation and testing problems based on~\eqref{eq:m_est_def}, including estimating a single coordinate of the minimizer and estimating the minimum achievable loss.

\Cref{sec:m_est_setup} will further develop the learning setting we address and present the assumptions we use.
We then state our results: \Cref{sec:loss_estimation} contains results for approximating the best-possible population loss and \Cref{sec:parameter_estimation} contains results for parameter estimation.

\subsection{Setup: Exponential Families and Concentration}
\label{sec:m_est_setup}

In the context of M-estimation, we associate a \emph{learning task} with a triple $(\Theta,\ell,\cD)$.
Here $\Theta$ is a set of parameters, $\ell = \{\ell_\theta : \theta\in \Theta\}$ a set of loss functions, and $\cD$ a data-generating distribution over a data space $\cZ$, which we leave implicit.
We assume that $\ell_\theta :\cZ_\bot \to \R_{\ge 0}$ satisfies $\ell_\theta(\bot)=0$ for all $\theta$ (recall our convention that $\cZ_\bot = Z \cup \{ \bot\}$).
Ideally, we would minimize over $\Theta$ with respect to the underlying distribution $\cD$ itself:
\[
    \theta(\cD) = \argmin_{\theta\in\Theta} \Ex_{z\sim \cD}[\ell_\theta(z)]
        = \argmin_{\theta\in\Theta} L_\theta(\cD), 
\]
where we have defined $L_\theta(\cD)$ to denote the expected population loss.
We will use $L(\cD)$ to denote $\min_\theta L_\theta(\cD)$, the smallest possible population loss.

As we can only approximate $\cD$ through a finite set of samples $Z = (z_1,\ldots,z_n)$ drawn i.i.d., we will extend our notation and define
\[
    \theta(Z) = \argmin_{\theta\in \Theta} \frac 1 n \sum_{i=1}^n \ell_\theta(z_i)
    = \argmin_{\theta\in\Theta} L_\theta(Z).
\]
Similarly, we define $L_\theta(Z)$ as the empirical loss and $L(Z)$ as the smallest possible empirical loss. 

\paragraph{Exponential Families}
One important example of M-estimation arises from the problem of modeling exponential families.
An \emph{exponential family model} is a set of distributions over an underlying space $\cZ$ of the form
\[
    p_\theta(z) = h(z) \cdot \exp\paren{\langle \theta, \phi(z)\rangle - A(\theta)}.
\]
The distribution is defined by the parameter vector $\theta$, the \emph{sufficient statistic} $\phi$,  and the \emph{carrier} $h$.
The function $A(\theta)$ is the \emph{log partition function}.
When fitting such a model to a set of observations $Z=(z_1,\ldots,z_n)$, we take as our loss function the negative log likelihood:
\begin{equation}
    \min_\theta \sum_i \ell_\theta(z_i) =  \min_\theta \sum_i - \log p_\theta(z_i) = \min_\theta \sum_i A(\theta) - \langle \theta  , \phi(z_i)\rangle.
    \label{eq:exp_fam_mle}
\end{equation}
This maximum likelihood problem has several advantageous features.
Since $A(\theta)$ is convex, the overall problem is also convex and can be solved efficiently.
Additionally, the solution to \eqref{eq:exp_fam_mle} admits a clean asymptotic description:
as $n\to \infty$, we have
\[
    \theta(Z) - \theta(\cD) \sim \cN(0, n^{-1} \nabla^2 A(\theta(\cD))^{-1}).
\]
If the distribution satisfies some mild regularity conditions and the sufficient statistic $\phi$ is \emph{minimal}, informally meaning that the family is not overparameterized, we furthermore have that $\nabla^2 A(\theta(\cD)) \succ 0$, i.e., in a neighborhood around the true minimizer the loss function behaves as if it is strongly convex.
Crucially for some of our downstream applications, we do not require strong convexity to hold for the whole loss function.
For our purposes, it will suffice to show that any parameter sufficiently far from the empirical minimizer will have appreciably larger loss.

For more discussion on this topic and as a point of comparison, we refer the reader to \cite{AsiDT25}.
Their methods rely not versions of the above conditions which are themselves privately certifiable; that is, their algorithms check whether these local conditions hold and this check must itself be privatizable.
In contrast, our approach allows us to write down the conditions we desire and establish accuracy whenever these conditions are met, a simpler argument.

\newcommand{\Nconc}{N_{\ref{ass:loss_concentration}}}
\newcommand{\Kconc}{K_{\ref{ass:loss_concentration}}}

\paragraph{Learning Task Assumptions}
We now operationalize the above discussion and state the assumptions we will use going forward to establish accuracy.
In both \Cref{ass:loss_concentration,ass:id}, we attempt to write down an interpretable set of general conditions under which our algorithm delivers improvements over existing results.
We emphasize that our privacy guarantees do not rely on assumptions about the distribution.

Our results in \Cref{sec:loss_estimation} provide accurate loss estimation for learning tasks $(\Theta, \ell,\cD)$ that satisfy the following assumption.
\begin{assumption}[Concentration of the empirical loss]
    \label{ass:loss_concentration}
    For task $(\Theta,\ell,\cD)$, there exists a function $\Nconc:\R\to\R$ and a value $\Kconc>0$ such that for all $\alpha,\beta>0$ and $n\geq \Nconc(\alpha) + \frac{\Kconc\log1/\beta}{\alpha^2}$ 
    \[
    \Pr_{Z\sim\cD^n}\brackets{\abs{L(Z) -L(\cD)}\geq \alpha}\leq \beta.
    \]   
\end{assumption}
Informally, this captures a type of subgaussian concentration. 
In our applications, $\Nconc$ will correspond to a dimension-dependent term or a measure of the complexity of $\Theta$.

\newcommand{\Kid}{K_{\ref{ass:id}}}
\newcommand{\Nid}{N_{\ref{ass:id}}}
\newcommand{\Mid}{M_{\ref{ass:id}}}

Our results in \Cref{sec:parameter_estimation}, which focus estimating a single coordinate of the larger minimizer $\theta(Z)$, require more care to set up.
Recall that $L^{(w)}(Z)$ denotes the minimum empirical loss achievable among parameters $\theta$ with $\theta_1=w$ and $Z_{-i}$ denotes the dataset $Z$ with the $i$-th element removed.

\begin{assumption}[Identifiability, Concentration, Stability] 
    \label{ass:id}
    For task $(\Theta, \ell, \cD)$, there exists a function $\Nid:\R\to\R$ and values $
    \Mid,\Kid,\allowbreak\lambda,\allowbreak\mu,\sigma>0$ such that the following holds for all $\alpha,\beta>0$:
    \begin{assumplist}
        \item\label{ass:param_concentration} For all $n\geq \Nid(\alpha) + \frac{\Kid\log1/\beta}{\alpha^2}$ we have
        \[
        \Pr_{Z\sim\cD^n}\brackets{\abs{\theta(Z)_1-\theta(\cD)_1}\geq \alpha\mu^{-1/2}}\leq \beta.
        \]
        \item\label{ass:tssc} For all $n\geq \Mid + \Kid\log1/\beta$ we have
        \[
        \Pr_{Z\sim\cD^n}\brackets{\forall w\in\R\colon \frac{\mu(w-\theta(Z)_1)^2}{2}\leq  L^{(w)}(Z) - L(Z) \leq 2\mu(w-\theta(Z)_1)^2}\geq 1-\beta.
        \]
        \item\label{ass:loo} For all $n\geq \Mid+\Kid\log1/\beta$ and all $i\in[n]$ we have 
        \[
        \Pr_{Z\sim\cD^n}\brackets{\abs{L(Z)-L(Z_{-i})}\geq \frac{\lambda\log 2/\beta}{n}} \leq \beta.
        \]
        \item\label{ass:tb} For all $n\in\N$ we have 
        \[
        \Pr_{Z\sim\cD^n}\brackets{L(Z)\geq \sigma^2\paren{1 + \frac{\log1/\beta}{n}}}\leq \beta.
        \]
    \end{assumplist}
\end{assumption}

We now discuss each part of the assumption informally.
In \Cref{sec:regression_facts} we prove that this assumption is satisfied for random design linear regression: \Cref{ass:param_concentration} is straightforward: we ask that the first coordinate of the empirical minimizer is usually close to the first coordinate of the true minimizer. \Cref{ass:tb} is similarly transparent: the minimum empirical loss should have  subexponential behavior on its upper tail.
As a key example, this will be satisfied when we have subexponential concentration of $\ell_{\theta(\cD)}(z)$, the loss of a single observation with respect to the true minimizer.
Under such conditions we also expect \Cref{ass:loo} to hold, as any single observation is unlikely to strongly influence the empirical minimizer.

Informally, \Cref{ass:tssc} enforces an approximate notion of \emph{identifiability}: 
any $\theta'$ with close-to-optimal loss must have $\theta_1'\approx \theta(Z)_1$.
Technically, we read it as a ``typical smoothness and strong convexity'' condition, yielding quantitative control over this tradeoff.
The steps needed to establish this tradeoff are significantly more delicate than those for the other three, as we have to reason about loss minimization with the first parameter fixed.
Weaker conditions would also suffice: the core of the argument requires us to show that poor choices for the first parameter lead to large losses.

\subsection{Estimating the population loss}
\label{sec:loss_estimation}

Our first result witnesses a sample and time complexity tradeoff for the problem of estimating the population loss. 
The subsampling hyperparameter $p$ controls this tradeoff.

\begin{theorem}[Private loss estimation via subsampling]
    \label{thm:loss_estimation}
    There exists a mechanism $\cM$ that, for any task $(\Theta,\ell,\cD)$ satisfying \Cref{ass:loss_concentration}, privacy parameters $\eps,\delta\in(0,1)$, accuracy parameter $\alpha>0$, and subsampling probability $p\in(0,\frac 14)$, satisfies $(\eps,\delta)$-DP and the following accuracy guarantee:

    If 
    \[
    n=\Omega\paren{\frac{\Nconc(\alpha/8)}{p} + \frac{\Kconc+L(\cD)^2}{\alpha^2}\paren{\frac{\log1/\delta}{\eps} + \frac{\log 1/\delta}{p} + \frac{\log ( p^{-1}\log1/\delta)}{p}}}
    \]
    then %
    \[
    \Pr_{Z\sim\cD^n}\brackets{\abs{\cM^L(Z) - L(\cD)} \leq \alpha}\geq 1-\delta.
    \]
    Moreover, $\cM$ has query complexity and runtime $\exp\paren{O\paren{\frac{p\log1/\delta}{\eps}}}\poly\paren{\frac{\log1/\delta}{p}}$.
\end{theorem}

Our next result provides an improved sample and time complexity tradeoff for the problem of testing whether the loss is large. Essentially, the result states that we can achieve the same tradeoff as in \Cref{thm:loss_estimation}, but with the privacy parameter $\delta$ replaced with a failure probability $\beta$. 

\begin{theorem}[Private loss testing via subsampling]
    \label{thm:loss_testing}
    There exists a mechanism $\cM$ that, for all tasks $(\Theta,\ell,\cD)$ satisfying \Cref{ass:loss_concentration}, privacy parameters $\eps,\delta\in(0,1)$, failure probability $\beta>0$, accuracy parameter $\alpha>0$, and subsampling probability $p\in(0,\frac 14)$, satisfies the following: 

    If 
    \[
    n=\Omega\paren{\frac{\Nconc(\alpha/2)}{p} + \frac{\Kconc+L(\cD)^2}{\alpha^2}\paren{\frac{\log1/\beta}{\eps} + \frac{\log1/\beta}{p} + \frac{\log ( p^{-1}\log1/\delta)}{p}}}
    \]
    then the following holds:  
    \begin{itemize}
        \item If $L(\cD)\geq 2\alpha$ then $\Pr_{Z\sim\cD^n}\brackets{\cM^L(Z)\text{ outputs reject}}\geq1-\beta$. 
        \item If $L(\cD)\leq \alpha$ then $\Pr_{Z\sim\cD^n}\brackets{\cM^L(Z)\text{ outputs accept}}\geq1-\beta$.
    \end{itemize}
    Moreover, $\cM$ has query complexity and runtime $\exp\paren{O\paren{\frac{p\log1/\beta}{\eps}}}\poly\paren{\frac{\log1/\delta}{p}}$.  
\end{theorem}

\subsubsection*{Proof of \texorpdfstring{\Cref{thm:loss_estimation}}{private loss estimation theorem}}

\begin{proof}
    The proof is a direct application of \Cref{thm:avg-of-q}: since $\ell_\theta$ is non-negative for all $z$ and $\theta$, we have that $L(Z)$ is monotone, and thus we can safely apply the average-of-quantiles mechanism  (\Cref{thm:avg-of-q}) with query access to $L$. Let $\cA$ denote the average-of-quantiles mechanism with accuracy parameter $\alpha'\gets\alpha/2$, privacy parameters $\eps$ and $\delta$, and subsampling probability $p$. Let $\cM$ be the mechanism which, given query access to $L$ and input $Z$, simulates $\cA$ with query access to $\frac{L}{p}$ and input $Z$. 

    We first prove a claim which states that if $Z\sim\cD^n$ and $S\sim\cS_p(Z)$ then with high probability $\abs{\frac{L(S)}{p} - L(\cD)}\lesssim \alpha + o(1)$. This is not immediate from \Cref{ass:loss_concentration} since $|S|\sim\Bin(n,p)$ and is not always equal to $np$. However, the proof follows from a straightforward application of Chernoff bounds for Bernoulli random variables.

    \begin{claim}
        \label{claim:approx_loss}
        There exists a constant $c>0$ such that for all 
        \[
        n=\Omega\paren{\frac{\Nconc(\alpha)}{p} + \frac{\Kconc\log1/\beta}{\alpha^2p}},
        \]
        if $Z\sim\cD^n$ and $S\sim\cS_p(Z)$ then 
        \[
        \abs{\frac{L(S)}{p} - L(\cD)}\leq \alpha\paren{1+\sqrt{\frac{c\log1/\beta}{n}}} + L(\cD)\sqrt{\frac{c\log1/\beta}{n}}
        \]
        with probability at least $1-\beta$.
    \end{claim}
    \begin{proof}
        Let $m=|S|$. Then $m\sim\Bin(n,p)$ and thus, by a Chernoff bound we have $\abs{m-pn}\leq \sqrt{cpn\log1/\beta}$ with probability at least $1-\beta/2$. It follows that $\abs{L(S)/p - L(S)n/m}=\frac{L(S)}{mp}\abs{m-np}$. Applying the bound on $\abs{m-pn}$ and multiplying by $n/n$ yields
        \[
        \abs{\frac{L(S)}{p} - \frac{n L(S)}{m}}\leq \frac{n L(S)}{m}\sqrt{\frac{c\log1/\beta}{np}}
        \]
        with probability at least $1-\beta/2$. Now, applying \Cref{ass:loss_concentration} to $S$ and $\frac{n L(S)}{m}$, we see that $\abs{\frac{n L(S)}{m} - L(\cD)}\leq \alpha$ with probability at least $1-\beta/2$. Combining the two bounds completes the proof.
    \end{proof}

    Let $T=\exp\paren{O\paren{\frac{p\log1/\delta}{\eps}}}\poly\paren{\frac{\log1/\delta}{p}}$ and let $S_1,\dots S_T\sim\cS_p(Z)$ denote the subsamples drawn by $\cA$. If $B$ is the event that $\max_{i\in [T]}\abs{L(\cD)-L(S_i)/p}\geq \alpha'/2$, then conditioned on $\overline B$, we have $\max L(S_i)/p - \min L(S_i)/p\leq \alpha'$, and thus, by the accuracy guarantee of \Cref{thm:avg-of-q}, that $\abs{\cA^{L/p}(Z) - L(\cD)}\leq 2\alpha'=\alpha$. Hence, it suffices to show that
    \[
    \Pr_{\substack{Z\sim\cD^n \\ S_1,\dots S_T\sim\cS_p(Z)}}\brackets{\max_{i\in [T]}\abs{L(\cD)-L(S_i)/p}\geq \alpha'/2}\leq \delta.
    \]
    To see why this holds, observe that if $n\geq \Omega\paren{\frac{\Nconc(\alpha'/4)}{p} + \frac{(\Kconc+L(\cD)^2)\log 1/\beta}{\alpha^2p}}$ then by \Cref{claim:approx_loss} and \Cref{ass:loss_concentration}  %
    \[
    \Pr_{\substack{Z\sim\cD^n \\ S\sim\cS_p(Z)}}\brackets{\abs{L(\cD)-L(S)/p}\geq \alpha'/2}\leq \beta,
    \]
    and thus, setting $\beta\gets \delta/T$ and applying a union bound over all $T$ subsamples completes the proof.
\end{proof}

\subsubsection*{Proof of \texorpdfstring{\Cref{thm:loss_testing}}{private loss testing theorem}}

\begin{proof}
    Let $\cA$ denote the median-of-quantiles mechanism (\Cref{thm:med-of-q}) with subsampling probability $p$, privacy parameters $\eps,\delta$, and failure probability $\beta/2$. Let $\cM$ be the mechanism which, given query access to $L$ and input $Z$, simulates $\cA$ with query access to $h$ and input $Z$, where $h(Z) = \Ind\brackets{\frac{L(Z)}{p}\geq 3\alpha/2}$. If $\cA^h(Z)=1$ then $\cM$ outputs ``reject'' and otherwise $\cM$ outputs ``accept''. 
    
    We argue that $\cA$ outputs $1$ with probability at least $1-\beta$ when $L(\cD)\geq 2\alpha$, and outputs $0$ with probability at least $1-\beta$ when $L(\cD)\leq \alpha$. Let $T=\exp\paren{O\paren{\frac{p\log1/\beta}{\eps}}}\poly\paren{\frac{\log1/\delta}{p}}$. By our setting of $n$, and the same argument as in the proof of \Cref{thm:loss_estimation} (except now with $\beta\gets \beta/2T$ instead of $\beta\gets \delta/T$), we have that 
    \[
    \Pr_{\substack{Z\sim\cD^n \\ S_1,\dots S_T\sim\cS_p(Z)}}\brackets{\max_{i\in [T]}\abs{L(\cD)-L(S_i)/p}\geq \alpha/2}\leq \beta/2.
    \]
    Thus, if $L(\cD)\geq 2\alpha$, then $h(S_i)=1$ for all $i\in[T]$ with probability at least $1-\beta/2$, and therefore $\cA$ outputs $1$ with probability at least $1-\beta$. By the same argument, if $L(\cD)\leq \alpha$ then $\cA$ outputs $0$ with probability at least $1-\beta$. The accuracy guarantee now follows from the definition of $\cM$. The runtime and privacy guarantees follow from \Cref{thm:med-of-q} and the fact that $h$ is monotone (since $L$ is monotone). 
\end{proof}

\subsection{Estimating a single parameter}
\label{sec:parameter_estimation}

The first result of this section states that one can achieve a sample and time complexity tradeoff for the problem of testing if $|\theta(\cD)_1|$ is large or small.

\begin{theorem}[Testing a single parameter via loss comparisons]
    \label{thm:param_testing}
    There exists a mechanism $\cM$ such that for all learning tasks $(\Theta,\ell,\cD)$ satisfying \Cref{ass:id}, privacy parameters $\eps,\delta>0$, accuracy parameter $\alpha>0$, and subsampling probability $p\in(0,1/4)$, satisfies $(\eps,\delta)$-DP and the following accuracy guarantee:
    Fix $\rho \geq 2\sigma^2 + \alpha^2$ and $t>20\alpha\mu^{-1/2}$. If 
    \[
    n=\Omega\paren{\frac{\Nid(\alpha/2) + \Mid}{p} + \paren{\frac{\Kid(1+\alpha^2)}{\alpha^2} + \frac{\rho^2+\lambda^2}{\alpha^4}}\paren{\frac{p^2}{\eps^3}}\poly\log\frac{n\rho}{\alpha\delta\eps}},
    \] 
    then 
    \begin{enumerate}
        \item If $|\theta(\cD)_1| > 2t$ 
        then $\Pr[\cM^L(Z,\rho,t,\alpha,p) = \sign(\theta(\cD)_1)]\geq 1-\delta$.
        \item If $|\theta(\cD)_1| < t/2$ 
        then $\Pr[\cM^L(Z,\rho,t,\alpha,p) = 0]\geq 1-\delta$.
    \end{enumerate}
     Moreover $\cM$ runs in time $\exp\paren{O\paren{\frac{p}{\eps}\log\frac{1}{\delta}}}\cdot\poly\paren{\frac{\log1/\delta}{p}}$.
\end{theorem}

Our next result provides a similar tradeoff for the problem of estimating $\theta(\cD)_1$.
We first introduce some additional notation: Let $\Theta|_1$ denote the set $\set{w\in\R \mid \exists\theta\in\Theta : \theta_1=w}$ and let $\Pi=\set{\pi_i}_{i\in[\kappa]}$ be an interval partition of $\Theta|_1$. For an interval $\pi$ and value $w\in\R$, let $\Delta(w,\pi)=\min_{w'\in\pi} \abs{w-w'}$.

\begin{theorem}[Estimating a single parameter via loss comparisons]
    \label{thm:theta1_estimation}
    There exists a mechanism $\cM$ such that for all learning tasks $(\Theta,\ell,\cD)$ satisfying \Cref{ass:id}, interval partition $\Pi=\set{\pi_i}_{i\in[\kappa]}$ of $\Theta|_1$, privacy parameters $\eps,\delta>0$, failure probability $\beta\in(0,1)$, accuracy parameter $\alpha>0$, and subsampling probability $p\in(0,1/4)$, is $(\eps,\delta)$-DP and satisfies the following:
    Fix clipping parameter $\rho \geq 2\sigma^2 + \alpha^2$. If 
    \[
    n=\Omega\paren{\frac{\Nid(\alpha/2) + \Mid}{p} + \paren{\frac{\Kid(1+\alpha^2)}{\alpha^2} + \frac{\rho^2+\lambda^2}{\alpha^4}}\paren{\frac{p^2}{\eps^3}}\poly\log\frac{n\rho\kappa\log1/\delta}{\alpha\beta\eps}},
    \] 
    then
    
    \[
    \Pr_{\substack{Z\sim\cD^n \\ \pi\sim \cM^L(Z,\rho,\alpha)}}\brackets{ \Delta(\theta(\cD)_1,\pi) < 10\alpha\mu^{-1/2}}\geq 1 - \beta.
    \]

     Moreover $\cM$ runs in expected time $\exp\paren{O\paren{\frac{p}{\eps}\log\frac1\beta\log\frac{\kappa\rho}{\alpha}}}\cdot \kappa\poly\paren{\frac{\log\kappa/\delta\log1/\beta}{p}}$. 
\end{theorem}

One can construct a mechanism $\cM'$ that has the same sample complexity, accuracy, and privacy guarantees as in \Cref{thm:theta1_estimation}, but with expected runtime $\exp\paren{O\paren{\frac{p}{\eps}\log\frac{\kappa\rho}{\alpha\beta}}}\poly\paren{\frac{\kappa\log\kappa/\delta\log1/\beta}{p\beta}}$. These running times are incomparable in general; the main difference is that the mechanism in \Cref{thm:theta1_estimation} allows us to set $p$ such that we avoid a $\poly(1/\beta)$ dependence. See \cref{rem:alt_mechanism} for more details.

In order to state the key technical lemma we will use to prove \Cref{thm:theta1_estimation,thm:param_testing}, we introduce the following additional notation: For each interval $\pi$ let $L^{(\pi)}=\min_{w\in\pi} L^{(w)}$ and let $\wt L$ and $\wt L^{(\pi)}$ denote the functions $\frac 1p\cdot \clip_{\rho}[L]$ and $\frac 1p \cdot \clip_{\rho}[L^{(\pi)}]$ respectively, where $\clip_{\rho}[L]$ denotes the function $L$ with range clipped to $[0,\rho]$. Informally, \Cref{lem:subsampling_accuracy} states that if $\theta(\cD)_1$ is in $\pi$ and far from $\pi'$, then we can distinguish $\pi$ and $\pi'$ by comparing the losses on different subsamples of $Z$. Somewhat more formally, if $S,S'\sim\cS_p(Z)$ then $L^{(\pi)}(S) < L^{(\pi')}(S')$ with high probability.  

\begin{lemma}
        \label{lem:subsampling_accuracy}
        Fix a learning task $(\Theta,\ell,\cD)$ satisfying \Cref{ass:id}. Fix $p\in(0,\frac14),\alpha,\beta>0$, and set $\rho\geq 2\sigma^2 + \alpha^2$, and let $\Pi$ be an interval partition of size $\kappa$. There exists a function $h:\cZ^*\to\R$ such that if
        \[
        n=\Omega\paren{\frac{\Nid(\alpha/2) + \Mid}{p} + \paren{\frac{\Kid(1+\alpha^2)}{\alpha^2p} + \frac{\rho^2+\lambda^2}{\alpha^4p} }\log\paren{\frac{n\rho\kappa}{\beta}}^3}
        \]
        then with probability at least $1-\beta$ over $Z\sim\cD^n$ and $S\sim\cS_{p}(Z)$ the following holds for all $\pi\in\Pi$:
        \begin{enumerate}
            \item If $\Delta(\pi,\theta(\cD)_1)\geq 10\alpha\mu^{-1/2}$ then $\wt L^{(\pi)}(S) > h(Z) + 3\alpha^2$.
            \item If $\theta(\cD)_1\in\pi$ then $h(Z) < \wt L^{(\pi)}(S) < h(Z) + \alpha^2$. 
        \end{enumerate}
    \end{lemma}

In the remainder of the section we prove \Cref{thm:param_testing,thm:theta1_estimation}, and then prove \Cref{lem:subsampling_accuracy}. 

\subsubsection*{Proof of \texorpdfstring{\Cref{thm:param_testing}}{private parameter testing theorem}}
\begin{proof}
    Our testing mechanism uses \Cref{thm:avg-of-q} to estimate the loss when the first coordinate is close to zero, and when the first coordinate is far from zero. It then compares the two outcomes and outputs the result. 

    \begin{algorithm}[H]
    \ifnum\neurips=0
    \Statex \textbf{Input:} Dataset $Z\in\cZ^n$, query access to loss $L$, parameters $\eps,\delta,\rho,\alpha,p$, and threshold $t\in\Theta|_1$.
    \Statex \textbf{Output:} Bit $b\in\zo$.
    \fi
	\begin{algorithmic}[1]
    	\caption{\label{alg:param_est} Parameter tester $\cM^{L}$}
        \ifnum\neurips=1
        \Statex \textbf{Input:} Dataset $Z\in\cZ^n$, query access to loss $L$, parameters $\eps,\delta,\rho,\alpha,p$, and threshold $t\in\Theta|_1$.
        \Statex \textbf{Output:} Bit $b\in\zo$.
        \fi
        \State Set $\eps'\gets\frac{\eps}{3}$, and $\delta'\gets\frac{\delta}{6}$.
        \State Let $\cA$ denote \Cref{alg:avg-of-q} with privacy parameters $\eps',\delta'$, and accuracy parameter $\alpha^2$. 
        \State Set $\pi_0\gets(-t,t)$ and $\pi_1\gets [t,\infty)$ and $\pi_{-1}\gets (-\infty, -t]$. 
        \State For each $j\in\set{-1,0,1}$, let $y_j\gets \cA^{\wt L^{(\pi_j)}}(Z)$. %
        \Comment{We treat the output $\bot$ as $\infty$.}
        \State Return $\arg\min_{j\in\set{-1,0,1}}\set{y_j}$. 
	\end{algorithmic}
    \end{algorithm}

    \subparagraph{Privacy and runtime.} Fix neighboring datasets $Z$ and $Z'$. Our proof proceeds by arguing that there exists an event $E$ of probability at least $1-3\delta/6$ such that conditioned on $E$ the mechanism $\cM$ is $(\eps,3\delta/6)$-DP. We then apply \Cref{fact:conditional_DP} to see that $\cM$ is $(\eps,\delta)$-DP. Let $E$ be the event that the interleaving relationship in \Cref{lem:quantile_finder} fails to hold for some $\wt L^{(\pi_j)}$. By \Cref{lem:quantile_finder} and our setting of $\delta'$ we have $\Pr[E]\leq 3\delta/6$. Since each call to $\cA$ is $(\eps',\delta')$-DP conditioned on $E$, \Cref{fact:conditional_DP} implies that $\cM$ is $(\eps,\delta)$-DP. By \Cref{thm:med-of-q}, each call to $\cA$ has runtime $\exp\paren{O\paren{\frac{p}{\eps}\log\frac{1}{\delta}}}\poly\paren{\frac{\log1/\delta}{p}}$. 

    \subparagraph{Accuracy.} The proof of accuracy proceeds by analyzing the following bad event: %
    Let $E$ denote the event that the conclusion of \Cref{lem:subsampling_accuracy} fails to hold for some $\pi_j$ and some subsample drawn by $\cA$. Assume without loss of generality that $\theta(\cD)_1 > 2t$ (the cases where $|\theta(\cD)_1| < t/2$ and $\theta(\cD)_1<-2t$ are symmetric). Then, conditioned on $\overline{E}$, we have that $\abs{\wt L^{(\pi_1)}(S) - \wt L^{(\pi_1)}(S')}\leq \alpha^2$, and that $\wt L^{(\pi_1)}(S) + \alpha^2 < \wt L^{(\pi_b)}(S') - \alpha^2$ for each $b\in\set{-1,0}$ and every $S$ and $S'$ drawn by $\cA$. By the guarantee of \Cref{thm:avg-of-q}, we have that $\cA^{\wt L^{(\pi_1)}}(Z) < \cA^{\wt L^{(\pi_b)}}(Z)$ and hence the mechanism correctly outputs $1$. To complete the accuracy argument, it suffices to bound the probability of the event $E$. Recall that by \Cref{thm:avg-of-q} each $\cA$ draws $T=\exp\paren{O(\frac{p}{\eps}\log\frac{1}{\delta})}\poly\paren{\frac{\log1/\delta}{p}}$ subsamples. By \Cref{lem:subsampling_accuracy}, our setting of $n$, and the union bound over the $T$ subsamples drawn in each call to $\cA$ we have $\Pr[E]\leq \delta$. 
    \end{proof}

\subsubsection*{Proof of \texorpdfstring{\Cref{thm:theta1_estimation}}{private parameter estimation theorem}}

    \begin{proof}
    Below, we define a mechanism for estimating a single parameter. Our algorithm uses median-of-quantiles to privately evaluate the loss when the first coordinate is restricted to a candidate interval $\pi_i$, and then uses private selection from private candidates \cite{LiuT19} to select the candidate that achieved the best loss.\footnote{\Cref{thm:pspc} does not appear explicitly in \cite{LiuT19}; however, it is a straightforward corollary of their Algorithm 2 and an analysis of geometric random variables.}

    \begin{theorem}[Private selection from private candidates \cite{LiuT19}]
        \label{thm:pspc}
        Fix $\eps,\beta\in(0,1)$, and $\kappa\in\N$, and let $T_{pspc}(\kappa,\beta)=O(\frac \kappa\beta\log\frac \kappa\beta\log\frac1\beta)$. Suppose $\set{\cM_i}_{i\in[\kappa]}$ is a collection of $\eps$-DP mechanisms with real-valued outputs. There exists a $3\eps$-DP mechanism $\cM_{pspc}$ that on input $D$ calls mechanisms $\set{\cM_i(D)}_{i\in[\kappa]}$, and outputs $\arg\min_i\min_j\set{y^j_i}$ where $y^j_i$ is the output of the $\ord{j}$ call to $\cM_i$. Moreover, $\cM_{pspc}$ makes $O(\frac{\kappa}{\beta}\log\frac{\kappa}{\beta})$ calls in expectation, and with probability at least $1-\beta$ makes at least one and at most $T_{pspc}(\kappa,\beta)$ calls to each $\cM_i$.
    \end{theorem}

    Because our mechanism uses median-of-quantiles, it will be convenient to introduce a notation for the loss with discretized range. Let $f^{(\pi)}$ denote $\wt L^{(\pi)}$ with range rounded to $\alpha^2\cdot \Z$.

    \begin{algorithm}[H]
    \ifnum\neurips=0
        \Statex \textbf{Input:} Dataset $Z\in\cZ^n$, query access to loss $L$, interval partition $\Pi$, and parameters $\eps,\delta,\beta,\rho,\alpha,p$.
        \Statex \textbf{Output:} Interval $\pi\in\Pi$
    \fi
	\begin{algorithmic}[1]
		\caption{\label{alg:theta1_est} $\theta_1$ estimation mechanism $\cM^{L}(Z,\rho,\alpha)$}
        \ifnum\neurips=1
            \Statex \textbf{Input:} Dataset $Z\in\cZ^n$, query access to loss $L$, interval partition $\Pi$, and parameters $\eps,\delta,\beta,\rho,\alpha,p$.
            \Statex \textbf{Output:} Interval $\pi\in\Pi$
        \fi
        \State Let $c\gets\frac{1}{10}$, $t\gets \frac5c\log\frac1\beta$, $\eps'\gets\frac{\eps}{3t}$, and $\delta'\gets\frac{\delta}{\kappa}$.
        \State Sample the random seed $r$ used by \Cref{alg:quantile_finder} run with failure probability $\delta'$.
        \State For all $i\in[\kappa]$ let $\cA_i(Z,r)$ denote the median-of-quantiles mechanism (\Cref{alg:med-of-q}) with input $Z$, and random seed $r$ given as input to \Cref{alg:quantile_finder}, privacy parameters $\eps'$ and $\delta'$, failure probability $\poly(c/\kappa)$, subsampling parameter $p$, and query access to $f^{(\pi_i)}$. 
        \State For each $j\in[t]$ let $i_j$ denote the output of the $\ord{j}$ run of $\cM_{pspc}$ run with privacy parameter $\eps'$, failure probability $c/2$, and mechanisms $\set{\cA_i(Z,r)}_{i\in[\kappa]}$.
        \State Output the median interval from $\set{\pi_{i_j} \mid j\in[t]}$. 
	\end{algorithmic}
    \end{algorithm}

    We describe \Cref{alg:theta1_est} at a high level below: The algorithm first samples a random seed $r$ to give to the quantile-finder (\Cref{alg:quantile_finder}), which will be used to fix the subsamples drawn across all calls to the $\cA_i$'s since \Cref{alg:quantile_finder} uses the same random seed each time; however, we do not fix the randomness used for the exponential mechanism in each call to $\cA_i$, as this is essential for privacy to hold. For each random seed $r$, the algorithm defines the mechanism $\cA_i(Z; r)$ as the median-of-quantiles mechanism run with the fixed random seed $r$ given to the quantile-finder, and query access to the function $f^{(\pi_i)}$---that is, the clipped and appropriately normalized loss corresponding to the $\ord{i}$ candidate $\pi_i$. Then, the algorithm runs $t$ iterations of $\cM_{pspsc}$, the selection mechanism in \Cref{thm:pspc}, with a small constant failure probability. The basic idea is that most of the runs of $\cM_{pspsc}$ will succeed and output a candidate $\pi_{i_j}$ such that $\Delta(\pi_{i_j}, \theta(\cD)_1)$ is small, and thus this will hold for the median interval in $\set{\pi_{i_j}}$ as well.
     
     \subparagraph{Privacy.} Fix neighboring datasets $Z$ and $Z'$. Our proof proceeds by arguing that there exists an event $E$ of probability at least $1-\delta$ over the choice of random seed, such that conditioned on $E$ each of the $t$ calls to $\cM_{pspc}$ is $3\eps'$-DP. We then apply \Cref{fact:conditional_DP} (DP with high probability) to complete the proof.

     By \Cref{lem:quantile_finder} and \Cref{thm:med-of-q}, for each $i\in[\kappa]$, there exists an event $E_i$ (the interleaving in \Cref{claim:quantile_interleaving}) that occurs with probability at least $1-\delta'$ over the choice of random seeds $r$ and $r'$ (the random seeds chosen on input $Z$ and $Z'$) such that conditioned on $E_i$ we have $\cA_i(Z,r)\approx_{\eps'} \cA_i(Z',r')$. Let $E=\bigcap_{i\in[\kappa]} E_i$. By the union bound $\Pr[E]\geq 1-\delta$, and moreover, conditioned on $E$ we have $\cA_{i}(Z,r)\approx_{\eps'}\cA_i(Z',r')$ for all $i\in[\kappa]$. Thus, by \Cref{thm:pspc} we have that $\cM_{pspc}(Z)\approx_{3\eps'}\cM_{pspc}(Z')$ conditioned on $E$. Basic composition and our setting of $\eps'$ implies that conditioned on $E$ we have $\cM(Z)\approx_{\eps}\cM(Z')$. Since $E$ occurs with probability at least $1-\delta$ we can apply \Cref{fact:conditional_DP} to see that $\cM$ is $(\eps,\delta)$-DP.

     \subparagraph{Runtime.} Since the range of $f^{(\pi)}$ has size $\rho/\alpha^2$ for each $\pi\in\Pi$, \Cref{thm:med-of-q} implies that each call to $\cA_i$ has runtime $\exp\paren{O(\frac{p}{\eps}\log\frac1\beta\log\frac{\kappa\rho}{\alpha})}\poly\paren{\frac{\log\kappa/\delta}{p}}$. Since $\cM_{pspc}$ makes $O(\kappa\log\kappa)$ calls in expectation, mechanism $\cM$ runs in expected time $\exp\paren{O(\frac{p}{\eps}\log\frac1\beta\log\frac{\kappa\rho}{\alpha})}\cdot\kappa\poly\paren{\frac{\log\kappa/\delta\log1/\beta}{p}}$.

     \subparagraph{Accuracy.} %

    We leverage \Cref{thm:med-of-q,lem:subsampling_accuracy,lem:quantile_finder} in order to analyze the following bad events: For each $j\in[t]$, let $E_j$ denote the event that in the $\ord{j}$ iteration of $\cM_{pspc}$, there exists $i\in[\kappa]$ such that some execution of $\cA_i$ fails to satisfy the accuracy guarantee of \Cref{thm:med-of-q}. Assume $\theta(\cD)_1\in\pi_{i^*}$ and let $E'_j$ denote the event that the $\ord{j}$ iteration of $\cM_{pspc}$ does not call $\cA_{i^*}$ or makes more than $T_{pspc}(\kappa,c)$ calls to the $\cA_i$'s. Let $E$ denote the event that $\frac1t\sum \Ind\brackets{E_j\cup E'_j}\geq 3c$. Additionally, let $B$ denote the event that there exists some $i\in[\kappa]$ and some subsample $S$ drawn by \Cref{alg:quantile_finder} such that $\wt L^{(\pi_i)}(S)$ fails to satisfy the conclusion of \Cref{lem:subsampling_accuracy}.

    Observe that if $B$, $E_j$ and $E'_j$, do not occur, then for every $i$ with $\Delta(\pi_{i}, \theta(\cD)_1)\geq \alpha\mu^{-1/2}$ we have $f^{(\pi_i)}(S) > h(Z) + 2\alpha^2$ (since $f^{(\pi)}$ and $\wt L^{\pi}$ differ by at most $\alpha^2$) for every subsample $S$ drawn by \Cref{alg:quantile_finder}, and thus $\cA_i(Z;r)>h(Z) + 2\alpha^2$. On the other hand, we have $f^{(\pi_{i^*})}(S) < h(Z) + 2\alpha^2$ and thus $\cA_{i^*}(Z;r) < h(Z) + 2\alpha^2$. It follows that the index $i_j$ returned by the $\ord{j}$ iteration of $\cM_{pspc}$ satisfies $\Delta(\pi_{i_j},\theta(\cD)_1) < \alpha\mu^{-1/2}$. Moreover, if the neither $B$ nor $E$ occur, then at least a $1-3c$ fraction of the $\pi_{i_j}$ satisfy the above guarantee, and hence the median $\pi_{i_j}$ does as well. Thus, it suffices to bound the probability of $B$ and $E$. 

    First, we bound the probability of $B$. By \Cref{lem:quantile_finder}, the total number of subsamples drawn by \Cref{alg:quantile_finder} is $T=\exp\paren{O(\frac{p}{\eps}\log\frac1\beta\log\frac{\kappa\rho}{\alpha})}\poly\paren{\frac{\log\kappa/\delta}{p}}$. By \Cref{lem:subsampling_accuracy} and the union bound over the $S_1,\dots, S_T$ subsamples, if 
    \[
    n=\Omega\paren{\frac{\Nid(\alpha/2) + \Mid}{p} + \paren{\frac{\Kid}{\alpha^2p} + \frac{\Kid}{p} + \frac{\rho^2+\lambda^2}{\alpha^4p}}\log\paren{\frac{n\rho\kappa\cdot T}{\beta}}^3}
    \]
    then $\Pr[B]\leq \beta/2$. 

    Next, we bound the probability of $E$. By \Cref{thm:pspc} and our setting of parameters in \Cref{alg:theta1_est} we have that $\Pr[E'_j]\leq c$. Next, by the law of total probability we have $\Pr[E_j]\leq \Pr[E_j\mid \overline{E'_j}] + \Pr[E'_j]$. If $E'_j$ does not occur, then $\cM_{pcpc}$ makes at most $T_{pspc}(\kappa,c)$ calls to any $\cA_i$. Since $\cA_i$ is executed with failure probability $\poly(c/\kappa)$, the $T_{pspc}(\kappa,c)$ calls implies that $\Pr[E_j\mid \overline{E'_j}]\leq c$. Since $\Pr[E_j\cup E'_j]\leq 2c$ for each $j\in[t]$, a Chernoff bound implies that $\Pr[E]\leq e^{-O(ct)}$ which is at most $\beta/2$ by our setting of $t$. 

    Thus $\Pr[B\cup E]\leq \beta$, which completes the proof of accuracy.
    \end{proof}

    \begin{remark}\label{rem:alt_mechanism}
    We briefly explain how to construct a mechanism matching the guarantees in \cref{thm:med-of-q} but with expected runtime $\exp\paren{O\paren{\frac{p}{\eps}\log\frac{\kappa\rho}{\alpha\beta}}}\poly\paren{\frac{\kappa\log\kappa/\delta\log1/\beta}{p\beta}}$.
    Let $\cM'$ be defined as \Cref{alg:theta1_est}, but set $c=\beta$ and $t=1$. Essentially, $\cM'$ executes $\cM_{pspc}$ (see \Cref{thm:pspc}) once with failure probability $\beta/2$ and outputs the result. The privacy analysis is the same as that of \Cref{alg:theta1_est}. The runtime follows from \Cref{thm:pspc,thm:med-of-q}, and our setting of $c=\beta$ and $t=1$. The accuracy argument is identical to that of \Cref{alg:theta1_est}, except that we no longer need to consider multiple runs of $\cM_{pspc}$.   
    \end{remark}

\subsubsection*{Proof of \texorpdfstring{\Cref{lem:subsampling_accuracy}}{subsampling accuracy lemma}}

\begin{proof}
We argue that there exists some $h(Z)$, such that with high probability over $S\sim\cS_p(Z)$, if $\Delta(\pi,\theta(\cD)_1) > 10\alpha\mu^{-1/2}$, then $\wt L^{(\pi)}(S) > h(Z) + 3\alpha^2$; while if $\theta(\cD)_1\in\pi$, then $h(Z) < \wt L^{(\pi)}(S) < h(Z) + \alpha^2$. 

\begin{claim}
    \label{claim:L_bounds}
    If $\rho\geq 2\sigma^2 + \alpha^2$ and
    \[
    n=\Omega\paren{\frac{\Nid(\alpha/2) + \Mid}{p} + \frac{\Kid\log1/\beta}{\alpha^2p} +  \frac{\Kid\log1/\beta}{p} + \frac{\rho^2\log1/\beta}{\alpha^4p}},
    \]
    then for all intervals $\pi$, the following holds:
    \begin{enumerate}
        \item If $\Delta(\theta(\cD)_1,\pi)\geq 10\alpha\mu^{-1/2}$, then
        \[
        \Pr_{Z,S}\brackets{\wt L^{(\pi)}(S)\geq \wt L(S) + 5\alpha^2} \geq 1-\beta.
        \]
        \item If $\theta(\cD)_1\in\pi$, then
        \[
        \Pr_{Z,S}\brackets{\wt L^{(\pi)}(S)\leq \wt L(S) + \alpha^2} \geq 1-\beta.
        \]
    \end{enumerate}
\end{claim}

While \Cref{claim:L_bounds} provides a guarantee in terms of $S$, our goal is to provide a bound in terms of $Z$. To obtain the desired bound, we will use \Cref{claim:L_concentration}, which states that $\wt L(S)$ concentrates around its mean---this allows us to set $h(Z)\approx \Ex_{S\sim\cS_p(Z)}[L(S)]$ and apply a union bound over the subsamples.   
  
\begin{claim}
    \label{claim:L_concentration}
    For all $n=\Omega\paren{\frac{\Mid + \Kid\log (n\rho/\beta)}{p}}$, we have
    \[
        \Pr_{\substack{Z\sim\cD^n \\ S\sim\cS_p(Z)}}\brackets{\abs{\wt L(S) - \Ex_{S'\sim\cS_p(Z)}[\wt L(S')]}\gtrsim \frac{\lambda(\log(n\rho / \beta))^{3/2}}{\sqrt{np}}} \leq \beta.
    \]
\end{claim}

We defer the proofs of \Cref{claim:L_bounds,claim:L_concentration} for later, and complete the proof of \Cref{lem:subsampling_accuracy}. Let $\Pi^{-}\subset\Pi$ denote the set of intervals $\pi$ such that $\Delta(\theta(\cD)_1,\pi)\geq 10\alpha\mu^{-1/2}$ and let $\pi^*\in\Pi$ be the interval containing $\theta(\cD)_1$. Define $E_1$ as the event that $\wt L^{(\pi)}(S)< \wt L(S) + 5\alpha^2$ for some $\pi\in\Pi^{-}$, or $\wt L^{(\pi^*)}(S) > \wt L(S) + \alpha^2$. Define $E_2$ as the event that $\abs{\wt L(S) - \Ex_{S'\sim\cS_p(Z)}[\wt L(S')]} > \frac{c\lambda(\log(n\rho \kappa/ \beta))^{3/2}}{\sqrt{np}}$. %
Observe that conditioned on $\overline{E_1\cup E_2}$ we have
\[
    \wt L^{(\pi)}(S) \geq \Ex[\wt L(S)] + 5\alpha^2 - \frac{c\lambda(\log(n\rho \kappa/ \beta))^{3/2}}{\sqrt{np}}
\]
for all $\pi\in\Pi^-$, and
\[
    \Ex[\wt L(S)] - \frac{c\lambda(\log(n\rho \kappa/ \beta))^{3/2}}{\sqrt{np}} \leq \wt L^{(\pi^*)}(S) \leq \Ex[\wt L(S)] + \alpha^2 + \frac{c\lambda(\log(n\rho \kappa/ \beta))^{3/2}}{\sqrt{np}}.
\]
Let $h(Z) = \Ex[\wt L(S)] - \frac{c\lambda(\log(n\rho \kappa/ \beta))^{3/2}}{\sqrt{np}}$. By our setting of $n$, we have that for all $\pi\in\Pi^-$, we have $\wt L^{(\pi)}(S) > h(Z) + 5\alpha^2$ and $\wt L^{(\pi^*)}(S)\in [h(Z), h(Z) + 3\alpha^2/2]$. 

Rescaling $\alpha\gets \alpha/\sqrt{1.5}$ and repeating the proof gives 
\[
\wt L^{(\pi)}(S) > h(Z) +3\alpha^2 ~~\text{and}~~ h(Z) < \wt L^{(\pi^*)}(S) < h(Z) + \alpha^2.
\]
By  \Cref{claim:L_bounds,claim:L_concentration} and the union bound we have $\Pr[E_1\cup E_2]\leq \beta$, which yields the result.

\end{proof}

In the remainder of the section we prove  \Cref{claim:L_bounds,claim:L_concentration}.

\begin{proof}[Proof of \Cref{claim:L_bounds}]
First set $\alpha_0\gets\alpha\mu^{-1/2}$. Let $E_1$ be the event that $m=|S|\not\in\brackets{pn\pm \sqrt{cpn\log1/\beta}}$, let  $E_2$ be the event that $|\theta(S)_1 - \theta(\cD)_1|\geq \alpha_0/2$, let $E_3$ be the event that $nL(S)/m>2\sigma^2$, and let $E_4$ be the event that some $w\in\R$ violates the inequality in \Cref{ass:tssc}. First, we show that the conclusion holds conditioned on $\overline{\bigcup_{i\in[4]} E_i}$, and then we show that $\Pr[\bigcup_{i\in[4]} E_i]\leq \beta$. 

Let $\alpha_1=10\alpha_0$ and consider the case where $\Delta(\theta(\cD)_1,\pi)\geq \alpha_1$. By the triangle inequality and the fact that $|\theta(S)_1 - \theta(\cD)_1|\geq \alpha_0/2$, we have $\inf_{w\in\pi} |w-\theta(S)_1|\geq \alpha_1/2$. Now, if $m=|S|$ then,
\[
\frac nm\cdot L^{(\pi)}(S) \geq \frac nm\cdot L(S) + \frac{\mu\alpha_1^2}{8} \geq  \frac 1p \clip_{\rho}[L(S)] - \frac{\rho}{p}\sqrt{\frac{c\log1/\beta}{np}} + \frac{\mu\alpha_1^2}{8},
\]
where the first inequality follows from \Cref{ass:tssc}, and the second inequality follows from a similar analysis as \Cref{claim:approx_loss}. Next, we show that the same holds for $\wt L$---that is, when $L$ is clipped and normalized by $\frac {1}{p}$ (instead of $\frac nm$). First, observe that if $L^{(\pi)}(S)\leq \rho$ then $\clip_{\rho} L^{(\pi)}(S) = L^{(\pi)}(S)$, and thus, %
\[
\wt L^{(\pi)}(S) \geq \wt L(S) - \frac{\rho}{p}\sqrt{\frac{c\log1/\beta}{np}} + \frac{\mu\alpha_1^2}{8} \geq \wt L(S) + 5\mu\alpha_0^2.
\]
To see why the same holds when $L^{(\pi)}(S)\geq \rho$ (and thus, $\wt L^{(\pi)}(S)\geq \rho/p$), recall that %
$nL(S)/m \leq 2\sigma^2$, %
and thus, by our assumption that $\rho\geq 2\sigma^2 + \alpha_0^2\mu$, and conditioned on event $\overline E_1$, we have 
\[
\wt L(S)\leq\frac{nL(S)}{m}\paren{1+\sqrt{\frac{c\log1/\beta}{np}}} \leq   \rho - \frac{\mu\alpha_1^2}{8} + \sigma^2\sqrt{\frac{c\log1/\beta}{np}}.
\]
Multiplying the right hand side by $\frac 1p$ yields
\[
\wt L^{(\pi)}(S) \geq \frac{\rho}{p} \geq \wt L(S) - \frac{\rho}{p}\sqrt{\frac{c\log1/\beta}{np}} + \frac{\mu\alpha_1^2}{8}\geq \wt L(S) + 5\mu\alpha_0^2.
\]

Next, we consider the case where $\Delta(\theta(\cD)_1,\pi)=0$. Let $w=\theta(\cD)_1$. Then conditioned on $\overline{E_4}$, we can apply \Cref{ass:tssc} to obtain
\[
\wt L^{(\pi)}(S)\leq \frac nm L^{(w)}(S) + \frac{\rho}{p}\sqrt{\frac{c\log1/\beta}{np}} \leq \frac nm L(S) + \frac{\mu\alpha_0^2}{2} + \frac{\rho}{p}\sqrt{\frac{c\log1/\beta}{np}}.
\]
Since $L(S)<\rho$ we have $L(S)=\clip_{\rho}[L(S)]$, and thus, 
\[
\wt L^{(\pi)}(S)\leq \wt L(S) + \frac{\mu\alpha_0^2}{2} + \frac{2\rho}{p}\sqrt{\frac{c\log1/\beta}{np}}\leq \wt L(S) + \mu\alpha_0^2.
\]

We complete the proof of \Cref{claim:L_bounds} by bounding the probability of $\bigcup_{i\in[4]} E_i$. By the argument in \Cref{claim:approx_loss}, we have $\Pr[E_1]\leq \beta/4$. Since $n=\Omega\paren{\frac{\Nid(\alpha/2)}{p} + \frac{\Kid\log1/\beta}{\alpha^2p}}$, \Cref{ass:param_concentration} implies that $\Pr[E_2]\leq \beta/4$. Similarly, by our setting of $n=\Omega\paren{\frac{\Mid}{p} + \frac{\Kid\log1/\beta}{p}}$ and \Cref{ass:tb,ass:tssc}, we have $\Pr[E_3]\leq \beta/4$ and $\Pr[E_4]\leq \beta/4$. By the union bound, $\Pr\brackets{\bigcup_{i\in[4]} E_i}\leq \beta$ holds, which completes the proof.  
\end{proof}

\begin{proof}[Proof of \Cref{claim:L_concentration}]
A key tool in our proof is a type of bounded-differences inequality that applies when the difference between function values on neighboring points is bounded with high probability. While the most general version of the statement appears in \cite{Warnke15}, we use a simplified form for binary random variables and where the differences are bounded by the same constant in every coordinate. At a high level, we use the inequality to argue that the function $\wt L$ concentrates around its expectation at a fast rate when $S\sim\cS_p(Z)$. In particular, this rate is independent of the rate at which the empirical loss concentrates around the population loss (e.g., as in \Cref{ass:loss_concentration}, which we do not use here) and only depends on \Cref{ass:loo}.   

    \begin{theorem}[Typical bounded differences inequality \cite{Warnke15}]\label{thm:TL}
        Let $X=(X_1, \dots, X_N)$ be a family of independent random 
        variables taking values in $\zo$, where $\Pr[X_i=1]=p$ for all $i\in[N]$. Let $\Gamma \subset \zo^N$ be an event and 
        assume that the function $f:\zo^N \to \R$ 
        satisfies the following \emph{typical Lipschitz condition}:
        \begin{itemize}
        \item[(TL)] There are numbers $a$ and 
        $b$ with $a \leq b$ such that whenever 
        $x,\tilde{x} \in \zo^N$ differ in exactly one coordinate, we have
        \begin{equation*}
        |f(x)-f(\tilde{x})| \; \le \; 
        \begin{cases}
        		a & \;\text{if $x \in \Gamma$,}\\ 
        		b & \;\text{otherwise.} 
        \end{cases}
        \end{equation*}
        \end{itemize} 
        For any $\gamma\in(0,1]$ there is an event $\cB=\cB(\Gamma,\gamma)$
        satisfying 
        \begin{equation*}
        \Pr(\cB) \leq N\gamma^{-1} \cdot \Pr(X \notin \Gamma) 
        \quad \text{and} \quad \neg\cB \subseteq \Gamma , 
        \end{equation*}
        such that for $\mu=\Ex f(X)$, $c=\gamma (b-a)$ and any 
        $t \ge 0$ we have 
        \begin{equation*}
        \Pr\brackets{\abs{f(X)-\mu} \ge t \text{ and } \overline\cB} \le 2\exp\left(-\frac{t^2}{2Np(1-p)(a+c)^2 + 2(a+c)t/3}\right) . 
        \end{equation*}
    \end{theorem}

    At a high level, we show that for a typical $Z$, one can apply \Cref{thm:TL} to $\wt L$ with the domain $\set{S\subset Z}$. We use the $\lambda$-leave-one-out property (\Cref{ass:loo}) to argue that $\wt L$ has bounded differences with high probability under $\cS_p(Z)$---that is, it satisfies the typical Lipschitz condition, and we use the definition of $\wt L$ to argue that the differences are always at most $\rho/p$.  

    In order to leverage \Cref{thm:TL}, we argue that for most $Z$, there exists a set $\Gamma$ of subsets $S\subset Z$ with probability mass at least $1-\beta$ over $\cS_p(Z)$ such that for all $S\in\Gamma$ the quantity $\sup_{i\in[n]}\frac 1p \cdot \abs{L(S^{i\gets z_i})-L(S_{-i})}$ is small. Let $c>0$ be a sufficiently large constant, and for all $Z$ and $\beta\in(0,1)$ define the set
    \[
        \Gamma(Z,\beta)=\set{S\subset Z \mid \sup_{i\in[n]}\frac 1p\cdot \abs{L(S^{i\gets z_i})-L(S_{-i})}\leq \frac{c\lambda\log(n/\beta)}{np}}.
    \]
    Next, we argue that for most $Z$, the set $\Gamma(Z,\beta)$ has probability mass at least $1-\beta$ under $\cS_p(Z)$. We say a dataset $Z$ is \emph{$\beta$-typical} if $\Pr_{S\sim\cS_p(Z)}\brackets{S\in\Gamma(Z,\beta)}\geq 1-\beta$. To apply \Cref{ass:loo}, we first relate the difference $\frac 1p\cdot\abs{L(S^{i\gets z_i})-L(S_{-i})}$ for a random subset $S\sim\cS_p(Z)$ to the difference $\frac{n}{|S^{i\gets z_i}|}\cdot\abs{L(S^{i\gets z_i})-L(S_{-i})}$---that is, where the losses are normalized appropriately for samples of size $|S^{i\gets z_i}|$. For all $S\subset Z$, where $S^{i\gets z_i}$ has size $m = pn - r$ for some $r\in\Z$, we have
    \begin{align*}
        \frac 1p\cdot \abs{L(S^{i\gets z_i}) - L(S_{-i})} 
        & = \frac{n}{m+r}\abs{L(S^{i\gets z_i}) - L(S_{-i})}\\
        & = n\paren{\frac 1m - \frac {r}{m(m+r)}}\cdot \abs{L(S^{i\gets z_i}) - L(S_{-i})}\\
        & = \frac nm\cdot\abs{L(S^{i\gets z_i}) - L(S_{-i})}\paren{1-\frac{r}{m+r}}
    \end{align*}
    By the arguments in \Cref{claim:approx_loss}, we have $r\in[\pm \sqrt{cpn\log1/\beta}]$ with probability at least $1-\beta^2/2$. Conditioned on this event, we have
    \[
    \frac 1p\cdot \abs{L(S^{i\gets z_i}) - L(S_{-i})}\leq  \frac{2n}{m}\cdot\abs{L(S^{i\gets z_i}) - L(S_{-i})}.
    \] 
    And thus, applying \Cref{ass:loo} and our setting of $n$ in the statement of \Cref{claim:L_concentration} yields
    \[
    \Pr_{Z,S}\brackets{\frac 1p\cdot \abs{L(S^{i\gets z_i}) - L(S_{-i})} \geq \frac{2c\lambda\log1/\beta}{m}}\leq \Pr_{Z,S}\brackets{\frac{n}{m}\cdot \abs{L(S^{i\gets z_i}) - L(S_{-i})} \geq \frac{c\lambda\log1/\beta}{m}}\leq \beta^2/2. 
    \]
    By the union bound over all $i\in[n]$ and the event that $r\not\in [\pm \sqrt{cpn\log1/\beta}]$, we obtain that $\Pr_{Z,S}\brackets{S\in \Gamma(Z,\beta)}\geq 1-\beta^2$. Now, let $E$ be the event that $Z$ is $\beta$-typical. Then by the law of total probability,
    \[
    1-\beta^2\leq \Pr[E] + (1-\beta)(1-\Pr[E]),
    \]
    and thus $Z$ is $\beta$-typical with probability at least $1-\beta$.

    We complete the proof of \Cref{claim:L_concentration} by applying \Cref{thm:TL} to $\wt L$ and $\beta$-typical $Z$. Let $a=\frac{c\lambda\log(n/\beta)}{np}$ and let $b=\rho/p$ and $\gamma=a/(b-a)$. Then by \Cref{thm:TL}, there is an event $\cB$ such that $\Pr[\cB]\leq n\gamma^{-1}\beta$ and 
    \[
    \Pr_{S\sim\cS_p(Z)}\brackets{\abs{\wt L(S) - \Ex_{S\sim\cS_p(Z)}[\wt L(S)]}\geq t\wedge \overline\cB} \leq 2\exp\paren{-\frac{t^2}{8np(1-p)a^2 + 4at/3}}.
    \]
    Since $p\leq \frac 14$, $\frac{\log(1/\beta)}{np}\leq 1$, and $\Pr[\cB]\leq n\gamma^{-1}\beta$, there exists an absolute constant $c>0$, such that if  $t=a\sqrt{cpn\log\frac1\beta}$, then
    \[
    \Pr_{S\sim\cS_p(Z)}\brackets{\abs{\wt L(S) - \Ex_{S\sim\cS_p(Z)}[\wt L(S)]}\geq t} \leq 2\exp\paren{-\log\frac1\beta} + n\gamma^{-1}\beta.
    \]
    Since $Z$ is $\beta$-typical with probability at least $1-\beta$, the union bound gives
    \[
    \Pr_{\substack{Z\sim\cD^n \\ S\sim\cS_p(Z)}}\brackets{\abs{\wt L(S) - \Ex_{S\sim\cS_p(Z)}[\wt L(S)]}\geq t} \leq 2n\gamma^{-1}\beta.
    \] 
    Since $\frac n\gamma \leq n^2\rho$, setting $\beta'\gets\frac{\beta}{2n^2\rho}$ and repeating the proof with $\beta'$ yields the result.
\end{proof}

\subsection{When an exact oracle for the loss is unavailable}

In this section, we explain how the mechanisms in \Cref{sec:loss_estimation,sec:parameter_estimation} can be realized with an oracle $\appxL$ that satisfies $\appxL(Z) \leq L(Z) + \gamma$ for some sufficiently small $\gamma>0$.
In this case, one cannot directly apply our theorems since $\appxL$ need not be monotone. In the two paragraphs below, we discuss a simple method for enforcing monotonicity while maintaining the accuracy guarantees. 

\paragraph{Enforcing monotonicity.} Instead of applying our mechanisms directly to $\appxL$, we instead use a slightly modified function $f$ defined by $f(Z) = \appxL(Z) + \gamma|Z|$. It is not hard to see that if $S\subsetneq Z$, then $f(Z) - f(S) = \appxL(Z) - \appxL(S) + \gamma\paren{|Z|-|S|}$. Since $\appxL(Z)\geq L(Z)$ and $\appxL(S)\leq L(S) + \gamma$, we have $f(Z) - f(S) \geq L(Z) - L(S) - \gamma + \gamma\paren{|Z| - |S|}$. Since $L$ is monotone and $|Z| > |S|$, we have $f(Z) - f(S) \geq 0$, and thus $f$ is monotone.

\paragraph{Maintaining accuracy.} While the above transformation enforces monotonicity, it is not immediately obvious how it affects the accuracy guarantees of our mechanisms. Fortunately, an additional simple transformation to $\appxL$ suffices to ensure that the accuracy guarantees are preserved up to a constant factor. Consider the following transformation: set $\gamma\approx \alpha/\sqrt{np}$ and define $f(Z) = \appxL(Z) + \gamma|Z| - \gamma np$. By the same argument as above, $f$ is monotone. To see how this affects our accuracy guarantees, observe that the mechanisms of \Cref{thm:med-of-q,thm:avg-of-q} only query $f$ on datasets of size $np \pm \sqrt{np}$ with high probability. Moreover, for a dataset $S\in\cZ^n_\bot$ of size $|S| = np \pm \sqrt{np}$, we have $f(S) \approx \appxL(S) \pm \gamma\sqrt{np} \approx L(S) \pm \alpha$. Thus, for a sufficiently small setting of $\alpha$, the accuracy arguments in \Cref{sec:loss_estimation,sec:parameter_estimation} will not be affected.

\section{Query complexity lower bound for monotone statistics}
\label{sec:query_lb}

In this section we prove a lower bound on the query complexity of any algorithm that is private for monotone functions and satisfies a weak accuracy guarantee. It will be convenient for the analysis of our lower bound to switch from treating datasets as $n$-tuples to treating datasets as sets of $n$ distinct elements. For a set $\cZ$, let $\cZ^*$ denote the set of finite subsets of $\cZ$. We say two datasets are neighbors if they have the same size and differ in exactly one element. Note that this is without loss of generality since we can always define the functions in our hard instances to treat $n$-tuples as sets, however, in the interest of simplifying the notation, we assume that the datasets are sets instead. We assume the dataset size is public, and we say two datasets $Z$ and $Z'$ are neighbors if they are the same size and differ in exactly one element: formally, $|Z|=|Z'|$ and $|Z\cap Z'| = |Z|-1$.

Before presenting our result, we define a weak accuracy guarantee that the mechanism should satisfy in order for our lower bound to apply. 

\begin{definition}[$(\nu,N)$-constant, $(p,\beta)$-weakly accurate]
    A function $f:\cZ^*\to\R$ is \emph{$(\nu,N)$-constant under a distribution $\cD$} if for all $n\geq N$, we have
    \[
\Pr_{Z\sim\cD^n}\brackets{f(Z)=\nu\mid |Z| \geq N}=1.
    \]
    An algorithm $\cA$ is \emph{$(p,\beta)$-weakly accurate for a function $f$} if for all distributions $\cD$ such that $f$ is $(\nu,N)$-constant under $\cD$, algorithm $\cA$ satisfies
    \[
    \Pr_{Z\sim\cD^{N/p}}\brackets{\abs{\cA^f(Z)-\nu}\geq 1/2 \mid |Z|=\frac Np}\leq \beta.
    \]
\end{definition}

In other words, the condition ``$(\nu,N)$-constant under $\cD$'' means that once the sample size reaches $N$, the value of $f(Z)$ is no longer random: for every $n\ge N$, a draw $Z\sim \cD^n$ satisfies $f(Z)=\nu$ with probability $1$ (so long as $Z$ has enough unique elements).
An algorithm $\cA$ is ``$(p,\beta)$-weakly accurate for $f$'' if, for every distribution $\cD$ where $f$ becomes such an eventual constant $\nu$, the algorithm can recover $\nu$ from only $N/p$ fresh samples---with probability at least $1-\beta$ over $Z\sim \cD^{N/p}$, its output is within additive error $1/2$ of $\nu$, as long as all elements of $Z$ are distinct. 

\begin{theorem}[Query complexity lower bound]
    \label{thm:query_lb}
    Fix a sufficiently small constant $c>0$. Let $\cZ=\N$, and fix range size $\kappa\in\N$, failure probability $\beta\in(0,\frac 14)$, privacy parameters $\eps\in(0,1)$ and $\delta\in(0,\eps\beta^2)$, and sampling parameter $p\in(0,1)$. Let $$\cF=\set{f:\cZ^*\to[\kappa]\mid \text{$f$ is monotone}} \, .$$ Suppose mechanism $\cM^f$ is $(\eps,\delta)$-DP and $(p,\beta)$-weakly accurate for all $f\in\cF$. Let $$\tau = \frac{c}{\eps}\log\min\paren{\frac\kappa\beta, \frac\eps\delta} \, .$$
    If $\tau \leq n$, then $\cM^f$ has expected query complexity $\exp\paren{{\Omega\paren{p\tau}}}$ for some $f\in\cF$ and $Z\in\cZ^n$.  
\end{theorem}

As we will see in the proof, the parameter $\tau$ controls the distance between the ``bad'' datasets we use to show a privacy violation. 

\begin{proof} The proof proceeds by constructing a family of monotone functions $\cF$, and two families of distributions $\set{\cD_0}$ and $\set{\cD_1}$ such that every $f\in\cF$ is $(y_0,N)$-constant under some $\cD_0$ and $(y_1,N)$-constant under some $\cD_1$. We then argue that for every mechanism $\cM$ with query complexity $\exp\paren{o\paren{p\tau}}$ that is $(p,\beta)$-weakly accurate for $\cF$, there exist neighboring inputs $X,X'$ and a function $f\in\cF$, such that $\cM^f(X)\not\approx_{\eps,\delta}\cM^f(X')$, i.e., privacy is violated.   

    Let $G=[M]$ for some sufficiently large $M$, and fix $t$ such that $n^2\ll t\ll M$ (we will set $M$ and $t$ later in the proof). To sample distributions $\cD_0$ and $\cD_1$, sample $B_0=\set{b^0_1,\dots,b^0_t}$ uniformly without replacement from $G$, and $B_1=\set{b^1_1,\dots,b^1_t}$ uniformly without replacement from $G\setminus B_0$. Let $\cD_0$ and $\cD_1$ be the uniform distributions over $B_0$ and $B_1$ respectively. To sample a function $f\in\cF$, sample $B_0$ and $B_1$, along with integers $y_0\in[\kappa/2]$ and $y_1\in[\kappa]\setminus[\kappa/2]$. Let $n=N/p$ and $\tau=\frac{1}{2\eps}\ln \gamma \min\paren{\frac\kappa\beta, \frac\eps\delta}$ for a sufficiently small absolute constant $\gamma\in(0,1)$.  Define 
    \[
    f(Z) = 
    \begin{cases}
    0 & \text{if}~~ |Z| < N\\
    y_0 & \text{else if}~~ Z\subset B_0 \\
    y_1 & \text{else if}~~ Z\subset B_0\cup B_1\\
    \kappa & \text{otherwise}%
    \end{cases}
    \]
    It is not hard to see that $f$ is monotone, and that it is $(y_0,N)$-constant and $(y_1,N)$-constant under the corresponding $\cD_0$ and $\cD_1$.

    To prove the privacy violation, we will sample sets $X$ and $X'$ as follows: Let $X=\set{x_1,\dots,x_n}$ where $x_1,\dots,x_\tau$ are sampled uniformly from $B_1$ without replacement, and $x_{\tau+1},\dots,x_n$ are sampled uniformly from $B_0$ without replacement. Let $X'=\set{x'_1,\dots, x'_{\tau},x_{\tau+1},\dots,x_n}$ where $x'_1,\dots,x'_{\tau}$ are sampled uniformly from $B_0\setminus X$ without replacement. Note that $X'\subset B_0$ and that $X$ and $X'$ are at distance $\tau$.

    \Cref{lem:inaccurate,lem:accurate} state that if $\cM$ is $(\eps,\delta)$-DP and $(p,\beta)$-weakly accurate, and if $f$, $X$, and $X'$ are sampled as above, then with high probability $\cM^f(X')$ will be far from $y_1$, while $\cM^f(X)$ will be close to $y_1$. Since the bounds we obtain are in terms of the query complexity of $\cM$, we can combine them with group privacy to obtain a lower bound on the query complexity of any mechanism satisfying the aforementioned guarantees. 

    Let parameters $\eps,\delta,\beta,\kappa$, $p$, and $\cM$ be as in \Cref{thm:query_lb}, and let $\tau$, $y_1$, $f$, $X$, and $X'$ be as defined above. Additionally, for all $q\in\N$, an algorithm $\cA$ is a \emph{$q$-query algorithm} if $\cA^f$ has expected query complexity at most $q$ for all inputs $Z$ and functions $f$. Then we can show the following lemmas:

\begin{lemma}
    \label{lem:inaccurate}
    For all $q\in\N$, if $\cM$ is a $q$-query algorithm, %
    then 
    \[
    \Pr\brackets{\abs{\cM^{f}(X')-y_1}<\frac 12}\leq \min\paren{\beta ,\frac{2}{\kappa}} + O\paren{\frac{tq}{M} + \frac{n^2}{t}}\, .
    \]
\end{lemma}

\begin{lemma}
    \label{lem:accurate}
    For all $q\in\N$, if $\cM$ is a $q$-query algorithm, %
    then 
    \[
    \Pr\brackets{\abs{\cM^{f}(X)-y_1}\geq 1/2}\leq \beta + qe^{-\Omega\paren{p\tau}} + O\paren{\frac{n^2}{t}}\, .
    \] 
\end{lemma}

To prove \Cref{lem:inaccurate,lem:accurate}, we show that the view of the mechanism on input $X'$ and $X$ is close in TV distance to the  view of the mechanism on input $Z'\sim\cD_0^n$ and $Z\sim\cD_1^n$. We then apply the accuracy guarantee of $\cM$ and the fact that $y_1$ is uniformly random in $[\kappa]\setminus[\kappa/2]$ to deduce each bound. We defer the proofs of \Cref{lem:accurate,lem:inaccurate} to the end of the section, and complete the proof of \Cref{thm:query_lb} below. 

Since the lower bound holds trivially  for $p=o\paren{1/\tau}$ we assume $p=\Omega\paren{1/\tau}$. Set $q=e^{O(p\tau)}$, set and $M$ and $t$ such that $\frac{qt}{M} + \frac{n^2}{t} = e^{-\omega(\tau)}$ (essentially, this term is small enough that we can ignore it). Suppose $\cM$ has query complexity $q$. Combining \Cref{lem:inaccurate,lem:accurate} with group privacy (\Cref{fact:group-privacy}) yields the desired privacy violation since
    \begin{align*}
    1-\beta - e^{-\Omega(\min(p\tau,pn))}
    &\leq \Pr\brackets{\abs{\cM^{f}(X)-y_1}< 1/2}\\
    &\leq e^{\eps\tau}\paren{\Pr\brackets{\abs{\cM^{f}(X')-y_1}< 1/2} + \delta/\eps}\\
    &\leq e^{\eps\tau}\paren{\min\paren{\beta,\frac2{\kappa}} + \delta/\eps}.
    \end{align*}
The first inequality follows from \Cref{lem:accurate}, the second follows from group privacy and the fact that $X$ and $X'$ are at distance $\tau$, and the third follows from \Cref{lem:inaccurate}. If $\tau=\frac{1}{2\eps}\ln \gamma\min\paren{\frac\kappa\beta, \frac\eps\delta}$, $\delta<\eps\beta^2$, and $\gamma$ is sufficiently small, the right hand side is at most $\frac 12$. Additionally, since $p=\Omega\paren{1/\tau}$ and $\beta<\frac 14$, the left hand side is at least $\frac 12$, which yields the desired contradiction.  
\end{proof}

In the remainder of the section we prove \Cref{lem:accurate,lem:inaccurate}. We first introduce a standard definition for reasoning about the distribution of the view of an algorithm.

\begin{definition}[The view of an algorithm {$\cA[\cD]$}]
    \label{def:view}
    For all algorithms $\cA$, and all distributions $\cD$ over tuples $\set{(f,I)}$, let $\cA[\cD]$ be the distribution over query answer histories $\set{(f(J_i),J_i)}$ when $(f,I)\sim\cD$ and $\cA^f$ is run on input $I$.  
\end{definition}

    One of the main challenges in proving query complexity lower bounds is handling algorithms that make ``adaptive'' queries to the function. An algorithm $\cA$ is \emph{nonadaptive} if given an input $Z$ and query access to $f$, it specifies all queries $J_1,\dots, J_q$ after reading $Z$---that is, its queries to $f$ do not depend on previous query answers. An algorithm $\cA$ is \emph{adaptive} if it is not nonadaptive. 
    
    In order to argue that the statements hold for algorithms that make adaptive queries to $f$, we leverage the following idea, which was introduced in \cite{LangeLRV25} who used it to prove a query complexity lower bound in the context of property testing. While the exact statement of \Cref{prop:adaptivity} does not appear in \cite{LangeLRV25}, it follows readily from their techniques. Informally, \Cref{prop:adaptivity} states that if a $q$-query nonadaptive algorithm cannot distinguish two distributions $\cD_0$ over $\set{(f_0,I)}$ and $\cD_1$ over $\set{(f_1,I)}$ when a description of $f_0$ is always included with its input $I$, then no $q$-query adaptive algorithm can distinguish between $\cD_0$ and $\cD_1$ when the description of $f_0$ is not included in the input.

\begin{proposition}
    \label{prop:adaptivity}
    Let $\cD$ be a distribution over tuples $(f_0,f_1,Z)$ where $f_0$ and $f_1$ are functions with domain $\cZ^*$ and $Z\in\cZ^*$. 
    Define the following distributions:
    \begin{itemize}
        \item Let $\cD_0$ and $\cD_1$ be distributions over tuples $(f_0,Z)$ and $(f_1,Z)$ respectively.
        \item Let $\wt\cD_0$ and $\wt\cD_1$ denote the distributions given by sampling $(f_0,f_1,Z)\sim\cD$, 
        and returning $(f_0, (f_0,Z))$ and $(f_1,(f_0,Z))$ respectively.
    \end{itemize}
    Then, for every $q$-query adaptive algorithm $\cA$, there exists a $q$-query nonadaptive algorithm $\cA_{na}$ such that 
    \[
    \dtv\paren{\cA[\cD_0],\cA[\cD_1]}\leq \dtv\paren{\cA_{na}[\wt\cD_0],\cA_{na}[\wt \cD_1]}.
    \]
\end{proposition}
\begin{proof}
    The proof proceeds by a simulation argument. Let $\cA$ be a $q$-query adaptive algorithm that gets input $Z$ and query access to $h$, 
    and consider the following nonadaptive algorithm $\cA_{na}$ that gets input $(g,Z)$ and query access to $h$:
    \begin{enumerate}
        \item Simulate $\cA$ and answer its queries $J_1,\dots, J_Q\in\cZ^*$ using $g$ (i.e., return $g(J_i)$).
        \item Nonadaptively query $h$ on $J_1,\dots, J_Q$. If $h(J_i) = g(J_i)$ for all $i\in[Q]$ then return $0$, otherwise return $1$.
    \end{enumerate}
    Let $\cV$ be a coupling over views $(v_0,v_1)\sim\paren{\cA[\cD_0'],\cA[\cD_1']}$, where both runs of $\cA$ are executed using the same random coins and the inputs are sampled jointly as $(f_0,f_1,Z)\sim\cD$. By the coupling lemma
    \[
    \dtv\paren{\cA[\cD_0],\cA[\cD_1]}\leq \Pr[v_0\neq v_1].
    \]
    Next, we show that $\Pr[v_0\neq v_1]\leq \dtv\paren{\cA_{na}[\wt\cD_0],\cA_{na}[\wt\cD_1]}$. Let $E$ be the set of query answer histories $\set{(h(J_i), J_i)}$ such that $h(J_i)\neq g(J_i)$ for some $i$. Observe that $\Pr_{\wt v_0\sim \cA_{na}[\wt\cD_0]}[\wt v_0\in E]=0$ since $h=g$ when the input is sampled from $\wt\cD_0$. On the other hand, when the input is sampled from $\wt\cD_1$, event $E$ occurs if and only if $v_0\neq v_1$ (since otherwise $h(J_i)=g(J_i)$ for all queries). Thus,
    \begin{align*}
    \Pr[v_0\neq v_1]&\leq \Pr_{\wt v_1\sim \cA_{na}[\wt\cD_1]}[\wt v_1\in E] \\
    &= \abs{\Pr_{\wt v_1\sim \cA_{na}[\wt\cD_1]}[\wt v_1\in E] - \Pr_{\wt v_0\sim \cA_{na}[\wt\cD_1]}[\wt v_0\in E]} \\
    &\leq \dtv\paren{\cA_{na}[\wt\cD_0],\cA_{na}[\wt\cD_1]},
    \end{align*}
    which completes the proof.
\end{proof}

We now complete the proofs of \Cref{lem:accurate,lem:inaccurate}. We will abuse notation and treat $f$ as both a function and a random variable whose distribution is given by the procedure for sampling $B_0$ and $B_1$ described in the proof of \Cref{thm:query_lb}. We will also define two additional functions $h$ and $g$ in the proofs of \Cref{lem:accurate,lem:inaccurate} respectively. As with $f$, we abuse notation and treat $g$ and $h$ as both functions and random variables whose distributions are given by the procedure for sampling $B_0$, $B_1$, $y_0$, and $y_1$, in the proof of \Cref{thm:query_lb}. 

\begin{proof}[Proof of \Cref{lem:inaccurate}]
    Let $\cD'$ denote the distribution over pairs $(f,X')$ sampled as in the proof of \Cref{thm:query_lb}. Define the function 
    \[
    h(Z) = 
    \begin{cases}
        0 &\text{if}~~ |Z|<N\\
        y_0 &\text{else if}~~ Z\subset B_0\\
        \kappa &\text{otherwise.}
    \end{cases}
    \]
    Let $\cD''$ denote the distribution over pairs of the form $(h,Z)$ where $h$ is sampled by sampling $B_0$, and $y_0$ as in the proof of \Cref{thm:query_lb}, and $Z\sim\cD_0^n$. We will first show that the bound holds when the input is sampled from $\cD''$. 
    Then, we will complete the argument by bounding the TV distance between $\cM[\cD']$ and $\cM[\cD'']$. 

    Since $Z\sim\cD_0^n$ is a subset of $B_0$ with probability $1$, the function $h$ is $(y_0,N)$-constant under $\cD_0$. To prove the first part of the bound, observe that since $\abs{y_0-y_1}\geq 1$, and a sample of $n$ i.i.d.\ elements from $\cD_0^n$ contains a collision with probability $O(n^2/t)$, the accuracy guarantee of $\cM$ implies that $\Pr_{Z\sim\cD_0^n}\brackets{\abs{\cM^h(Z)-y_1}<1/2}\leq \beta + O\paren{n^2/t}$. The $2/\kappa$ part of the bound holds since the distribution of $y_1$ is uniformly random over $[\kappa]\setminus[\kappa/2]$ even when conditioned on the input $Z$ and a complete description of the function $h$. 

    To complete the argument, we bound the TV distance between $\cA[\cD']$ and $\cA[\cD'']$ for any $q$-query algorithm $\cA$. Let $\cD^*$ denote the distribution over pairs $(h,X')$ where $h$ and $X'$ are sampled as above.

    \begin{claim}
        \label{claim:D'-D*}
        Let $\cA$ be a $q$-query algorithm. Then, 
        \[
        \dtv\paren{\cA[\cD'],\cA[\cD^*]}\leq O\paren{\frac{qt}{M}}.
        \]
    \end{claim}
    \begin{proof}
        For a fixed $B_0$ and $B_1$ and a query $J$, if $f(J)\neq h(J)$, then $|J|\geq N$ and $J\subset B_0\cup B_1$ and $J\cap B_1\neq \emptyset$.

        We first consider nonadaptive algorithms $\cA_{na}$ that are given additional input $(B_0,y_0)$. Since $B_0$ and $y_0$ uniquely determine $h$, we can leverage \Cref{prop:adaptivity} to lift the argument to adaptive algorithms.

        To bound the probability that the algorithm makes a query $J$ such that $f(J)\neq h(J)$, it will be convenient to view the distribution over $h,f$ and $X'$ as follows: Instead of sampling $B_0$ and $B_1$ and then sampling $X'$, we can equivalently first sample an $X'$, then sample the remaining $t-n$ elements of $B_0$, and finally sample the $t$ elements of $B_1$ from $G\setminus B_0$. Now, fix some query $J$ made by $\cA_{na}$ on input $X'$ and $(B_0,y_0)$. Over the randomness of $B_1$, the probability that $J\subset B_0\cup B_1$ and $J\cap B_1\neq \emptyset$ is at most the probability that the elements in $J\setminus B_0$ are in $B_1$. Since $|B_1|=t$ and $|G| = M$, the probability is at most $O\paren{\frac{t}{M}}$. Let $Q$ be a random variable denoting the number of queries made by $\cA_{na}$ and denote the random queries $J_1,\dots J_Q$. Since $\cA_{na}$ has expected query complexity $q$, we can apply the union bound to see that
        \[
        \Pr\brackets{\bigcup_{i=1}^Q \set{f(J_i)\neq h(J_i)}}\leq  \Ex_{Q}\brackets{\sum_{i=1}^{Q}\Pr[f(J_i)\neq h(J_i)\mid Q]}\leq O\paren{\frac{qt}{M}}.
        \]
        Since $\cA_{na}$ is provided with a complete description of $h$ as input, we can apply \Cref{prop:adaptivity} to complete the proof.  
    \end{proof}

    \begin{claim}
        \label{claim:D*-D'}
        For all algorithms $\cA$ we have
        \[
        \dtv\paren{\cA[\cD^*],\cA[\cD'']}\leq O\paren{\frac{n^2}{t}}.
        \]
    \end{claim}
    \begin{proof}
        Since $\cD^*$ and $\cD''$ both have the same distribution over the function and only differ in distribution on the dataset, it suffices to bound the TV distance between the distributions of $X'$ and $Z$. Notice that $X'$ is defined by sampling $n$ points from $B_0$ uniformly without replacement, while $Z$ is given by sampling $n$ points uniformly and independently from $B_0$. Conditioned on the event that $Z$ does not contain a collision, these distributions are the same. Since $Z$ is uniform over $B_1$, a set of size $t$, it contains a collision with probability $O(n^2/t)$, and thus the TV distance between $X'$ and $Z$ is $O(n^2/t)$.   
    \end{proof}

    By the triangle inequality and \Cref{claim:D'-D*,claim:D*-D'}, we have $\dtv\paren{\cA[\cD'],\cA[\cD'']}\leq O\paren{\frac{qt}{M} + \frac{n^2}{t}}$. It follows that 
    \[
    \Pr\brackets{\abs{\cM^{f}(X')-y_1}<1/2}\leq \min\paren{\beta,\frac{2}{\kappa}} + O\paren{\frac{qt}{M} + \frac{n^2}{t}},
    \]
    which completes the proof of the lemma.   
\end{proof}

\begin{proof}[Proof of \Cref{lem:accurate}]  The proof is similar to that of \Cref{lem:inaccurate}. Let $\cD'$ denote the distribution over $(f,X)$ in the proof of \Cref{thm:query_lb} and let $\cD''$ denote the distribution over $(f,Z)$ where $Z\sim\cD_1^n$. We aim to bound $\dtv(\cM[\cD'],\cM[\cD''])$. Since $f$ is $(y_1,N)$-constant under $\cD_1$, the accuracy guarantee of $\cM$ and the TV distance bound together imply the result. 

    In order to facilitate the argument, we introduce an intermediate function $g$ defined as 
    \[
    g(Z) =
    \begin{cases}
        0 & \text{if}~~ |Z|<N\\
        y_1 & \text{else if}~~ Z\subset B_0\cup B_1\\
        \kappa &\text{otherwise}.
    \end{cases}
    \]
    Let $\wt\cD'$ and $\wt\cD''$ denote the distributions over pairs $(g,X)$ and $(g,Z)$ respectively. 
    
    \begin{claim}
        For any $q$-query algorithm $\cA$ we have 
        \[
        \dtv\paren{\cA[\cD'],\cA[\wt\cD']}\leq qe^{-\Omega\paren{p\tau}}.
        \]
    \end{claim}
    \begin{proof}
        Observe that for any query $J$, if $f(J)\neq g(Z)$ then $J\subset B_0$ and $|J|\geq N$. Thus, in order to bound the TV distance between $\cA[\cD']$ and $\cA[\wt\cD']$ it suffices to bound the probability that any query $J$ of size at least $N$ is a subset of $B_0$. 

        Let $B=B_0\cup B_1$. We first prove the claim for nonadaptive algorithms $\cA_{na}$ that are given additional input $(B, y_1)$, and then use \Cref{prop:adaptivity} to prove the bound for adaptive algorithms. We note that $\cA_{na}$ is given $B$ for free, but is not provided with the partition of $B$ into $B_0$ and $B_1$. Since $B_0$ and $B_1$ are chosen uniformly at random without replacement from $G$, we can equivalently sample them by first sampling the set $B$, and then sampling $B_0$ by choosing $t$ elements uniformly without replacement from $B$, and letting $B_1$ be the remaining $t$ elements. Additionally, since $X$ consists of a uniform set of $\tau$ distinct elements from $B_1$ and $n-\tau$ distinct elements from $B_0$, the distribution of $B_0$ and $B_1$ conditioned on $B$ and $X$ is equivalent to the following distribution: First sample $\tau$ elements at random from $X$ and place them in $B_1$ and place the remaining $n-\tau$ elements in $B_0$. Second, partition the remaining $2t-n$ elements from $B$ by choosing $t-\tau$ elements uniformly and placing them in $B_1$ and placing the remaining $t-n+\tau$ elements in $B_0$.

        Below, we leverage this alternative view of the sampling process to prove the bound for nonadaptive algorithms.

        Fix a query $J$, and let $J_X$ and $J_{\overline X}$ denote $J\cap X$ and $J\setminus X$. Recall that $X$ is composed of $\tau$ distinct elements from $B_1$, and thus, over the randomness of $B_1$, the probability that $J_X$ is a subset of $B_0$ is the probability that no element in $J_X$ is placed in $B_1$. Since the $\tau$ elements from $B_1$ are chosen uniformly at random without replacement from $X$, a set with $n$ elements, we have  
        \[
        \Pr\brackets{J_X\subset B_0}\leq \paren{1-\frac{|J_X|}{n}}^\tau\leq \exp\paren{-|J_X|\cdot \frac{\tau}{n}}.
        \]
        Similarly, over the randomness of the partition of the elements in $B\setminus X$ into $B_0$ and $B_1$, the probability that $J_{\overline X}\subset B_0$ is at most $\exp\paren{{\Omega\paren{-|J_{\overline{X}}|}}}$.

        Since the partition of the elements in $X$ and $B\setminus X$ is chosen independently, the events $J_X\subset B_0$ and $J_{\overline X}\subset B_0$ are independent. Hence, for all $\tau\leq n$ we have
        \[
        \Pr[J\subset B_0] 
        = \Pr[J_X\subset B_0]\Pr[J_{\overline{X}}\subset B_0] \leq \exp\paren{-\Omega\paren{\frac{\tau}{n}\paren{|J_X| + |J_{\overline X}|}}}.
        \]
        Since $|J|\geq N=pn$, we obtain $\Pr[J\subset B_0]\leq e^{-\Omega\paren{p\tau}}$. A union bound %
        as in the proof of \Cref{claim:D'-D*} implies the bound for nonadaptive algorithms. Since $B$ and $y_1$ uniquely determine the function $g$, we can apply \Cref{prop:adaptivity} and lift the bound to adaptive algorithms.
    \end{proof}

    \begin{claim}
        For any $q$-query algorithm $\cA$ we have 
        \[
        \dtv\paren{\cA[\cD''],\cA[\wt\cD'']}\leq q\cdot e^{-\Omega(pn)}.
        \]
    \end{claim}
    \begin{proof}
        Fix a query $J$ made by $\cA$, and observe that if $J\subset Z$ then $f(J)=g(J)$. Thus, we can assume without loss of generality that $J\not\subset Z$. If $f(J)\neq g(J)$ then $|J|\geq N$ and $J\subset B_0$. We first prove the bound for a nonadaptive algorithm $\cA_{na}$ that is given additional input $(B,y_1)$ where $B=B_0\cup B_1$. Note that this uniquely determines the function $h$, but conditioned on $B$ and $Z$, the distribution over $B_0$ is equivalent to sampling $t$ distinct elements uniformly from $B\setminus Z$, and thus, each element is included with probability $\frac{t}{2t-n}\leq \frac 23$. 
    
        Since $Z\sim\cD_1^n$ we have that $Z\subset B_1$ and thus, if $f(J)\neq g(J)$ then $J$ must contain at least $N$ elements from $B_0$ and no elements from $G\setminus B_0$. But over the random sampling of $B_0$ conditioned on $Z$ and $B$ we have $\Pr[J\subset B_0]\leq e^{-\Omega(N)}$. A union bound %
        argument as in the proof of \Cref{claim:D'-D*} suffices to prove the bound for nonadaptive algorithms. Applying \Cref{prop:adaptivity} completes the proof.  
    \end{proof}

    \begin{claim}
        For any $q$-query algorithm $\cA$ we have
        \[
        \dtv\paren{\cA[\wt\cD'],\cA[\wt\cD'']}\leq O\paren{\frac{n^2}{t}}.
        \]
    \end{claim}
    \begin{proof}
        Since $B_0$ and $B_1$ are sampled uniformly without replacement from $G$, the set $B_0\cup B_1$ is sampled uniformly without replacement from $G$, as well. Let $B = B_0\cup B_1$. Then since $g$ only depends on the elements in $B$, and not on the partition of $B$ into $B_0$ and $B_1$, the distribution $\wt\cD'$ is the same whether $X$ is sampled by choosing $\tau$ elements from $B_0$ and $n-\tau$ elements from $B_1$, or instead by choosing $n$ uniformly random elements from $B$ (without replacement). Similarly, conditioned on the event that there are no collisions in the sampling of $Z$ from $B_1$, the distribution of $\wt\cD''$ is the same whether $Z$ is sampled from $\cD_1^n$, or instead sampled uniformly without replacement from $B$. Thus, conditioned on the event that $Z$ does not contain any collisions, the distributions $\wt\cD'$ and $\wt\cD''$ are equivalent, and hence, the TV distance between $\cA[\wt\cD']$ and $\cA[\wt\cD'']$ is at most the probability that $Z$ contains a collision. Since $Z$ consists of $n$ independent and uniform samples from $B_1$, a set of size $t$, the probability of a collisions is at most $O(n^2/t)$.
    \end{proof}

    To complete the proof of \Cref{lem:accurate}, it suffices to apply the triangle inequality and the above claims to obtain, 
    \[
    \dtv\paren{\cM[\cD'],\cM[\cD'']}\leq qe^{-\Omega\paren{\min\set{p\tau,pn}}} + O\paren{\frac{n^2}{t}}.
    \]
    Since $Z\sim\cD_1^n$ contains a collision with probability at most $O\paren{n^2/t}$, the $(p,\beta)$ accuracy guarantee of $\cM$ and the TV distance bound yields the result. 
\end{proof}

\ifnum\neurips=0

\bibliography{bibliography,biblio}

\appendix
\fi

\section{Technical discussion of existing guarantees}
\label{sec:technical_extensions}

In this section we present some minor extensions to the guarantees provided in prior work.
The arguments are standard; we present them to give a more complete comparison with our results.

\subsection{Extension of  \texorpdfstring{\cite{LinderRSS25}}{privately evaluating black-box functions}}
\label{sec:lrss_comparison}

The work of \cite{LinderRSS25} gives algorithms for privately evaluating arbitrary functions.
They achieve nearly optimal error and sample complexity but require exponentially many queries (and thus exponential time).
In this section, we discuss how amplification by subsampling allows one to reduce the time complexity at the cost of increased error.
The same techniques would apply identically to the algorithms of \cite{FangDY22}.

In the setting of \Cref{sec:results_intro}, amplification by subsampling achieves the following tradeoff: sample complexity $\frac{N\log\kappa/\beta}{\trade\eps}$ at a cost of runtime $N^{\trade}$. 
In contrast, our work
achieves the same sample complexity at a cost of runtime $e^\trade$, 
but is restricted to monotone functions.
This difference is significant for all settings of $\trade$, however it is most striking when $\trade=\Theta(\log N)$, in which case our mechanisms run in time $\poly N$, whereas the application of \cite{LinderRSS25} runs in time $N^{\Theta(\log N)}$.

Recall \Cref{ass:intro_concentration}: there exists some $\nu\in\R$ and $N:\R\to\R$ such that if $n\geq N(\alpha) + \frac{\log1/\beta}{\alpha^2}$ and $Z\sim\cD^n$, then $\abs{f(Z)-\nu}\leq \alpha$ with probability at least $1-\beta$. In this setting, amplification by subsampling and \cite{LinderRSS25} achieve the following tradeoff.

\begin{theorem}
    \label{thm:privacy_wrappers_tradeoff}
    Fix range size $\kappa>0$, privacy parameter $\eps\in(0,1]$, failure probability $\beta>0$, and tradeoff parameter $\trade\leq \frac{1}{\eps}\log\frac\kappa\beta$. There exists a mechanism $\cM$ such that for all functions $f:\cZ^*\to[\kappa]$ mechanism $\cM^f$ is $\eps$-DP, and on datasets of size $n$ has query complexity and runtime $n^{O(\trade\log\frac{\kappa}{\beta})}$. Additionally, if $f$ and $\cD$ satisfy \Cref{ass:intro_concentration}, then for all $\alpha>0$ and 
    \[
    n = \Omega\paren{\frac{N(\alpha)\log\kappa/\beta}{\trade\eps} + \frac{\log\kappa/\beta\log1/\beta}{\trade\alpha^2\eps} + \frac{\log\kappa/\beta\log n}{\alpha^2\eps}}
    \]
    we have
    \[
    \Pr_{Z\sim\cD^n}\brackets{\abs{\cM^f(Z) - \nu}\geq \alpha}\leq \beta.
    \]
\end{theorem}

While \Cref{thm:privacy_wrappers_tradeoff} applies to bounded range functions, analogous to \Cref{thm:med-of-q}, a similar result can be shown for the setting of functions with range $\R$, analogous to \Cref{thm:avg-of-q}. Additionally, for simplicity we only compare to the pure DP version of ``sens-o-matic'' (Theorem 3.1 in \cite{LinderRSS25}); however, we can obtain a similar guarantee to their approximate DP version by substituting a generalized interior point mechanism for the exponential mechanism used to compute the median in \Cref{alg:med-of-q}. Indeed, this is exactly the approach taken by \cite{LinderRSS25} to obtain their approximate DP result.  We apply standard amplification by subsampling for pure DP.
\begin{proposition}[Privacy Amplification by Subsampling, \citep{FiorettoVH2025}, Theorem 3.28]
    Let $\cM:\cZ_\bot^*\to \cY$ satisfy $\eps$-DP.
    For $p\in (0,1)$, define $\cM_p(Z) = \cM(S)$ where $S\sim \cS_p(Z)$, i.e., we set each element of $Z$ to $\bot$ independently with probability $p$ and then run $\cM$.
    Then $\cM_p$ is $\eps'$-DP for $\eps'=\log(1+p(e^\eps-1))$.
    In particular, if $\eps\le 1$ then $\eps'\le 2p \eps$.
\end{proposition}

In order to state the guarantee of the sens-o-matic mechanism, we define the following notation: For a parameter $\lambda\in\N$ and a dataset $Z\in\cZ^*$, let $N_\lambda(Z)=\set{S\subset Z \mid |S|\geq |Z|-\lambda}$, that is $N_\lambda(Z)$ is the set of subsets of $Z$ given by removing up to $\lambda$ points from $Z$. Additionally, for a function $f$, we let $f(N_\lambda(Z))$ denote the set $\set{f(S) \mid S\in N_\lambda(Z)}$.

\begin{theorem}[\cite{LinderRSS25}, Theorem 3.1, pure DP]

Fix privacy parameter $\eps>0$, error probability, $\beta > 0$, and range size $\kappa\in\N$, there exist a mechanism $\cM$ such that for every function $f:\cZ^*\to[\kappa]$ and dataset $Z\in\cZ^n$, with probability at least $1-\beta$,
    \[ 
    \cM^f(Z) \in [\min f \bparen{N_\lambda(Z)}, \max f \bparen{N_\lambda(Z)}] \, ,
    \]
    where  $\lambda = O\bparen{\frac1\eps\log\frac{\kappa}\beta}$ .
\end{theorem}

With these in hand, the proof of \cref{thm:privacy_wrappers_tradeoff} is almost immediate.

\begin{proof}[Proof sketch]
    We subsample $S\subset Z$ by including each element with probability $p\geq \eps$ and then run the sens-o-matic mechanism with privacy parameter $\eps/p$ and failure probability $\beta$. 
    This mechanism is $2\eps$-DP and runs in time $\paren{pn}^{O\paren{\frac p\eps\log\frac\kappa\beta}}$. 
    
    To see why accuracy holds, observe that sens-o-matic queries the function $f$ on all subsets of the input dataset $S$ given by removing at most $\frac p\eps\log\frac\kappa\beta$. Since there are at most $n^{O\paren{\frac p\eps\log\frac\kappa\beta}}$ such subsets, the applying \Cref{ass:intro_concentration} and the union bound implies that all subsets are within $\alpha$ of $\nu$, and thus the mechanism outputs a $y$ such that $\abs{y-\nu}\leq \alpha$. Reparameterizing $p=\eps\trade$ yields the result. 
\end{proof}

\subsection{Translating \texorpdfstring{\cite{SteinkeS25}}{Steinke Steinke} to our setting}
\label{sec:ss25}

In this section, we translate the guarantees provided by \cite{SteinkeS25} to our setting. Specifically, we explain how their result yields a similar tradeoff between sample and \emph{query} complexity for general functions---that is, one can save a factor of $\trade$ over the sample complexity of subsample-and-aggregate, while paying a factor of $e^{\trade}$ in query complexity. However, as discussed in the introduction, the mechanism in \Cref{thm:ss25} nevertheless has runtime $n^{O\paren{\frac1\eps{\log\frac1\delta}}}$.  

Below, we state an informal corollary of \cite{SteinkeS25} in the style of \Cref{thm:avg_of_q-intro}. Recall, that we are interested in functions $f$ and distributions $\cD$ that satisfy \Cref{ass:intro_concentration}---that is, there exists $\nu\in\R$ and $N=N(\alpha)$ such that if $n\geq N + \frac{\log1/\beta}{\alpha^2}$ and $Z\sim\cD^n$ then $|f(Z)/n - \nu|\leq \alpha$ with probability at least $1-\beta$. Since the mechanism of \cite{SteinkeS25} is for functions $f:\cZ^*\to [\kappa]$ that need not be monotone, we modify this assumption to remove the normalization by $n$---that is, we make the slightly more general assumption that
\begin{equation}
    \label{eq:ss25_concentration}
    \text{if}~~ n\geq N + \frac{\log1/\beta}{\alpha^2} ~~\text{then}~~ \Pr_{Z\sim\cD^n}\brackets{|f(Z) - \nu|\geq \alpha}\leq \beta.
\end{equation}

For functions and distributions that satisfy \eqref{eq:ss25_concentration}, the mechanism of \cite{SteinkeS25} yields the following result: 

\begin{corollary}[Corollary of \cite{SteinkeS25}]
    Fix privacy parameters $\eps,\delta>0$, and $\trade\leq \wt O\paren{\frac1\eps\log\frac1\delta}$. There exists a mechanism $\cM$ such that for all functions $f:\cZ^*\to[\kappa]$ mechanism $\cM^f$ is $(\eps,\delta)$-DP, has query complexity $\wt O\paren{e^{\trade}\cdot \frac{\log(1/\delta)\log(n)}{\eps}}$, and has runtime $n^{\wt O\paren{\frac1\eps\log\frac1\delta}}$. Additionally, if $f$ and $\cD$ satisfy \eqref{eq:ss25_concentration}, then for all $\alpha,\beta>0$ and
    
    \[
    n = \wt\Omega\paren{\frac{N(\alpha)\log1/\delta}{\trade\eps} + \frac{\log1/\delta \log1/\beta}{\trade\alpha^2\eps} + \frac{\log1/\delta}{\alpha^2\eps} + \frac{\log(1/\delta)^2}{\eps^2\trade}}
    \]
    we have
    \[
    \Pr_{Z\sim\cD^n}\brackets{\abs{\cM^f(Z,\alpha) - \nu}\geq \alpha}\leq \beta.
    \]    
    Here $\wt O$ and $\wt\Omega$ hide an additional factor of $\exp(O(\log^* \kappa))$.
\end{corollary}

\begin{proof}[Proof sketch]

The proof is a special case of the main result of \cite{SteinkeS25} which we state below.

\begin{theorem}[Theorem 1.1 \cite{SteinkeS25}] \label{thm:ss25}
    Let $\mathcal{Y} \subseteq \mathbb{R}$ be finite and let $\mathcal{X}$ be arbitrary; denote $\mathcal{X}^* = \bigcup_{n \in \mathbb{N}} \mathcal{X}^n$.
    Let $\varepsilon,\delta>0$ and $n, m, \tau \in \mathbb{N}$ satisfy \[n \ge m \ge \tau=\frac{1}{\varepsilon} \log(1/\delta) \exp(O(\log^* |\mathcal{Y}|)). \label{eq:thm:main:t}\] Let 
    \[
    k = \frac{{n \choose \tau}}{{m \choose \tau}} \left(1 + \log{ m \choose \tau} \right) +1 \label{eq:thm:main:k}
    \]
    Then, for all $f : \mathcal{X}^* \to \mathcal{Y}$, there exists an algorithm $\cM^f : \mathcal{X}^n \to \mathcal{Y}$ with the following properties.
    \begin{itemize}
        \item \textbf{Privacy:} $\cM^f$ is $(\varepsilon,\delta)$-differentially private. 
        \item \textbf{Statistical Accuracy:} Let $\mathcal{D}$ be an arbitrary probability distribution on $\mathcal{X}$. Suppose 
        \[\Pr_{X \gets \mathcal{D}^{n-m}}\brackets{|f(X)-\nu|\le\alpha}\ge 1-\beta
        \]
        for some $\alpha, \beta, \nu \in \mathbb{R}$. Then 
        \[
        \Pr_{X \gets \mathcal{D}^n}\brackets{|\cM^f(X)-\nu|\le\alpha}\ge 1- k\beta
        \]
        \item \textbf{Oracle Efficiency:} On input $x \in \mathcal{X}^n$, mechanism $\cM^f(x)$ selects $O(k)$ subsets of $x$, each of size $n-m$, and evaluates $f$ on those subsets; other than these $k$ evaluations, $\cM^f(x)$ does not depend on either $f$ or $x$.
    \end{itemize}
\end{theorem}

Recall that in our setting, we trade off between sample and time complexity by changing the subsampling probability $p$ used in \Cref{thm:med-of-q,thm:avg-of-q}. Since our mechanisms query $f$ on subsets at depth approximately $pn$, the relevant point of comparison is when $m=(1-p)n$. The following proposition allows us to bound the query complexity of $\cM$, the mechanism in \Cref{thm:ss25}, when $m=(1-p)n$.

\begin{proposition}
    \label{prop:binom-ratio-bound}
    Fix $n,\tau\in\N$ and $p\in(0,1)$. If $\tau<(1-p)n$ and $p\geq 2\tau/n$ then 
    \[
    \frac{\binom{n}{\tau}}{\binom{(1-p)n}{\tau}}\leq e^{p\tau}.
    \]
\end{proposition}
\begin{proof}
Expanding the binomial coefficients we obtain
\[
\frac{\binom{n}{\tau}}{\binom{(1-p)n}{\tau}}
=
\prod_{i=0}^{\tau-1}
\frac{n-i}{(1-p)n-i}.
\]
Since $\frac{n-i}{(1-p)n-i}\leq \frac{n}{(1-p)n-(\tau-1)}$, we see that
\[
\frac{\binom{n}{\tau}}{\binom{(1-p)n}{\tau}}
\leq
\paren{\frac{n}{(1-p)n-(\tau-1)}}^\tau\leq \paren{\frac{1}{(1-p)-\tau/n}}^\tau.
\]
Since $p>2\tau/n$ we have $1-p-\tau/n\geq 1-p/2$ the above is at most $e^{p\tau}$.
\end{proof}

Now, since all subsets are of size $pn$, we can apply \eqref{eq:ss25_concentration} to see that if $n\geq \frac{N(\alpha)}{p} + \frac{\log1/\beta}{\alpha^2 p}$ then the statistical accuracy assumption in \Cref{thm:ss25} will hold. Scaling $\beta\gets \beta\cdot e^{-p\tau}$, we see that $\abs{\cM^f(X)-\nu}\leq \alpha$ with probability at least $1-\beta$.     

To obtain the desired tradeoff and complete the proof sketch we can simply set $p=\frac{\trade}{\tau}$.
\end{proof}

\section{Related work for applications}
\label{sec:more_related_work}

\subsection{Subsample-and-aggregate baseline}

In this section we formalize the baseline provided by the subsample-and-aggregate framework explained in the introduction. Recall that we are interested in providing accuracy guarantees for functions $f$ and distributions $\cD$ satisfying \Cref{ass:intro_concentration}, which states that there exists $N$ and $\nu$ such that for all $n\geq N(\alpha) + \frac{\log1/\beta}{\alpha^2}$ we have $\abs{\frac{f(Z)}{n} - \nu}\leq \alpha$ with probability at least $1-\beta$ over the sampling of $Z\sim\cD^n$. 

\begin{proposition}[Subsample-and-aggregate baseline]
    \label{prop:ssa_baseline}
    For all $\eps,\delta>0$ there exists an $(\eps,\delta)$-DP mechanism $\cM$ such that for all $f$ and $\cD$ that satisfy \Cref{ass:intro_concentration}, and all $\alpha,\beta>0$ the following holds: If 
    \[
    n=\Omega\paren{\frac{N(\alpha)\log(1/\beta\delta)}{\eps} + \frac{\log(1/\beta\delta)}{\alpha^2\eps}}
    \]
    then
    \[
    \Pr_{Z\sim\cD^n}\brackets{\abs{\cM^f(Z) - \nu}\geq 3\alpha}\leq \beta.
    \]
    Moreover, $\cM$ runs in time $O\paren{\frac1\eps\log\frac{1}{\beta\delta}}$
\end{proposition}

\begin{remark}
    While we only include a formal result for the $\delta>0$ case, for functions with finite range of size $\kappa$ the standard exponential mechanism for median is $(\eps,0)$-DP and has the same sample complexity with $1/\delta$ replaced by $\kappa/\beta$. 
\end{remark}

\begin{proof}
    We use the Stable Histograms approach as in \cite{KarwaV18}. The exact guarantees we cite appear in Appendix C of \cite{BrownGSUZ21}. 

    \begin{lemma}[Stable Histogram guarantees]
    \label{lem:stablehistogram}
    There exists a constant $C>0$ and an $(\eps,\delta)$-differentially private mechanism $\cA_{\eps,\delta}$ that gets as input a set of bins $\set{B_b}_{b\in\Z}$ and a dataset $Z=z_1,\ldots,z_n$ drawn i.i.d.\ from distribution $P$, and satisfies the following guarantees: Suppose that there exists $b\in\Z$ and a constant $\beta'<\frac{1}{4}$, such that $\Pr[z_i\notin B_{b-1}\cup B_b \cup B_{b+1}]\leq \beta'$ for any fixed $i\in[n]$. Then for all $\eps,\beta,\delta\in (0,1)$, if 
    \[
    n\geq\frac{C}{\eps}\log\frac{1}{\beta\delta},
    \] 
    then $\cA_{\eps,\delta}(z_1,\dots,z_k)\in \{b-1,b,b+1\}$ with probability at least $1-\beta$.
    \end{lemma}

    Let $C$ and $\cA_{\eps,\delta}$ be as in \Cref{lem:stablehistogram}. We define $\cM^f(Z)$ as follows: \textit{Partition $Z$ into $k=\frac{C}{\eps}\log\frac{1}{\beta\delta}$ datasets $\wt Z_1,\dots, \wt Z_k$ of equal size. Let $y_i\gets f(\wt Z_i)/n$ for each $i\in[k]$ and let $b^*\gets \cA_{\eps,\delta}(y_1,\dots, y_k)$ run with buckets $B_i$ defined by a partition of $\R$ into intervals of size $\alpha$. Release the midpoint of $B_{b^*}$} 

    Privacy of $\cM$ follows immediately from the privacy guarantees of $\cA_{\eps,\delta}$. To see why accuracy holds, observe that by our setting of $n$, each bucket $B_i$ consists of at least $N(\alpha) + \frac{\log 5}{\alpha^2}$ i.i.d.\ samples from $\cD$. Hence, by our assumption on $f$ and $\cD$, there exists a $b\in \N$ such that $\nu\in B^* = B_{b-1}\cup B_b \cup B_{b+1}$ and $\Pr\brackets{y_i\not\in B^*}\leq 1/5$ for each $i\in[k]$, and therefore, $\cA_{\eps,\delta}(y_1,\dots, y_k)$ returns a $b^*\in\set{b-1,b,b+1}$ with probability at least $1-\beta$. Since $\nu\in B^*$, the midpoint of $B_{b^*}$ is at distance at most $3\alpha$ from $\nu$, and $\abs{\cM^f(Z)-\nu}\leq 3\alpha$ with probability at least $1-\beta$. Runtime follows by inspection of $\cM$ and algorithm 3 in Appendix C of \cite{BrownGSUZ21}.
\end{proof}

\subsection{Eigenvalue estimation}

We compare with \cref{thm:eigenvalue_est}, which with probability $1-\beta$ gives a $1\pm \alpha$ approximation to the $i$-th eigenvalue of the covariance given 
\[
n=\Omega\paren{\frac{1}{\alpha^2}\paren{\frac{d + \log(1/\beta)}{p} + \frac{\log(1/\delta)}{\eps} + \frac{\log(p^{-1}\log 1/\delta)}{p}}}
\]
samples from a $d$-dimensional subgaussian distribution.
Here $p\in (0,1/4)$ is the subsampling hyperparameter
(Throughout this section, we suppress the subgaussian constant $K_{\mathcal{D}}$.)
The algorithm runs in time $\exp\paren{O\paren{\frac{p\log1/\delta}{\eps}}}\poly\paren{\frac{\log1/\delta}{p}}$.

Subsample-and-aggregate yields the same guarantee with
\[
    n \gtrsim \Omega\paren{\frac{ d\log(1/\delta\beta)}{\alpha^2\eps}}.
\]

Off-the-shelf, any spectral approximation $\hat\Sigma$ of the true covariance $\Sigma$ immediately leads to multiplicative eigenvalue approximations.
if $(1-\alpha) \Sigma \preceq \hat\Sigma \preceq (1+\alpha) \Sigma$ then for all $i\in [d]$
\[
    (1-\alpha) \lambda_i(\Sigma) \le \lambda_i(\hat\Sigma)\le (1+\alpha)\lambda_i(\Sigma).
\]
This is because the L{\"o}wner order ``$\preceq$'' is monotone under $\lambda_i$.
It is beyond the scope of this work to survey the landscape of differentially private covariance estimation.
In the setting we consider, producing such a private approximation to the entire covariance requires $n\gtrsim d^{3/2}$ samples, a polynomial-in-$d$ overhead. 
See \cite{DworkTTZ14,Narayanan24,portella2025lower} for lower bounds and some discussion of existing algorithms.

The StableCovariance estimator of \cite{BrownHS23} gives a direct way to estimate a single eigenvalue.
Once $n\gtrsim \frac{\log(1/\delta)}{\eps}$, on adjacent datasets $Z$ and $Z'$ the algorithm returns\footnote{Formally, the algorithm either fails or returns covariance estimates which satisfy this guarantee. On data from a subgaussian distribution, with high probability it does not fail.} non-private covariance estimates $\Sigma_1$ and $\Sigma_2$ such that $(1-\gamma) \Sigma_1 \preceq \Sigma_2 \preceq (1+\gamma) \Sigma_1$ for $\gamma = O(d/n)$.
Once $n$ is large enough to make $\gamma$ a small constant, this leads to 
\[
    \abs{\log\lambda_i(\Sigma_1) - \log\lambda_i(\Sigma_2)} \le O(d/n).
\]
Therefore, we can add Laplace noise to $\log(\lambda_i(\Sigma_1))$ and obtain a multiplicative $1\pm \alpha$ approximation to the underlying $i$-th eigenvalue as long as 
\[
n \gtrsim \frac{d}{\alpha^2} + \frac{d}{\alpha \eps} + \frac{d \log(1/\delta)}{\eps}.
\]

Finally, we note that under additional assumptions on the data one can directly bound the global sensitivity of an individual eigenvalue. 
For example, if all data points are clipped to enforce $\norm{Z_i}{2}\le R$ then for datasets $Z$ and $Z'$ differing in index $1$ we can relate their covariance matrices as
\[
    \opnorm{\frac 1 n \sum_i Z_i Z_i^T - \frac 1 n \sum_i Z_i' Z_i^{'T}}
        = \frac 1 n \opnorm{Z_1 Z_1^T - Z'_1 Z_1^{'T}} \le \frac{2R^2}{n}
\]
(using the triangle inequality and the fact that $\opnorm{Z_i Z_i} = \norm{Z_i}{2}^2 \le R^2$).
Thus, Weyl's inequality (\cref{lem:weyl}) gives 
\[
    \abs{\lambda_i(\Sigma_Z) - \lambda_i(\Sigma_{Z'})} \le \opnorm{\Sigma_Z - \Sigma_{Z'}} \le \frac{2R^2}{n}.
\]
Thus, we can add Laplace noise with scale $2R^2/\eps n$ for a pure-DP additive approximation to the eigenvalue.
This is an important step in multiple algorithms for private covariance estimation: see \cite{mangoubi2022private,dong2022differentially} and references therein.

\subsection{Testing problems}

\cref{thm:loss_testing,thm:param_testing} concern deciding between two or three alternate hypotheses.
In this setting, subsample-and-aggregate takes a simple form as one can aggregate with a DP histogram, where each ``bin'' is a hypothesis.
Adding Laplace noise to each bin with scale $1/\eps$ yields pure DP and, for a constant number of hypotheses, every noisy count will be within $O(\log (1/\beta)/\eps)$ with probability at least $1-\beta$.
One can use any variation (such as Stable Histogram as above) with subsample-and-aggregate, but all approaches yield at least a $1/\eps$ factor blowup in sample complexity beyond the non-private cost.

There is substantial work on hypothesis testing and selection under differential privacy \cite[see, e.g.,][]{CanonneKMSU19,BunKSW19,kazan2023test,asi2024universally}. 
There are several approaches among these which bypass the $1/\eps$ blowup for specific testing problems, but (to the best of our knowledge) these all rely on detailed distributional information that is not available in the settings we consider.

Several approaches consider hypothesis testing specifically for linear regression \cite{Sheffet17,FerrandoWS22,alabi2023differentially,pena2022differentially}.
These are analyzed in various settings but, to the best of our knowledge, when applied to the problems we consider all would incur at least a $\sqrt{d}$ or $1/\eps$ factor increase above the non-private sample complexity.

\subsection{Loss estimation}

\cref{thm:loss_estimation} gives our result for estimating the minimum achievable loss on a learning task under \cref{ass:loss_concentration}, a loss-concentration assumption.
Under this assumption, \cref{prop:ssa_baseline} gives us the sample complexity for subsample-and-aggregate.

The other off-the-shelf approach, for any loss estimation task, is to privately learn a near-optimal parameter vector; this is a central topic in differentially private learning and for many learning tasks there are efficient algorithms that avoid the $\frac{\log 1/\delta}{\eps}$ penalty of subsample-and-aggregate.
For any fixed parameter vector, one can expect to estimate the loss using very few additional samples.

One alternate approach comes from the observation that, to estimate the loss of a near-optimal model, we only need its predictions on fresh data.
The paradigm of \emph{private prediction} \cite{dwork2018privacy} enables exactly that, giving predictions which are close to those of the best-possible predictor.
Work in this area shows how to accomplish this with significantly fewer samples than required for learning.
For applications in the style we consider, one would need results for agnostic PAC learning \cite{nandi2020privately,dagan2020pac}, as in the realizable setting the optimal loss is zero by assumption. 

Beyond this, we are aware of few techniques that permit estimating the minimum loss of a learning task with a sample complexity that beats subsample-and-aggregate, under any assumptions on the learning task.
\cite{edmonds2022learning} considers the problem of \emph{refutation} under local differential privacy. In the context of binary classification, refutation consists of determining if the loss is $1/2$ or $1/2-\Omega(1)$.
\cite{varshney2022nearly} provides a private gradient descent method for well-specified linear regression and the algorithm tracks $\sigma^2$, the variance of the label noise, implicitly.

\subsection{Single-parameter estimation}

We compare with the single-parameter estimation result of \cref{thm:theta1_estimation}, which relies on \Cref{ass:id}.
To provide comparable guarantees, subsample-and-aggregate requires only the first part of that assumption, which is itself a reformulation of \Cref{ass:intro_concentration}.
\cref{prop:ssa_baseline} tells us that, under that assumption, subsample-and-aggregate can return an estimate $\tilde\theta_1$ such that,
if 
\[
    n = \Omega\left(\frac{N(\alpha) \log(1/\beta\delta)}{\eps} + \frac{\log(1/\beta \delta)}{\alpha^2 \eps}\right)
\]
then with probability at least $1-\beta$ we have $|\tilde\theta_1 - \theta(\cD)_1|\le 3\alpha$.

We are not aware of any algorithms that rely on (any subset of) \Cref{ass:id} for this task.
The closest work, from which our work draws an example, is that of \cite{AsiDT25}.
They provide a sophisticated algorithm for privately estimating a single parameter within a large parametric model.
It is not immediately clear the exact sample complexity their algorithm would obtain in our setting (their work focuses on achieving the optimal error rate as $n\to \infty$ under different, but related, assumptions), but it appears unable to provide accurate estimates at a sample size smaller than that of subsample-and-aggregate.

Our primary motivation is estimating a single parameter in a least squares model. 
As with covariance estimation, any private estimator for the full least-squares model can be used for single-parameter estimation.
And, as with covariance estimation, the landscape of DP least squares is too broad to survey here.
We mention a few key themes with pointers to recent work. 
Many techniques perturb the sufficient statistics $X^TX, X^Ty$; accurate estimation here requires as many samples as private covariance estimation, incurring a $d^{3/2}$ in the sample complexity \cite{Wang18,Sheffet19,tang2024improved}.
Work building on robust statistics requires either exponential time \cite{liu2022differential} or $d^2$ samples \cite{AndersonBMT25}, and this is believed to be inherent.
A line of work analyzes gradient descent under distributional assumptions \cite{varshney2022nearly,liu2023label,bombari2026high}, and some of these arguments characterize output noise in a way that may be amenable to direct calculation about per-coordinate error \cite{brown2024private}. 
However, all are first-order methods and have an dependence on the condition number of the covariance that \cref{thm:theta1_estimation} does not.
The approach of \cite{BrownHHKLOPS24} admits clean reasoning about single parameters (see Claim 3), but requires $n\ge \frac{d}{\eps^2} \log^2(1/\delta)$ samples to produce a non-trivial estimate, a larger requirement than subsample-and-aggregate.

\section{Analysis of random-design linear regression}
\label{sec:regression_facts}

In this section, we provide the formal claims necessary to show that \Cref{ass:id} applies to the task of well-specified random-design linear regression with Gaussian covariates and Gaussian noise.

\begin{definition}[Random-design linear regression task]\label{def:random_design}
Fix dimension $d\in\mathbb{N}$. 
Examples are $z= (x,y)\in \cZ = \R^d \times \R$.
We define $\Theta = \R^d$ and let $\ell_\theta(z) = (\langle x,\theta\rangle - y)^2$ be the standard squared loss.
The distribution $\cD$ over $\cZ$ is defined by the symmetric positive definite matrix $\Sigma\in\R^{d\times d}$, noise variance $s^2>0$, and true parameter $\theta\in\R^d$.
To sample $(x,y)\sim \cD$ we draw $x_i\sim \cN(0,\Sigma)$, $\eta_i \sim \cN(0,s^2)$, and set $y_i = \langle x_i,\theta^*\rangle + \eta_i$.
\end{definition}

\begin{proposition}\label{prop:regression_works}
    Let $(\Theta,\ell,\cD)$ be a learning task as above, with $\cD$ defined by $\Sigma\succ 0, s^2>0$, and $\theta\in\R^d$.
    Then $(\Theta, \ell, \cD)$ satisfies \Cref{ass:id} with 
    \begin{itemize}
        \item $\Nid = \Mid = 4d$ (independent of $\alpha$).
        \item $\Kid = 144 \max\{1, s^2\}$
        \item $\lambda = 2 s^2$
        \item $\mu = \frac{1}{(\Sigma^{-1})_{11}}$
        \item $\sigma = 2s$.
    \end{itemize}
\end{proposition}

\Cref{ass:id} has four conditions; we establish them in order in \Cref{lem:part1,lem:part2,lem:part3,lem:part4}.
To begin, we state a few standard concentration inequalities.

\begin{claim}[Gaussian tail]\label{clm:gauss-tail}
If $G\sim\cN(0,1)$ then for all $\beta\in(0,1)$,
$\Pr\!\left(G^2\ge 2\log\frac{2}{\beta}\right)\le \beta$.
\end{claim}

We use a slightly weaker (and simpler) version of the standard Laurent--Massart tail bounds for chi-squared random variables.
\begin{claim}[Laurent--Massart]\label{clm:LM}
If $U\sim\chi^2_k$ then for all $t\ge 0$, $\Pr(U \ge k + 2\sqrt{kt} + 2t)$ and $\Pr\paren{U\le k - 2\sqrt{kt}}\le e^{-t}$.
\end{claim}

\begin{claim}[Inverse-Wishart scalar identity]\label{clm:invwishart-scalar}
Let $G\in\R^{n\times d}$ have i.i.d.\ $\cN(0,1)$ entries and let $W=G^\top G$ (so $W\sim \mathrm{Wishart}(I_d,n)$).
For any fixed nonzero $v\in\R^d$,
\[
\frac{\|v\|_2^2}{v^\top W^{-1}v}\sim \chi^2_{n-d+1}.
\]
\end{claim}
This identity is standard; it can be proved by rotating $v$ to $e_1$ and using the Bartlett decomposition for Wishart distributions.

\begin{lemma}[Concentration of $\hat\theta_1$]\label{lem:part1}
Let $\hat{\theta}$ be the OLS estimator.
In the setting of \Cref{def:random_design},
for all $\alpha,\beta\in(0,1)$, if
\[
n\;\ge\; d + \frac{24 \max\set{1,s^2} }{\alpha^2}\log(4/\beta)
\]
we have
\[
\Pr\!\left(\bigl|\hat\theta_1-\theta_1\bigr|\ge \alpha\,\mu^{-1/2}\right)\le \beta,
\qquad
\text{where }\ \ \mu=\frac{1}{(\Sigma^{-1})_{11}}.
\]
\end{lemma}

\begin{proof}
Conditioned on $X$, the OLS estimator satisfies
\[
\hat\theta-\theta=(X^\top X)^{-1}X^\top \eta,
\qquad \eta\sim \cN(0,s^2I_n),
\]
hence
\[
(\hat\theta_1-\theta_1)\mid X \sim \cN\!\Bigl(0,\ s^2\,((X^\top X)^{-1})_{11}\Bigr).
\]
Whiten and write $X=G\Sigma^{1/2}$ with i.i.d.\ standard normal $G$, so
\[
(X^\top X)^{-1}=\Sigma^{-1/2}(G^\top G)^{-1}\Sigma^{-1/2}.
\]
For the first standard unit vector $e_1$, let $v=\Sigma^{-1/2}e_1$, so $\|v\|_2^2=e_1^\top\Sigma^{-1}e_1=(\Sigma^{-1})_{11}$ and
\[
((X^\top X)^{-1})_{11} = v^\top (G^\top G)^{-1}v.
\]
By \Cref{clm:invwishart-scalar}, with $k\coloneqq n-d+1$,
\[
\frac{\|v\|_2^2}{v^\top (G^\top G)^{-1}v}\sim \chi^2_k.
\]
Let $U\sim\chi^2_k$ denote this variable; then
\[
v^\top (G^\top G)^{-1}v = \frac{\|v\|_2^2}{U}=\frac{(\Sigma^{-1})_{11}}{U}.
\]
Using \Cref{clm:LM} (lower tail) with $t=\log(2/\beta)$,
\[
\Pr\!\left(U\le k-2\sqrt{k\log\frac{2}{\beta}}\right)\le \frac{\beta}{2}.
\]
If $k\ge 16\log(2/\beta)$, then $2\sqrt{k\log(2/\beta)}\le k/2$ and thus
\[
\Pr\!\left(U\ge \frac{k}{2}\right)\ge 1-\frac{\beta}{2}.
\]
On this event,
\[
((X^\top X)^{-1})_{11} = \frac{(\Sigma^{-1})_{11}}{U}\le \frac{2(\Sigma^{-1})_{11}}{k}\le \frac{4(\Sigma^{-1})_{11}}{n}.
\]
Therefore, on this event,
\[
\mathrm{Var}(\hat\theta_1-\theta_1\mid X)\le s^2\cdot \frac{4(\Sigma^{-1})_{11}}{n}.
\]
Conditioning on $X$ and applying a Gaussian tail bound gives, for any $t>0$, 
\[
\Pr\!\left(|\hat\theta_1-\theta_1|\ge t \mid X\right)
\le 2\exp\!\left(-\frac{t^2 n}{8s^2(\Sigma^{-1})_{11}}\right).
\]
Unconditioning and adding the $\beta/2$ failure probability of the event yields
\[
\Pr\!\left(|\hat\theta_1-\theta_1|\ge t\right)
\le 2\exp\!\left(-\frac{t^2 n}{8s^2(\Sigma^{-1})_{11}}\right) + \frac{\beta}{2}.
\]
Now $\mu=\bigl((\Sigma^{-1})_{11}\bigr)^{-1}$ and choose $t=\alpha\mu^{-1/2}$, so the exponential term becomes $2\exp\paren{-\frac{\alpha^2 n}{8 s^2}}$.
If $n\ge \frac{8 s^2 \log (4/\beta)}{\alpha^2}$ then this is at most $\beta/2$.

Our other restriction on $n$ (via $k=n-d+1$) is that $n\ge d-1 + 16 \log(2/\beta)$.
Taking $n\ge d + \frac{24 \max\set{1,s^2} \log(4/\beta)}{\alpha^2}$ ensures both of these conditions are satisfied.
\end{proof}

\begin{lemma}[Typical Smoothness and Strong Convexity]\label{lem:part2}

In the setting of \Cref{def:random_design}, if 
$n \ge 4d + 18 \log (2/\beta)$
then with probability at least $1-\beta$ we have, for all $w\in \R$, 
\[
    \frac{\mu}{2}(w-\hat\theta_1)^2 \;\le\; L^{(w)}(Z)-L(Z)\;\le\; 2\mu (w-\hat\theta_1)^2,
\]
where $\mu:=\frac{1}{(\Sigma^{-1})_{11}}$.
Recall that $L^{(w)}(Z)$ denotes the minimum achievable loss among parameters whose first coordinate is fixed to $w$.

\end{lemma}

\begin{proof}
Write the design matrix as $X=(x^{(1)},X_{-1})$, where $x^{(1)}\in\R^n$ is the first column and $X_{-1}\in\R^{n\times(d-1)}$ are the remaining columns.
For a fixed $w\in\R$, the constrained least-squares problem defining $L^{(w)}(Z)$ is
\[
L^{(w)}(Z)=\min_{\theta_{-1}\in\R^{d-1}}\frac1n\|y-x^{(1)}w-X_{-1}\theta_{-1}\|_2^2.
\]
Let $P$ denote the orthogonal projector onto $\mathrm{col}(X_{-1})$ and set $M:=I-P$ (so $M$ is an orthogonal projector, hence symmetric and idempotent).
Then
\[
L^{(w)}(Z)=\frac1n\|M(y-x^{(1)}w)\|_2^2.
\]
Define $r\coloneqq My$ and $v\coloneqq Mx^{(1)}$, so $L^{(w)}(Z)=\frac1n\|r-vw\|_2^2$.
As a function of $w$, this is a one-dimensional quadratic with minimizer $w=\hat w$, where $\hat w$ is exactly the first coordinate of the unconstrained OLS solution, i.e.\ $\hat w=\hat\theta_1$.
Therefore,
\[
L^{(w)}(Z)-L(Z)
=\frac1n\Big(\|r-vw\|_2^2-\|r-v\hat\theta_1\|_2^2\Big)
=\frac{v^\top v}{n}\,(w-\hat\theta_1)^2.
\]

Thus, it remains to control the random scalar
\[
c_Z\coloneqq \frac{v^\top v}{n}=\frac{1}{n}x^{(1)\top}Mx^{(1)}.
\]

We now show that $c_Z$ concentrates around $\mu\coloneqq 1/(\Sigma^{-1})_{11}$. Because $(x_1,\dots,x_d)$ is jointly Gaussian with covariance $\Sigma$, the conditional variance of the first coordinate given the others is
\[
\mathrm{Var}(x_1 \mid x_{2:d}) = \frac{1}{(\Sigma^{-1})_{11}} = \mu.
\]
Conditioned on $X_{-1}$, the residualized column $v=Mx^{(1)}$ is Gaussian in $\mathrm{range}(M)$ with covariance $\mu M$:
\[
v\mid X_{-1}\sim \cN(0,\mu M).
\]
Since $M$ is an idempotent projector with rank $\rank(M)=n-(d-1)=k$, it follows that
\[
\frac{1}{\mu}v^\top v \ \big|\ X_{-1} \sim \chi^2_k,
\]
and hence unconditionally $(1/\mu)v^\top v\sim \chi^2_k$.
Let $U\sim\chi^2_k$ denote this variable, so $c_Z=\mu U/n$. Apply \Cref{clm:LM} with $t=\log(2/\beta)$:
\[
\Pr\!\left(U\notin\left[k-2\sqrt{k\log\frac{2}{\beta}},\ k+2\sqrt{k\log\frac{2}{\beta}}+2\log\frac{2}{\beta}\right]\right)\le \beta.
\]

Now we need to impose two conditions. First condition we require is $k - 2 \sqrt{k t} \ge n/ 2$, or equivalently, since $k = n - d + 1$, 
\begin{equation*}
n - d + 1 - 2 \sqrt{(n - d + 1) t} \ge n/2 \, .
\end{equation*}
In order for the above bound to hold, it suffices to take $n \ge 4(d-1) + 18t$.

The other constraint we need to impose is that
$k + 2 \sqrt{k t} + 2t \le 2n$, since $k \le n$, in order to ensure it suffices to take $n\ge 16 t$. Therefore, as long as $n \ge 4d + 18 \log(2/\beta)$, we will have 

\begin{equation*}
\frac{1}{2} \le \frac{U}{n} \le 2 \, 
\end{equation*}
and therefore,
\begin{equation*}
\frac{\mu}{2} \le c_Z \le 2 \mu \, 
\end{equation*}
as desired.
\end{proof}

\begin{lemma}[Leave-one-out loss]\label{lem:part3}
In the setting of \Cref{def:random_design}, assume $n\ge d+1$,
fix $i\in[n]$, and let $Z_{-i}$ be the dataset with point $i$ removed. 
Recall $L_\theta(Z_{-i})=\frac1n\sum_{j\ne i}(y_j-x_j^\top\theta)^2$.
Then with probability at least $1-\beta$ we have
\[
\Pr\!\left(|L(Z)-L(Z_{-i})|\ge \frac{\lambda\log(2/\beta)}{n}\right)\le \beta
\qquad\text{with}\qquad \lambda=2s^2.
\]
\end{lemma}

\begin{proof}
By the standard Sherman--Morrison identity for deletions in least-squares, we have
\[
    L(Z)-L(Z_{-i}) = \frac{1}{n}\cdot \frac{r_i^2}{1-h_{ii}},
\]
where $r_i=(y_i - x_i^T \hat\theta(Z))$ is the residual and $h_{i} = x_i^T (X^TX)^{-1}x_i$ the leverage score.
Note that this is always positive.
For any fixed dataset $X$, we use the following distributional fact:
\[
    \frac{r_i}{\sqrt{1-h_{i}}}~\Big\vert~X \sim \cN(0,s^2).
\]
In particular, its distribution is independent of $X$, so the same statement holds unconditionally.

Thus we can apply \Cref{clm:gauss-tail}: for $G\sim \cN(0,1)$ we have
\begin{align*}
    \Pr\left[L(Z) - L(Z_{-i}) \ge \frac{2 s^2 \log(2/\beta)}{n} \right]
        &= \Pr\brackets{\frac{s^2}{n}\cdot G^2 \ge \frac{2 s^2 \log(2/\beta)}{n}} \\
        &= \Pr\brackets{ G^2 \ge 2 \log(2/\beta)} \\
        &\le \beta.
\end{align*}
Thus we can take $\lambda = 2s^2$.
\end{proof}

\begin{lemma}[Tail bound for $L(Z)$]\label{lem:part4}
In the setting of \Cref{def:random_design}, for all $\beta\in(0,1)$ we have
\[
\Pr\!\left(L(Z)\ge \sigma^2\left(1+\frac{\log(1/\beta)}{n}\right)\right)\le \beta
\qquad\text{with}\qquad \sigma^2:=4s^2.
\]
\end{lemma}

\begin{proof}
    The loss of the empirical minimizer $\hat\theta$ on $Z$ is at most the loss of the true minimizer $\theta$:
    \begin{align}
        L(Z) = L_{\hat\theta}(Z) &\le L_{\theta}(Z) \nonumber \\
            &= \frac 1 n \sum_{i=1}^n \paren{x_i^T \theta - y_i}^2 \nonumber \\
            &= \frac 1 n \sum_{i=1}^n \eta_i^2,\label{eq:chi_square_loss}
    \end{align}
    where we have used the fact that $y_i = x_i^T \theta + \eta_i$.
    Since each $\eta_i$ is drawn i.i.d.\ from $\cN(0,s^2)$, \cref{eq:chi_square_loss} is distributed as $\frac{s^2}{n} U$ where $U\sim \chi_n^2$. 
    Thus
    \begin{align*}
        \Pr\left[L(Z) \ge 2s^2 + \frac{4 t s^2}{n}\right]
            &\le \Pr\left[ \frac{s^2}{n}U \ge 2s^2 + \frac{4 t s^2}{n}\right] \\
            &= \Pr\left[ U \ge 2n  + 4t\right] \\
            &\le e^{-t}.
    \end{align*}
    Taking $t=\log(1/\beta)$, we are done.
\end{proof}

\ifnum\neurips=1
\newpage
\section*{NeurIPS Paper Checklist}

\begin{enumerate}

\item {\bf Claims}
    \item[] Question: Do the main claims made in the abstract and introduction accurately reflect the paper's contributions and scope?
    \item[] Answer: \answerYes{} %
    \item[] Justification: Yes. Our main claims are statements about the theorems we prove and how they extend and improve upon prior work.  %
    \item[] Guidelines:
    \begin{itemize}
        \item The answer \answerNA{} means that the abstract and introduction do not include the claims made in the paper.
        \item The abstract and/or introduction should clearly state the claims made, including the contributions made in the paper and important assumptions and limitations. A \answerNo{} or \answerNA{} answer to this question will not be perceived well by the reviewers. 
        \item The claims made should match theoretical and experimental results, and reflect how much the results can be expected to generalize to other settings. 
        \item It is fine to include aspirational goals as motivation as long as it is clear that these goals are not attained by the paper. 
    \end{itemize}

\item {\bf Limitations}
    \item[] Question: Does the paper discuss the limitations of the work performed by the authors?
    \item[] Answer: \answerYes{} %
    \item[] Justification: Yes, our work is entirely theoretical and we take care to lay out what our analysis does and does not cover. %
    \item[] Guidelines:
    \begin{itemize}
        \item The answer \answerNA{} means that the paper has no limitation while the answer \answerNo{} means that the paper has limitations, but those are not discussed in the paper. 
        \item The authors are encouraged to create a separate ``Limitations'' section in their paper.
        \item The paper should point out any strong assumptions and how robust the results are to violations of these assumptions (e.g., independence assumptions, noiseless settings, model well-specification, asymptotic approximations only holding locally). The authors should reflect on how these assumptions might be violated in practice and what the implications would be.
        \item The authors should reflect on the scope of the claims made, e.g., if the approach was only tested on a few datasets or with a few runs. In general, empirical results often depend on implicit assumptions, which should be articulated.
        \item The authors should reflect on the factors that influence the performance of the approach. For example, a facial recognition algorithm may perform poorly when image resolution is low or images are taken in low lighting. Or a speech-to-text system might not be used reliably to provide closed captions for online lectures because it fails to handle technical jargon.
        \item The authors should discuss the computational efficiency of the proposed algorithms and how they scale with dataset size.
        \item If applicable, the authors should discuss possible limitations of their approach to address problems of privacy and fairness.
        \item While the authors might fear that complete honesty about limitations might be used by reviewers as grounds for rejection, a worse outcome might be that reviewers discover limitations that aren't acknowledged in the paper. The authors should use their best judgment and recognize that individual actions in favor of transparency play an important role in developing norms that preserve the integrity of the community. Reviewers will be specifically instructed to not penalize honesty concerning limitations.
    \end{itemize}

\item {\bf Theory assumptions and proofs}
    \item[] Question: For each theoretical result, does the paper provide the full set of assumptions and a complete (and correct) proof?
    \item[] Answer: \answerYes{} %
    \item[] Justification: All of our theoretical results state their assumptions and follow the standard theorem-lemma-proof structure.%
    \item[] Guidelines:
    \begin{itemize}
        \item The answer \answerNA{} means that the paper does not include theoretical results. 
        \item All the theorems, formulas, and proofs in the paper should be numbered and cross-referenced.
        \item All assumptions should be clearly stated or referenced in the statement of any theorems.
        \item The proofs can either appear in the main paper or the supplemental material, but if they appear in the supplemental material, the authors are encouraged to provide a short proof sketch to provide intuition. 
        \item Inversely, any informal proof provided in the core of the paper should be complemented by formal proofs provided in appendix or supplemental material.
        \item Theorems and Lemmas that the proof relies upon should be properly referenced. 
    \end{itemize}

    \item {\bf Experimental result reproducibility}
    \item[] Question: Does the paper fully disclose all the information needed to reproduce the main experimental results of the paper to the extent that it affects the main claims and/or conclusions of the paper (regardless of whether the code and data are provided or not)?
    \item[] Answer: \answerNA{} %
    \item[] Justification: We have no experiments. %
    \item[] Guidelines:
    \begin{itemize}
        \item The answer \answerNA{} means that the paper does not include experiments.
        \item If the paper includes experiments, a \answerNo{} answer to this question will not be perceived well by the reviewers: Making the paper reproducible is important, regardless of whether the code and data are provided or not.
        \item If the contribution is a dataset and\slash or model, the authors should describe the steps taken to make their results reproducible or verifiable. 
        \item Depending on the contribution, reproducibility can be accomplished in various ways. For example, if the contribution is a novel architecture, describing the architecture fully might suffice, or if the contribution is a specific model and empirical evaluation, it may be necessary to either make it possible for others to replicate the model with the same dataset, or provide access to the model. In general. releasing code and data is often one good way to accomplish this, but reproducibility can also be provided via detailed instructions for how to replicate the results, access to a hosted model (e.g., in the case of a large language model), releasing of a model checkpoint, or other means that are appropriate to the research performed.
        \item While NeurIPS does not require releasing code, the conference does require all submissions to provide some reasonable avenue for reproducibility, which may depend on the nature of the contribution. For example
        \begin{enumerate}
            \item If the contribution is primarily a new algorithm, the paper should make it clear how to reproduce that algorithm.
            \item If the contribution is primarily a new model architecture, the paper should describe the architecture clearly and fully.
            \item If the contribution is a new model (e.g., a large language model), then there should either be a way to access this model for reproducing the results or a way to reproduce the model (e.g., with an open-source dataset or instructions for how to construct the dataset).
            \item We recognize that reproducibility may be tricky in some cases, in which case authors are welcome to describe the particular way they provide for reproducibility. In the case of closed-source models, it may be that access to the model is limited in some way (e.g., to registered users), but it should be possible for other researchers to have some path to reproducing or verifying the results.
        \end{enumerate}
    \end{itemize}

\item {\bf Open access to data and code}
    \item[] Question: Does the paper provide open access to the data and code, with sufficient instructions to faithfully reproduce the main experimental results, as described in supplemental material?
    \item[] Answer: \answerNA{} %
    \item[] Justification: We have no experiments. %
    \item[] Guidelines:
    \begin{itemize}
        \item The answer \answerNA{} means that paper does not include experiments requiring code.
        \item Please see the NeurIPS code and data submission guidelines (\url{https://neurips.cc/public/guides/CodeSubmissionPolicy}) for more details.
        \item While we encourage the release of code and data, we understand that this might not be possible, so \answerNo{} is an acceptable answer. Papers cannot be rejected simply for not including code, unless this is central to the contribution (e.g., for a new open-source benchmark).
        \item The instructions should contain the exact command and environment needed to run to reproduce the results. See the NeurIPS code and data submission guidelines (\url{https://neurips.cc/public/guides/CodeSubmissionPolicy}) for more details.
        \item The authors should provide instructions on data access and preparation, including how to access the raw data, preprocessed data, intermediate data, and generated data, etc.
        \item The authors should provide scripts to reproduce all experimental results for the new proposed method and baselines. If only a subset of experiments are reproducible, they should state which ones are omitted from the script and why.
        \item At submission time, to preserve anonymity, the authors should release anonymized versions (if applicable).
        \item Providing as much information as possible in supplemental material (appended to the paper) is recommended, but including URLs to data and code is permitted.
    \end{itemize}

\item {\bf Experimental setting/details}
    \item[] Question: Does the paper specify all the training and test details (e.g., data splits, hyperparameters, how they were chosen, type of optimizer) necessary to understand the results?
    \item[] Answer: \answerNA{} %
    \item[] Justification: We have no experiments. %
    \item[] Guidelines:
    \begin{itemize}
        \item The answer \answerNA{} means that the paper does not include experiments.
        \item The experimental setting should be presented in the core of the paper to a level of detail that is necessary to appreciate the results and make sense of them.
        \item The full details can be provided either with the code, in appendix, or as supplemental material.
    \end{itemize}

\item {\bf Experiment statistical significance}
    \item[] Question: Does the paper report error bars suitably and correctly defined or other appropriate information about the statistical significance of the experiments?
    \item[] Answer: \answerNA{} %
    \item[] Justification: We have no experiments. %
    \item[] Guidelines:
    \begin{itemize}
        \item The answer \answerNA{} means that the paper does not include experiments.
        \item The authors should answer \answerYes{} if the results are accompanied by error bars, confidence intervals, or statistical significance tests, at least for the experiments that support the main claims of the paper.
        \item The factors of variability that the error bars are capturing should be clearly stated (for example, train/test split, initialization, random drawing of some parameter, or overall run with given experimental conditions).
        \item The method for calculating the error bars should be explained (closed form formula, call to a library function, bootstrap, etc.)
        \item The assumptions made should be given (e.g., Normally distributed errors).
        \item It should be clear whether the error bar is the standard deviation or the standard error of the mean.
        \item It is OK to report 1-sigma error bars, but one should state it. The authors should preferably report a 2-sigma error bar than state that they have a 96\% CI, if the hypothesis of Normality of errors is not verified.
        \item For asymmetric distributions, the authors should be careful not to show in tables or figures symmetric error bars that would yield results that are out of range (e.g., negative error rates).
        \item If error bars are reported in tables or plots, the authors should explain in the text how they were calculated and reference the corresponding figures or tables in the text.
    \end{itemize}

\item {\bf Experiments compute resources}
    \item[] Question: For each experiment, does the paper provide sufficient information on the computer resources (type of compute workers, memory, time of execution) needed to reproduce the experiments?
    \item[] Answer: \answerNA{} %
    \item[] Justification: We have no experiments. %
    \item[] Guidelines:
    \begin{itemize}
        \item The answer \answerNA{} means that the paper does not include experiments.
        \item The paper should indicate the type of compute workers CPU or GPU, internal cluster, or cloud provider, including relevant memory and storage.
        \item The paper should provide the amount of compute required for each of the individual experimental runs as well as estimate the total compute. 
        \item The paper should disclose whether the full research project required more compute than the experiments reported in the paper (e.g., preliminary or failed experiments that didn't make it into the paper). 
    \end{itemize}
    
\item {\bf Code of ethics}
    \item[] Question: Does the research conducted in the paper conform, in every respect, with the NeurIPS Code of Ethics \url{https://neurips.cc/public/EthicsGuidelines}?
    \item[] Answer: \answerYes{} %
    \item[] Justification: Yes. Our work is entirely theoretical and deals with no human subjects.  Our work is designed to advance the state-of-the-art in private learning, which we expect to have positive social effects. We see no potential social harms. %
    \item[] Guidelines:
    \begin{itemize}
        \item The answer \answerNA{} means that the authors have not reviewed the NeurIPS Code of Ethics.
        \item If the authors answer \answerNo, they should explain the special circumstances that require a deviation from the Code of Ethics.
        \item The authors should make sure to preserve anonymity (e.g., if there is a special consideration due to laws or regulations in their jurisdiction).
    \end{itemize}

\item {\bf Broader impacts}
    \item[] Question: Does the paper discuss both potential positive societal impacts and negative societal impacts of the work performed?
    \item[] Answer: \answerYes{} %
    \item[] Justification: Our work is designed to advance the state-of-the-art in private learning, which we expect to have positive social effects. We see no potential social harms. %
    \item[] Guidelines:
    \begin{itemize}
        \item The answer \answerNA{} means that there is no societal impact of the work performed.
        \item If the authors answer \answerNA{} or \answerNo, they should explain why their work has no societal impact or why the paper does not address societal impact.
        \item Examples of negative societal impacts include potential malicious or unintended uses (e.g., disinformation, generating fake profiles, surveillance), fairness considerations (e.g., deployment of technologies that could make decisions that unfairly impact specific groups), privacy considerations, and security considerations.
        \item The conference expects that many papers will be foundational research and not tied to particular applications, let alone deployments. However, if there is a direct path to any negative applications, the authors should point it out. For example, it is legitimate to point out that an improvement in the quality of generative models could be used to generate Deepfakes for disinformation. On the other hand, it is not needed to point out that a generic algorithm for optimizing neural networks could enable people to train models that generate Deepfakes faster.
        \item The authors should consider possible harms that could arise when the technology is being used as intended and functioning correctly, harms that could arise when the technology is being used as intended but gives incorrect results, and harms following from (intentional or unintentional) misuse of the technology.
        \item If there are negative societal impacts, the authors could also discuss possible mitigation strategies (e.g., gated release of models, providing defenses in addition to attacks, mechanisms for monitoring misuse, mechanisms to monitor how a system learns from feedback over time, improving the efficiency and accessibility of ML).
    \end{itemize}
    
\item {\bf Safeguards}
    \item[] Question: Does the paper describe safeguards that have been put in place for responsible release of data or models that have a high risk for misuse (e.g., pre-trained language models, image generators, or scraped datasets)?
    \item[] Answer: \answerNA{} %
    \item[] Justification: There are no such risks associated with our work. %
    \item[] Guidelines:
    \begin{itemize}
        \item The answer \answerNA{} means that the paper poses no such risks.
        \item Released models that have a high risk for misuse or dual-use should be released with necessary safeguards to allow for controlled use of the model, for example by requiring that users adhere to usage guidelines or restrictions to access the model or implementing safety filters. 
        \item Datasets that have been scraped from the Internet could pose safety risks. The authors should describe how they avoided releasing unsafe images.
        \item We recognize that providing effective safeguards is challenging, and many papers do not require this, but we encourage authors to take this into account and make a best faith effort.
    \end{itemize}

\item {\bf Licenses for existing assets}
    \item[] Question: Are the creators or original owners of assets (e.g., code, data, models), used in the paper, properly credited and are the license and terms of use explicitly mentioned and properly respected?
    \item[] Answer: \answerNA{} %
    \item[] Justification: We did not use any existing assets. %
    \item[] Guidelines:
    \begin{itemize}
        \item The answer \answerNA{} means that the paper does not use existing assets.
        \item The authors should cite the original paper that produced the code package or dataset.
        \item The authors should state which version of the asset is used and, if possible, include a URL.
        \item The name of the license (e.g., CC-BY 4.0) should be included for each asset.
        \item For scraped data from a particular source (e.g., website), the copyright and terms of service of that source should be provided.
        \item If assets are released, the license, copyright information, and terms of use in the package should be provided. For popular datasets, \url{paperswithcode.com/datasets} has curated licenses for some datasets. Their licensing guide can help determine the license of a dataset.
        \item For existing datasets that are re-packaged, both the original license and the license of the derived asset (if it has changed) should be provided.
        \item If this information is not available online, the authors are encouraged to reach out to the asset's creators.
    \end{itemize}

\item {\bf New assets}
    \item[] Question: Are new assets introduced in the paper well documented and is the documentation provided alongside the assets?
    \item[] Answer: \answerNA{} %
    \item[] Justification: We release no new assets. %
    \item[] Guidelines:
    \begin{itemize}
        \item The answer \answerNA{} means that the paper does not release new assets.
        \item Researchers should communicate the details of the dataset\slash code\slash model as part of their submissions via structured templates. This includes details about training, license, limitations, etc. 
        \item The paper should discuss whether and how consent was obtained from people whose asset is used.
        \item At submission time, remember to anonymize your assets (if applicable). You can either create an anonymized URL or include an anonymized zip file.
    \end{itemize}

\item {\bf Crowdsourcing and research with human subjects}
    \item[] Question: For crowdsourcing experiments and research with human subjects, does the paper include the full text of instructions given to participants and screenshots, if applicable, as well as details about compensation (if any)? 
    \item[] Answer: \answerNA{} %
    \item[] Justification: Our work involves no crowdsourcing or human subjects. %
    \item[] Guidelines:
    \begin{itemize}
        \item The answer \answerNA{} means that the paper does not involve crowdsourcing nor research with human subjects.
        \item Including this information in the supplemental material is fine, but if the main contribution of the paper involves human subjects, then as much detail as possible should be included in the main paper. 
        \item According to the NeurIPS Code of Ethics, workers involved in data collection, curation, or other labor should be paid at least the minimum wage in the country of the data collector. 
    \end{itemize}

\item {\bf Institutional review board (IRB) approvals or equivalent for research with human subjects}
    \item[] Question: Does the paper describe potential risks incurred by study participants, whether such risks were disclosed to the subjects, and whether Institutional Review Board (IRB) approvals (or an equivalent approval/review based on the requirements of your country or institution) were obtained?
    \item[] Answer: \answerNA{} %
    \item[] Justification: Our work involves no crowdsourcing or human subjects. %
    \item[] Guidelines:
    \begin{itemize}
        \item The answer \answerNA{} means that the paper does not involve crowdsourcing nor research with human subjects.
        \item Depending on the country in which research is conducted, IRB approval (or equivalent) may be required for any human subjects research. If you obtained IRB approval, you should clearly state this in the paper. 
        \item We recognize that the procedures for this may vary significantly between institutions and locations, and we expect authors to adhere to the NeurIPS Code of Ethics and the guidelines for their institution. 
        \item For initial submissions, do not include any information that would break anonymity (if applicable), such as the institution conducting the review.
    \end{itemize}

\item {\bf Declaration of LLM usage}
    \item[] Question: Does the paper describe the usage of LLMs if it is an important, original, or non-standard component of the core methods in this research? Note that if the LLM is used only for writing, editing, or formatting purposes and does \emph{not} impact the core methodology, scientific rigor, or originality of the research, declaration is not required.
    \item[] Answer: \answerNA{} %
    \item[] Justification: The core methods in this research did not involve LLMs in their development. %
    \item[] Guidelines:
    \begin{itemize}
        \item The answer \answerNA{} means that the core method development in this research does not involve LLMs as any important, original, or non-standard components.
        \item Please refer to our LLM policy in the NeurIPS handbook for what should or should not be described.
    \end{itemize}

\end{enumerate}
\fi

\end{document}